\documentclass[aps,pra,superscriptaddress,eps2pdf,notitlepage,twocolumn,nofootinbib]{revtex4-1}
\usepackage{tikz-cd}
\usepackage{tikz}
\usepackage{bbold}
\usepackage{amsthm}
\usepackage{epsfig}
\usepackage{amsfonts}
\usepackage{amssymb}
\usepackage{mathrsfs}
\usepackage{amsmath}
\usepackage{times}
\usepackage{color}
\usepackage{soul}
\usepackage{mathtools}
\usepackage{enumerate}
\usepackage{enumitem}
\usetikzlibrary{calc}
\usepackage{ifpdf}
\ifpdf
\usepackage{epstopdf}   
\fi
\usepackage{algorithm}
\usepackage{algpseudocode}

\usepackage{tabularx}
\usepackage{multirow}
\newcolumntype{C}{>{$}c<{$}} 
\newcolumntype{L}{>{$}l<{$}} 
\newcolumntype{R}{>{$}r<{$}}

\makeatletter
\newcommand{\multiline}[1]{%
  \begin{tabularx}{\dimexpr\linewidth-\ALG@thistlm}[t]{@{}X@{}}
    #1
  \end{tabularx}
}
\makeatother
\usepackage[caption=false]{subfig}
\usepackage{hyperref}
\hypersetup{
    colorlinks,
    linkcolor={red!50!black},
    citecolor={green!50!black},
    urlcolor={blue!50!black}
}
\usepackage[capitalize]{cleveref}
\usepackage{pgfplots}
\pgfplotsset{compat=1.14}
\usetikzlibrary{arrows.meta}

\newtheorem{claim}{Claim}

\newtheorem{theorem}{Theorem}
\newtheorem{proposition}{Proposition}
\newtheorem{defin}{Definition}
\newtheorem{lemma}{Lemma}

\DeclareMathOperator{\supp}{supp}
\newcommand{\1}{\mathbb 1}

\newcommand{\ox}[0]{\ensuremath{\otimes}}

\DeclareMathOperator{\im}{\ensuremath{\mathrm{Im}}}
\DeclareMathOperator{\rank}{\ensuremath{\mathrm{rank}}}

\DeclareMathOperator{\tshadow}{\text{-}\mathrm{Shadow}}

\definecolor{grey_pale}{RGB}{157,157,157}
\definecolor{cb-blue}{RGB}{0,119,187}
\definecolor{cb-cyan}{RGB}{51,187,238}
\definecolor{cb-teal}{RGB}{0,153,136}
\definecolor{cb-orange}{RGB}{238,119,51}
\definecolor{cb-red}{RGB}{204,51,17}
\definecolor{cb-magenta}{RGB}{238,51,119}
\definecolor{cb-grey}{RGB}{187,187,187}

\begin{document}
\title{Single-shot error correction of three-dimensional homological product codes}

\author{Armanda O.\ Quintavalle}
\email{armandaoq@gmail.com}
\affiliation{Department of Physics \& Astronomy, University of Sheffield, Sheffield, S3 7RH, United Kingdom}
\author{Michael Vasmer}
\email{mvasmer@perimeterinstitute.ca}
\affiliation{Perimeter Institute for Theoretical Physics, Waterloo, ON N2L 2Y5, Canada}
\affiliation{Institute for Quantum Computing, University of Waterloo, Waterloo, ON N2L 3G1, Canada}
\author{Joschka Roffe}
\email{joschka@roffe.eu}
\affiliation{Department of Physics \& Astronomy, University of Sheffield, Sheffield, S3 7RH, United Kingdom}
\author{Earl T.\ Campbell}
\email{earltcampbell@gmail.com}
\affiliation{Department of Physics \& Astronomy, University of Sheffield, Sheffield, S3 7RH, United Kingdom}
\affiliation{AWS Center for Quantum Computing, Pasadena, CA 91125 USA}

\begin{abstract}
Single-shot error correction corrects data noise using only a single round of noisy measurements on the data qubits, removing the need for intensive measurement repetition. We introduce a general concept of confinement for quantum codes, which roughly stipulates qubit errors cannot grow without triggering more measurement syndromes.  We prove confinement is sufficient for single-shot decoding of adversarial errors and linear confinement is sufficient for single-shot decoding of local stochastic errors. Further to this, we prove that all three-dimensional homological product codes exhibit confinement in their $X$-components and are therefore single-shot for adversarial phase-flip noise. For local stochastic phase-flip noise, we numerically explore these codes and again find evidence of single-shot protection. Our Monte Carlo simulations indicate sustainable thresholds of $3.08(4)\%$ and $2.90(2)\%$ for 3D surface and toric codes respectively, the highest observed single-shot thresholds to date. To demonstrate single-shot error correction beyond the class of topological codes, we also run simulations on a randomly constructed 3D homological product code.

\end{abstract}

\maketitle

\section{Introduction}

Quantum error correction encodes logical quantum information into a codespace \cite{roffe2019quantum}.  Given perfect measurement of the codespace stabilisers we obtain the syndrome of any error present. A suitable decoding algorithm can determine a recovery operation that returns the system to the codespace. Either this recovery is a perfect success, or a failure resulting in a high weight logical error.  However, in real quantum systems the measurements are not perfect and this simple story becomes more involved.  The  three main strategies for tackling noisy measurements are: repeated measurements on the code~\cite{dennis02,FowlerRepMeasure}; performing measurement driven error-correction on a cluster state~\cite{Raussendorf2007,bolt2016foliated,nickerson2018measurement,bombin20182d,brown2019fault,newman2019generating}; or using a single-shot code and decoder~\cite{bombin2015single}.  Focusing on the last strategy, the single-shot approach has the advantage of no additional time cost or cluster-state generation cost and provides a resilience against time-correlated noise~\cite{bombin2016resilience}.  In single-shot error correction, some residual error persists after each round of error correction, but this residual error is kept small and does not rapidly accumulate. However, only a special class of codes support single-shot error correction, but exactly which codes and why is not yet fully understood.

Bomb\'{i}n coined the phrase single-shot error correction and remarked that it ``\textit{is related to self-correction and confinement phenomena in the corresponding quantum Hamiltonian model.}"~\cite{bombin2015single}.  He defined confinement for subsystem codes, and showed that it is sufficient for single-shot error correction with a limited class of subsystem codes. In particular, he proved that the 3D gauge color code supports single-shot error correction, though it is unknown whether the corresponding Hamiltonian exhibits self-correction.  Later single-shot error correction was numerically observed in a variety of higher dimensional topological codes, including: the 3D gauge color code~\cite{brown15}, 4D surface codes~\cite{duivenvoorden2018renormalization} and their hyperbolic cousins~\cite{breuckmann2020single}, and 3D surface codes with phase noise~\cite{Kubica2018a,Kubica2019,vasmer2020cellular}.  Campbell established a general set of sufficient conditions, encapsulated by a code property called good soundness, that ensured adversarial noise~\cite{campbell2019theory} could be suppressed using a single-shot decoder. While Campbell's sufficiency conditions explained single-shot error correction in a wide range of codes, around the same time it was shown that quantum expander codes~\cite{tillich2014quantum,leverrier15,fawzi2018} supported single-shot error correction~\cite{leverrier18}.  However, quantum expander codes lack the soundness property so neither Bomb\'{i}n's notion of confinement or Campbell's notion of soundness is sufficient to encompass all known examples of single-shot error correction.  

We can use different classical algorithms to decode a given quantum code, and this choice will affect the utility of the code. Different decoders have various time complexities and error tolerances, which affects the resources required by a quantum computer based on the code~\cite{fowler12b,terhal2015review,Das2020}. Thus far, single-shot decoders come in two flavours. The first are two-stage decoders~\cite{brown15,duivenvoorden2018renormalization}, where: stage 1 decoding repairs the noisy syndrome using redundancy in the parity check measurements; stage 2 decoding solves the corrected syndrome problem. The second flavour of decoders compute a correction from the noisy syndrome without attempting to repair it. Most examples of such decoders are local decoders, meaning that the whole correction is made up of corrections computed in small local regions of the code using syndrome information in the immediate neighbourhood~\cite{Kubica2018a,Kubica2019,vasmer2020cellular,fawzi2018,grospellier2018numerical,grospellier2020combining,breuckmann2020single}. However, there are some examples of non-local decoders such as belief propagation (BP) being used for single-shot error correction without syndrome repair~\cite{breuckmann2020single,grospellier2020combining}.  A natural question to ask is: what is the optimal decoding strategy for single-shot codes? Even in the simple case of the 3D toric code this is not well-understood. 

The remainder of this article is structured as follows. In Sec.~\ref{sec:summary}, we give a summary of our results. In Sec.~\ref{sec:formal}, we formally state our results on confinement and single-shot decoding. In Sec.~\ref{sec:3d_product}, we detail the construction of 3D product codes.  In Sec.~\ref{sec:numerics}, we present our numerical simulations and analyse their results. Finally, in Sec.~\ref{sec:conc}, we discuss future research directions that flow from this work.

\section{Summary of results \label{sec:summary}}
This article is in two parts: on the one hand, we propose the concept of confinement as an essential characteristic for a code family to display single-shot properties; on the other, we investigate the single-shot performances of the class of 3D homological product codes~\cite{tillich2014quantum, bravyi14, audoux2015tensor}, which we call 3D product codes.
We introduce confinement in Sec.~\ref{sec:formal}. Loosely, confinement stipulates that low-weight qubit errors will result in low-weight syndromes. We then formalise the notion of a code family having good confinement, which we prove is a sufficient condition for single-shot decoding in the adversarial noise setting. In addition to that, we prove that good \emph{linear} confinement is a sufficient condition for a family of codes to exhibit a sustainable single-shot threshold for local stochastic noise (App.~\ref{app:threshold}). 
We review the construction of the 3D product codes in App.~\ref{sec:3d_product}, and show that the 3D surface and toric codes are particular instances of this more general class of codes in Sec.~\ref{app:3dspace}. We prove that all 3D product codes have (cubic) confinement for phase-flip errors (App.~\ref{app:3dconfinement}), and therefore have single-shot error correction for adversarial phase-flip noise. We expect these codes to have single-shot error correction for local stochastic phase-flip noise as well.
In fact, our definition of confinement generalises the definition proposed by Bomb\'{i}n \cite{ bombin2015single} for the gauge color code and the notion of robustness for expander codes \cite{leverrier15}; since both class of codes are proven to have a single-shot threshold for local stochastic noise \cite{bombin2015single, leverrier18} we conjecture that low density parity check (LDPC) codes with good (super linear) confinement have a threshold too. We investigate this case numerically. 

In the single-shot setting, the code always has some residual error present and the error correction procedure introduces noise correlations in subsequent rounds of single-shot error correction.  How then do we assess success or failure of a decoding algorithm? The concept of the sustainable threshold was proposed by Brown, Nickerson and Browne~\cite{brown15} as a metric for single-shot codes and decoders.  We use $p_{\mathrm{th}}(N)$ to denote the threshold of a code-decoder family given $N$ cycles of qubit noise, noisy syndrome extraction and single-shot decoding, with the $N^{\mathrm{th}}$ cycle followed by a single round of noiseless syndrome extraction and decoding.  The final round ensures that we can return the system to the codespace and assess success by the absence of a logical error. We define the sustainable threshold of the code-decoder family to be 
\begin{equation}
    p_{\mathrm{sus}}=\lim_{N\rightarrow\infty}p_{\mathrm{th}}(N).
\end{equation}
Numerically, this is estimated by plotting $p_{\mathrm{th}}(N)$ against $N$ and fitting to the following ansatz,
\begin{equation}
    p_{\mathrm{th}}(N) = p_{\mathrm{sus}}
    [1 - (1 - p_{\mathrm{th}}(0)/p_{\mathrm{sus}})
    e^{-\gamma N}].
    \label{eq:psus_fit}
\end{equation}

\begin{figure}
    \centering
    \includegraphics[width=\columnwidth]{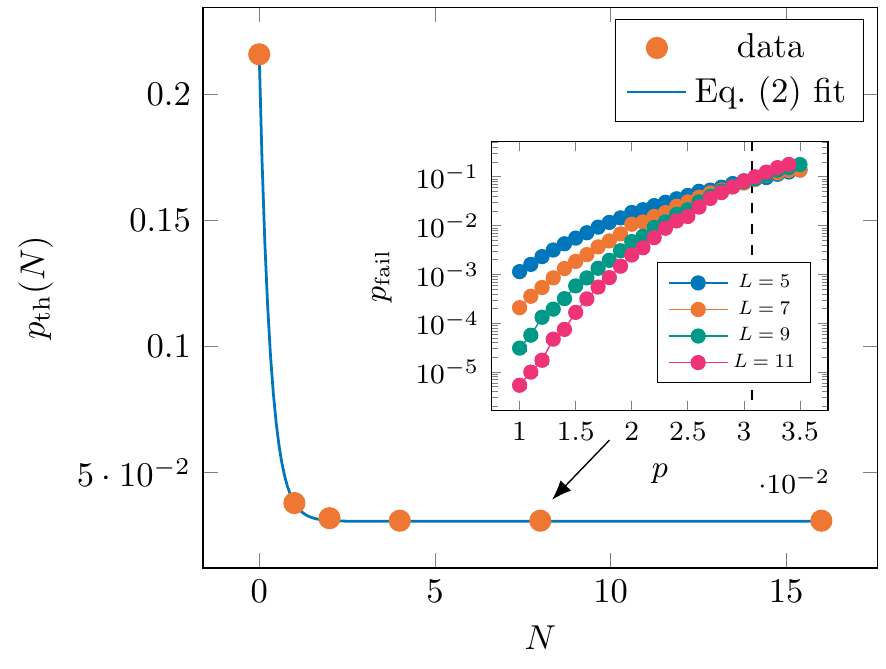}
    \caption{Numerical estimate of the sustainable threshold of the 3D surface code for a two-stage decoder where we repair the syndrome using MWPM, and solve the corrected syndrome problem using BP-OSD. We plot the error threshold $p_{\mathrm{th}}(N)$ for different numbers of cycles, $N$. Using the ansatz in \cref{eq:psus_fit}, we estimate the sustainable threshold to be $p_{\mathrm{sus}}=0.0308(4)$ with $\gamma=3.23$. The inset shows a plot of the logical error rate, $p_{\mathrm{fail}}$, against the phase-flip and measurement error rate, $p$, for $N=8$. The error threshold $p_{\mathrm{th}}(8)$ is the point at which the curves intersect ($L$ is the code distance).}
    \label{fig:3DSC_psus}
\end{figure}

We numerically estimate the sustainable error thresholds of 3D toric and surface codes for two different two-stage decoders. We surpass all previous single-shot error thresholds for these code families, and we also obtain the highest code-capacity noise (no measurement error) threshold; see \cref{tab:pth_compare}. For our single-shot simulations, we use an independent and identically distributed (iid) noise model where each qubit experiences a phase-flip error with probability $p$, and each stabiliser measurement outcome is flipped with probability $q=p$. We investigate two decoding strategies: one where we use minimum-weight perfect matching (MWPM) for stage 1 decoding and belief propagation with ordered statistics decoding (BP+OSD) for stage 2 decoding and another where we use BP+OSD for both decoding stages. Fig.~\ref{fig:3DSC_psus} shows the 3D surface code sustainable threshold fit, using the MWPM \& BP+OSD decoding strategy. We find a comparable sustainable threshold for the 3D surface code using BP+OSD for both decoding stages, as shown in \cref{tab:3D-PC-results}. We achieve very similar performance for the 3D toric code, although there is a subtlety present in stage 1 decoding that is not present in the 3D surface code case; see Sec.~\ref{subsec:2stage}. For both code families, we provide evidence that the performance of stage 1 decoding is the bottleneck of the full decoding procedure, and we achieve near-optimal performance within this constraint. 

\begin{table}[ht!]
    \centering
    \begin{ruledtabular}
	\begin{tabular}{lll}
        Decoder & $q=0$ & $q=p$ \\
        \hline
		Erasure Mapping~\cite{Aloshious2019} & $12.2\%$ & N/A \\
		Toom's Rule~\cite{Kubica2018a} & $14.5\%$ & N/A \\
		Sweep*~\cite{vasmer2020cellular} & $15.5\%$ & $1.7\%$ \\
		Renormalization Group~\cite{duivenvoorden2018renormalization} & $17.2\%$ & 7.3(1)\% \\
		Neural Network~\cite{Breuckmann2018} & 17.5\% & 7.1(3)\% \\
		\textbf{MWPM + BP/OSD*} & $\mathbf{21.55(1)\%}$ & $\mathbf{2.90(2)\%}$ \\
		Optimal~\cite{Ozeki1998,Ohno2004,Hasenbusch2007,takeda2004,Kubica2018} & 23.180(4)\% & 11.0\% 
    \end{tabular}
    \end{ruledtabular}
	\caption{Comparison of the error thresholds of toric code decoders against phase-flip noise with ($q=p$) and without ($q=0$) measurement errors. Starred entries are single-shot decoders.  Unstarred entries use the whole syndrome history.  Text in bold highlights our results.}
	\label{tab:pth_compare}
\end{table}

\begin{table}[ht!]
    \centering
    \begin{ruledtabular}
	\begin{tabular}{lll}
        Code & MWPM \& BP+OSD & BP+OSD x2 \\
        \hline
		Surface & 3.08(4)\% & 2.90(1)\%  \\
	    Toric & 2.90(2)\% & 2.78(2)\%
    \end{tabular}
    \end{ruledtabular}
	\caption{Sustainable thresholds for 3D toric and surface codes for different single-shot decoding strategies. For each entry in the table, we did an analogous simulation to that described in Fig.~\ref{fig:3DSC_psus}. The numbers in brackets are the standard errors.}
	\label{tab:3D-PC-results}
\end{table}

The advantage of using BP+OSD for stage 1 decoding is that, unlike MWPM, this decoder does not rely on the special structure of the loop-like syndrome present in 3D toric and surface codes. Therefore, we anticipate that one can use BP+OSD for single-shot decoding of general 3D product codes. We numerically test this prediction by decoding a family of non-topological 3D product codes using BP+OSD for both decoding stages, achieving sustainable thresholds that are comparable to those of the 3D toric and surface codes. This provides evidence that BP+OSD can be used as a generic two-stage decoder for single-shot LDPC 3D product codes. 

\section{Formal statements \label{sec:formal}}
In this Section we introduce the definition of \emph{confinement} for a stabiliser code and exhibit a theoretical two-stage decoder, the \emph{Shadow decoder}, which we prove is single-shot on confined codes against adversarial noise. We refer the reader to App.~\ref{app:threshold} to see how a variant of the Shadow decoder can be used to prove that good families of codes with linear confinement have a single-shot threshold for local stochastic noise.

A stabiliser code encoding $k$ logical qubits into $n$ physical qubits can be described by its stabiliser group $\mathcal S$ and a syndrome map $\sigma(\cdot )$. The stabiliser group $\mathcal S$ is an Abelian subgroup of the Pauli group $\mathcal P_n$ on $n$ qubits which does not contain $- \mathbb 1$ and has dimension $n-k$. The syndrome map is not unique: any generating set of the group $\mathcal {S}$ defines a valid syndrome map for the code.  If $\{s_1, \dots, s_m\}$ is one of such generating sets, the associate function $\sigma(\cdot)$ maps a qubit operator $p \in \mathcal{P}_n$ into the binary vector $(\bar s_1, \dots, \bar s_m)^T \in \mathbb F_2^m$, where $\bar s_i = 1$ if $s_i$ anti-commutes with $p$ and $0$ otherwise. Importantly, $\sigma(\cdot)$ is linear, meaning that $\sigma(p \cdot q) = \sigma(p) + \sigma(q)$ over $\mathbb F_2^m$.
Because any Pauli operator $p \in \mathcal P_n$ can be factorised as the product of an $X$ and a $Z$-operator $p_X$ and $p_Z$, we can identify it with a binary vector $\bar p = (\bar p_X, \bar p_Z)^T \in \mathbb F_2^{2n}$, where the $i$th entry of $\bar p_{X}/\bar p_{Z}$ is $1$ if and only if $p_X/p_Z$ acts non-trivially on the $i$th qubit.
Given a Pauli operator $p$, its weight $\lvert p \rvert$ is the number of qubits on which its action is not the identity. Fixed a stabiliser code with syndrome function $\sigma(\cdot)$, the reduced weight of a Pauli operator $p \in \mathcal P_n$ on the physical qubits is
\begin{align*}
|p|^{\mathrm{red}} \coloneqq \min\{|p\cdot q|:\, &\sigma(p\cdot q)=\sigma(p),\\ &\text{ for some } q \in \mathcal{P}_n\}.
\end{align*}
A stabiliser code is said to be distance $d$ if $d$ is the minimum weight of a Pauli operator not in $\mathcal S$ that has trivial syndrome. We will refer to a code of length $n$, dimension $k$ and distance $d$ as a $[[n, k, d]]$ code.

For a stabiliser code, we then have
\begin{defin}[Confinement]
Let $t$ be an integer and $f: \mathbb Z \rightarrow \mathbb Z$ some increasing function with $f(0)=0$. We say that a stabiliser code has $(t, f)$-confinement if, for all errors $e$ with $\lvert e\rvert^{\mathrm{red}}\le t$, it holds
\begin{align*}
    f(\lvert \sigma(e)\rvert)  \ge \lvert e \rvert^{\mathrm{red}}.
\end{align*}
\end{defin}
Let us contrast this with Bomb\'{i}n's notion of confinement (Def.~16 of Ref.~\cite{bombin2015single}) that has some similarities but only allows for linear functions of the form $f(x)=\kappa x$ for some constant $\kappa$.  Many codes, including 3D product codes, have superlinear confinement functions, as such Bomb\'{i}n's definition does not encompass them. Moreover, the concept of confinement is closely related to soundness \cite{campbell2019theory} but it is weaker and so able to encompass more families of codes, such as the expander codes~\cite{tillich2014quantum,leverrier15,fawzi2018} which are confined but not sound. Roughly speaking, a code has good confinement if small qubit errors produce small measurement
syndromes; this differs from good soundness which entails that small syndromes can be produced by small errors.

Formally, we define the following notion of good confinement for a family of stabiliser codes
\begin{defin}[Good confinement]\label{def:good_confinement}
Consider an infinite family of stabiliser codes. We say that the family has good confinement if each code in it has $(t, f)$-confinement where:
\begin{enumerate}
\item $t$ grows with the length $n$ of the code: $t\in \Omega(n^{b})$ with $b> 0$;
\item and $f(\cdot)$ is monotonically increasing and independent of $n$.
\end{enumerate}
We say the code family has good $X$-confinement if the above holds only for Pauli-$Z$ errors.
\end{defin}
Our main analytic result is that codes with good confinement are single-shot
\begin{theorem} \label{thm:main}
Consider a family of $[[n,k,d]]$ quantum-LDPC codes with good confinement such that  $d \geq a n^b$ with $a>0$ and $b>0$.  This code family is single-shot for the adversarial noise model. If the code family only has good $X$-confinement then it is single-shot with respect to Pauli-$Z$ noise.
\end{theorem}
We conjecture that  the result of Thm.~\ref{thm:main} can be extended to deal with local stochastic noise and used to show that LDPC codes with good confinement have a single-shot threshold. In this direction, we are able to prove that \emph{linear} confinement is sufficient for codes to exhibit a single-shot threshold in the local stochastic noise setting:
\begin{theorem}\label{thm:threshold}
Consider a family of $[[n,k,d]]$ quantum-LDPC codes with qubit degree at most $\omega -1$ and good linear confinement such that $d \ge an^b$ with $a>0$ and $b>0$.  This code family has a sustainable single-shot threshold for any local stochastic noise model. If the code family only has good $X$-confinement then it has a sustainable single-shot threshold with respect to Pauli-$Z$ noise.
\end{theorem}
We further prove that 3D product codes have $X$-confinement:
\begin{theorem}\label{thm:3d}
All 3D product codes have $(t, f)$ $X$-confinement where $t$ is equal to the $Z$-distance of the code and $f(x) = x^3/2$ or better.
\end{theorem}
Thm.~\ref{thm:main} and Thm.~\ref{thm:3d} together motivate our numerical experiments reported in Sec.~\ref{sec:numerics}. 

We now proceed to prove Thm.~\ref{thm:main}. To this end, we use the Shadow decoder that we introduce in Def.~\ref{def:ball_decoder}. The Shadow decoder differs from previous single-shot two stage decoders (e.g. the MW single-shot decoder introduced in Def.~6 of \cite{campbell2019theory}) in that it does not rely on metachecks on syndromes. If syndromes are protected by a classical code, as it is the case for $X$-syndromes of 3D product codes introduced in Sec.~\ref{sec:3d_product}, then a single-shot decoding strategy could work as follows: (1) correct the measured syndrome whenever it does not satisfy all the constraints defined by the metacode; (2) find a recovery operator on qubits that has syndrome equal to the one found at point (1). The Shadow decoder, instead, corrects the syndrome both anytime it fails to satisfy all the constraints of the metacode and when it is generated by high weight errors. We do not describe how to implement it or make statements concerning the complexity of decoding. Our proof makes similar assumptions as the Kovalev-Pryadko quantum-LDPC threshold theorem~\cite{Kovalevbadcode} where they assumed a minimum weight decoder without addressing implementation issues.  Indeed, decoding for arbitrary LDPC codes is an NP-complete problem that we do not expect to be efficiently solvable in full generality.

The building blocks of the Shadow decoder are the $t\tshadow$s of the code. A $t\tshadow$ is a set in the syndrome space which contains all the images of Pauli errors $e$ on the physical qubits that have weight at most $t$. In other words, if we identify Pauli errors $e$ on $n$ qubits with $2n$-bit strings and we consider the metric space $\mathcal{M}=\mathbb F_2^{2n}$ endowed with the Hamming distance (i.e. the distance $d(\bar e_1, \bar e_2)$ between the vectors $\bar e_1$ and $\bar e_2$ corresponding to the Pauli errors $e_1$ and $e_2$ respectively is defined as $d(\bar e_1, \bar e_2) = |e_1 + e_2|$) then the $t\tshadow$ of the code is the image, via the syndrome function $\sigma(\cdot)$, of the ball of radius $t$ centered at $0$ in $\mathcal{M}$. Note that, because balls on $\mathcal{M}$ are not vector spaces, the $\mathrm{Shadows}$ are not vector spaces either.

We are now ready to introduce the Shadow decoder.
\begin{defin}[Shadow decoder]
\label{def:ball_decoder}
The Shadow decoder has variable parameter $t>0$. Given an observed syndrome $\bar s = \sigma(e) + \bar s_e$ where $\bar s_e \in \mathbb F_2^{m}$ is the syndrome error, the Shadow decoder of parameter $t$ performs the following 2 steps:
\begin{enumerate}
\item Syndrome repair: find a binary vector $\bar s_r$ of minimum weight $\lvert \bar s_r\rvert$ such that $\bar s + \bar s_r$ belongs to the $t\tshadow$ of the code, where
\begin{equation*}
t\tshadow = \{\sigma(e) : \lvert e \rvert  \le t \}.
\end{equation*}
\item Qubit decode: find $e_r$ of minimum weight $\lvert e_r \rvert$ such that $\sigma(e_r)= \bar s + \bar s_r$.
\end{enumerate}
We call $r =  e_r \cdot  e$ the residual error.
\end{defin}
A key result in proving Thm.~\ref{thm:main} is the following promise on the performance of the Shadow decoder: when a code has confinement, the weight of the residual error after one decoding cycle is bounded by a function of the weight of the syndrome error.
\begin{lemma}\label{lem:ball_decoder}
Consider a stabiliser code that has $(t, f)$-confinement. Provided that the original error pattern $e$ has $\lvert e\rvert^{\mathrm{red}} \le t/2$, on input of the observed syndrome $\bar s = \sigma(e) + \bar s_e$, the residual error $r$ left by the Shadow decoder of parameter $t/2$ satisfies:
\begin{align}
  \lvert r \rvert^{\mathrm{red}} \leq f(2\lvert \bar s_e\rvert)  .  
\end{align}
\end{lemma}
\begin{proof}
Assume $\lvert e\rvert^{\mathrm{red}} \le t/2$. By construction, $e_r$ has minimum weight among all errors with syndrome $\sigma(e) + \bar s_e + \bar s_r \in t/2\tshadow$ of the code. In particular $\lvert e_r\rvert \le t/2$. 
By the triangular inequality for the weight function,
\begin{align}
  \lvert e_r \cdot e \rvert^{\mathrm{red}}  \le \lvert e_r \rvert^{\mathrm{red}}  + \lvert e \rvert^{\mathrm{red}}\le t.
\end{align}
Therefore, we can apply the confinement property on the residual error $r = e_r \cdot e$:
\begin{align}\label{eq:proof_eb}
f\big(\lvert \sigma(e_r \cdot e)\rvert  \big) \ge \lvert e_r \cdot e \rvert^{\mathrm{red}}.
 \end{align}
By linearity of the syndrome function $\sigma (\cdot)$:
\begin{align}\label{eq:proof_sy}
\sigma(e_r \cdot e) = \sigma(e_r) + \sigma (e) = \bar s_e + \bar s_r.
\end{align}
Note that the syndrome error $\bar s_e$ is a possible solution of the syndrome repair step of the Shadow decoder, because by assumption $\lvert e\rvert^{\mathrm{red}} \le t/2$. Thus, $\lvert \bar s_r\rvert \le \lvert \bar s_e\rvert $ and
\begin{align}\label{eq:proof_weight}
  \lvert \bar s_e + \bar s_r\rvert \le \lvert \bar s_e \rvert + \lvert \bar s_r \rvert \le 2 \lvert \bar s_e\rvert.
  \end{align}
Combining these and the monotonicity of $f$ leads to the required bound on the residual error:
\begin{align}
  \lvert e_r \cdot e \rvert^{\mathrm{red}} \le f\big(2\lvert \bar s_e\rvert\big).  
\end{align}
\end{proof}
Thm.~\ref{thm:main} follows directly from Lem.~\ref{lem:ball_decoder}. In particular, Lem.~\ref{lem:ball_decoder} entails that a code with $(t, f)$-confinement is robust against $N$ cycles of qubit noise, noisy syndrome extraction and single-shot decoding, as explained below.

At each cycle $\tau$, we assume that a new error $e^{\tau}$ is introduced in the system and it is added to the residual error $r^{\tau-1}$. We assume that for the new physical error $e^{\tau}$ and the syndrome measurement error $\bar s^{\tau}_e$ the following hold:
\begin{equation}\label{eq:tau_bound}
\lvert e^{\tau}\rvert^{\mathrm{red}} \le t/4 \quad \text{and}\quad f(2 \lvert \bar s_e^{\tau}\rvert )\le t/4.
\end{equation}
We perform syndrome extraction on the state $\hat e^{\tau}= r^{\tau -1}+ e^{\tau}$. The noisy syndrome $\bar s^{\tau} = \sigma(\hat e^{\tau})+ \bar s^{\tau}_e$ is used as input for the Shadow decoder of parameter $t/2$. The recovery operator $e_r^{\tau}$ found by the Shadow decoder is then applied to the state and finally a new cycle starts where $r^{\tau} = \hat e^{\tau} + e^{\tau}_r$. Let $r^0 = 0$, so that the initial state of the system is given by $\hat e^1 = e^1$, $\bar s^1 = \sigma(\hat e^1) + \bar s_e^1$. Note that if
\begin{equation}
    \lvert e^{\tau}\rvert^{\mathrm{red}}, \lvert r^{\tau-1}\rvert^{\mathrm{red}} \le t/4
\end{equation}
then $\lvert \hat e^{\tau} \rvert^{\mathrm{red}}= \lvert e^\tau \cdot r^{\tau-1}\rvert^{\mathrm{red}}\le t/2$ and the hypotheses of Lem.~\ref{lem:ball_decoder} hold. Combining this with the bound on the syndrome error \eqref{eq:tau_bound}, we obtain
\begin{align*}
\lvert r^{\tau}\rvert^{\mathrm{red}} \le f(2 \lvert \bar s^{\tau}_e\rvert) \le \frac{t}{4}.
\end{align*}
In conclusion, provided that the conditions on the physical and the measurement error \eqref{eq:tau_bound} are satisfied for each iteration up to $\tau-1$, the residual error after the $\tau^{\mathrm{th}}$ cycle is kept under control too.

Thm.~\ref{thm:threshold} is proven in App.~\ref{app:threshold}. There, we introduce a novel notion of weight to describe local stochastic errors: the \emph{closeness} weight. We then present the Stochastic Shadow decoder, a variant of the (Adversarial) Shadow decoder of Def.~\ref{def:ball_decoder}. Importantly, on confined codes, it keeps the the closeness weight of the residual error under control over repeated correction cycles. Finally, the proof of Thm.~\ref{thm:threshold} follows by combining these results with some classic percolation theory bounds.

The proof of Thm.~\ref{thm:3d} is very technical and is deferred to App.~\ref{app:3dconfinement}. It is an adaption of the one of soundness for 4D codes given in \cite{campbell2019theory}, and it is reported in this manuscript for completeness. We remind the reader that, for our numerical studies on 3D product codes, we do not use the Shadow decoder, but rather heuristics that perform well in practice. In particular, we use a two-stage decoder that exploits a metacheck structure on syndromes and attempts to repair the syndrome if and only if it does not pass all metachecks.

Our main motivation to introduce the concept of confinement and the Shadow decoder was to find a feature of codes able to encompass all known examples of single-shot codes.
Campbell \cite{campbell2019theory} introduced the notion of soundness and showed that this property is a sufficient condition for codes to show single-shot properties in the adversarial setting. Nonetheless, Fawzi et al. \cite{leverrier18} showed that expander codes have a single-shot threshold for local stochastic noise, even though they do not have the soundness property. As already said though, expander codes do have confinement. In Corollary 9 of \cite{leverrier15} the authors prove that their confinement function is linear and call this property \textit{robustness}. Confinement, in other words, fills the gap leaved by the concept of soundness. Furthermore, as Lem.~\ref{lemma:soundness} states, it is a requirement strictly weaker than soundness: any LDPC family of codes with good soundness has good confinement.
\begin{defin}[Soundness \cite{campbell2019theory}]
Let $t$ be an integer and $f: \mathbb{Z}\rightarrow \mathbb{R}$ be a function with $f(0)=0$. Given a stabiliser code with syndrome map $\sigma(\cdot)$ we say it is $(t, f)$-sound if for all error sets $e$ with $|\sigma(e)|\le t$ it follows that $f(|\sigma(e)|)\ge|e|^{\mathrm{red}}$.
\end{defin}
\begin{defin}[Good soundness \cite{campbell2019theory}]
Consider an infinite family of codes with syndrome maps $\sigma_n(\cdot)$. We say that the family has good soundness if each code in it is $(t, f)$-sound where:
\begin{enumerate}
    \item $t$ grows with $n$ such that $t \in \Omega(n^b)$ with $b > 0$;
    \item and $f(\cdot)$ is monotonically increasing and independent of $n$.
\end{enumerate}
\end{defin}
It follows easily from Campbell's definition of soundness and our definition of confinement that the former entails the latter.
\begin{lemma} \label{lemma:soundness}
Consider a LDPC code that is $(t, f)$-sound with $f$ increasing. If its qubit degree is at most $\omega$, then it has $(\frac{t}{\omega}, f)$-confinement.
\end{lemma}
\begin{proof}
If $e$ is an error set with $|e|^{\mathrm{red}}\le \frac{t}{\omega}$, for its syndrome the following holds:
\begin{align}
 |\sigma(e)|\le \frac{t}{\omega} \cdot \omega = t.   
\end{align}
By soundness of the code:
\begin{align}
    f(|\sigma(e)|) \ge |e|^{\mathrm{red}}.
\end{align}
\end{proof}
In conclusion, confinement does answer to the need of finding general and inclusive properties related to single-shot error correction. The reminder of this article is devoted to the study of the 3D product codes. We recall their construction in Sec.~\ref{sec:3d_product} and we numerically assess their single-shot performance under local stochastic noise in Sec.~\ref{sec:numerics}. 
\section{Code construction}
\label{sec:3d_product}
The identification of pauli operators $p \in \mathcal P_n$ with binary vectors $\bar p = (\bar p_X, \bar p_Z) \in \mathbb F_2^{2n}$ is a group homomorphism (i.e.~ multiplication of Pauli operators corresponds to the sum of their vector representation in $\mathbb F_2^{2n}$) and because $\sigma(\cdot)$ is linear, syndrome measurement can be simulated via a matrix-vector product:
\begin{align*}
    \sigma : \mathbb F_2^{2n} &\longrightarrow \mathbb F_2^{m}\\
    \begin{pmatrix}
    \bar p_X\\
    \bar p_Z
    \end{pmatrix}
    &\mapsto
     H  \begin{pmatrix}
    \bar p_X\\
    \bar p_Z
    \end{pmatrix},
\end{align*} 
where the vector $(\bar p_X, \bar p_Z)^T \in \mathbb F_2^{2n}$ represents a Pauli error on the physical qubits.
Following the nomenclature from classical coding theory, we refer to the syndrome matrix $H$ as parity check matrix and we say that a code is LDPC when its parity check is low density.

A stabiliser code is a CSS code if its stabiliser group can be generated by the disjoint union of a set of $X$-operators and a set of $Z$-operators. In this case, its parity check is a block matrix:
\begin{align}
\label{eq:css}
H = \left(\begin{array}{c|c}
    0 & H_X\\
   \hline
    H_Z & 0 
\end{array}\right),
\end{align}
where $H_X$ has size $m_x \times n$ and $H_Z$ has size $m_z \times n$ if the generating set of $X$-stabilisers/$Z$-stabilisers has cardinality $m_x$/$m_z$. Eq.~\eqref{eq:css} entails that syndrome extraction can be performed separately for the $X$-component and for the $Z$-component. In fact, if a Pauli operator has vector representation $(\bar p_X, \bar p_Z)^T = (\bar p_X, 0)^T + (0, \bar p_Z)^T \in \mathbb F_2^{2n}$, then for its syndrome holds:
\begin{align*}
    H \begin{pmatrix}
    \bar p_X\\
    \bar p_Z
    \end{pmatrix}&= H\begin{pmatrix}
    \bar p_X\\
    0
    \end{pmatrix} + H \begin{pmatrix}
    0\\
    \bar p_Z
    \end{pmatrix}\\
    &=\begin{pmatrix}
    0\\
    H_Z \bar p_X
    \end{pmatrix}
    + 
    \begin{pmatrix}
    H_X \bar p_Z\\
    0
    \end{pmatrix}\\
    &=\begin{pmatrix}
    0\\
    \bar s_Z
    \end{pmatrix} + \begin{pmatrix}
    \bar s_X\\
    0
    \end{pmatrix}
\end{align*}
where $\bar s_Z \in \mathbb{F}_2^{m_z}$ and $\bar s_X \in \mathbb{F}_2^{m_x}$. In other words, it is possible to truncate these vectors without loosing information and deal with $X$ and $Z$ operators separately. For this reason, we say that a CSS code is provided with two syndrome maps which correspond to the two blocks/matrices $H_X$ and $H_Z$ respectively. Accordingly, a CSS code will have a $Z$-distance and a $X$-distance and can be compactly be refereed to as a $[[n, k, d_x, d_z]]$ code. 

To our purpose, it is handy to describe CSS codes in terms of chain complexes. A length $\ell$ chain complex is a collection of vector spaces $C^0, \dots, C^\ell$ and linear maps $\delta_i : C^i \rightarrow C^{i + 1}$ with the only constraint
\begin{equation}
\label{eq:chain_complex}
\delta_{i+1}\delta_i = 0,
\end{equation}
for $i = 0, \dots, \ell-1$. If $H_X$ and $H_Z$ are $m_x \times n$ and $m_z \times n$ binary matrices as in Eq.~\eqref{eq:css}, then the chain complex:
\begin{equation}
    \mathbb{F}_2^{m_z} \xrightarrow[]{H_Z^T} \mathbb{F}_2^n \xrightarrow[]{H_X} \mathbb{F}_2^{m_x}
\end{equation}
is well defined. In fact, the commutative condition on the stabilisers of the code entails that the $X$-generators and the $Z$-generators of the code have even overlap, such that they are orthogonal when seen as binary vectors. In other words, $H_X \cdot H_Z^T = 0$ is equivalent to the defining property of chain complexes (see Eq.~\eqref{eq:chain_complex}). In general, we can associate a CSS code to any chain complex $\{C^i, \delta_i\}$ of length at least $3$ by equating $H_Z^T = \delta_i$ and $H_X= \delta_{i + 1}$ for some index $i$. In the chain complex language, we say that the code has length $\dim C_{i + 1}$ and the dimension of the $(i + 1)$-th homology group $\mathcal H_{i + 1} = \ker \delta _{i + 1} / \im \delta_i$. Equivalently, the dimension of the code is the dimension of the $(i + 1)$-th cohomology group $\mathcal{H}^*_{i + 1} = \ker \delta_i ^T / \im \delta_{i+1}^T$. The $X$ and $Z$ distances are given respectively by the minimum weight of any non-zero vector in $\mathcal H^*_{i + 1}$ and $\mathcal H_{i+1}$.

Here, we study some decoding properties of 3D product codes. By this nomenclature we refer to the CSS codes obtained by the homological product of three length-1 chain complexes as described in \cite{pryadko2018hp}. Given three classical linear codes with parity check matrices $\delta_A, \delta_B$ and $\delta_C$ we can build a 3D quantum code as follows. If $\delta_\ell$ is a binary matrix of size $m_\ell \times n_\ell$, it defines a linear map $\delta_\ell : C_\ell^0 \rightarrow C_\ell^1$, where $C_\ell^0$, $C_\ell^1$ are vector spaces over $\mathbb F_2$ of dimension $n_\ell$ and $m_\ell$ respectively; in other words, each linear map defines a length-1 chain complex. The 3D product of the \emph{seed matrices} $\delta_A, \delta_B, \delta_C$, is the length-3 chain complex $\mathcal C$ given by:~\\
\begin{tikzpicture}[transform canvas={scale=0.7}]
\begin{tikzcd}
& C_A^1 \ox C_B^1 \ox C_C^1 \ar[from=2-1]  \ar[from=2-2, swap] \ar[from=2-3, swap]&\\
C^1_A \ox C^1_B \ox C^0_C & C^1_A \ox C^0_B \ox C^1_C & C^0_A \ox C^1_B \ox C^1_C\\
C_A^1 \ox C_B^0 \ox C_C^0  \ar[to=2-1] \ar[to=2-2]& C_A^0 \ox C_B^1 \ox C_C^0 \ar[to=2-1] \ar[to=2-3]& C_A^0 \ox C_B^0 \ox C^1_C \ar[to=2-2] \ar[to=2-3]\\
& C_A^0 \ox C_B^0 \ox C_C^0 \ar[to=3-1]  \ar[to=3-2] \ar[to=3-3, swap]&
\end{tikzcd}\hspace{0.5 cm}
\begin{tikzcd}
\mathcal{C}_3\\
\mathcal{C}_2 \ar[to=1-1, swap, "\delta_{2}"]\\
\mathcal{C}_1 \ar[to=2-1, swap, "\delta_{1}"]\\
\mathcal{C}_{0}\ar[to=3-1, swap, "\delta_{0}"]
\end{tikzcd}
\end{tikzpicture} \noindent
where:
\begin{align*}
\delta_{0}&=\begin{pmatrix}
\delta_A \ox \1 \ox \1 \\ 
\\
\1 \ox \delta_B \ox \1 \\
\\
\1 \ox \1 \ox \delta_C  
\end{pmatrix},  \\ &\\
\delta_1 &= \begin{pmatrix}
\1\ox \delta_B \ox \1 & \delta_A \ox \1 \ox \1 &0 \\
&&\\
\1 \ox \1 \ox \delta_C & 0 &  \delta_A \ox \1 \ox \1 \\
&&\\
0 & \1 \ox \1 \ox \delta_C & \1 \ox \delta_B \ox \1   
\end{pmatrix}, \\ &\\
\delta_2 &= \begin{pmatrix}
\1\ox \1 \ox \delta_C & \1 \ox \delta_B \ox \1&  \delta_A \ox \1 \ox \1
\end{pmatrix}.
\end{align*}
The symbol $\ox$ represents the tensor product. Given two vector spaces $A$ and $B$ over a field $\mathbb F$, their tensor product is the vector space $A \ox B$ generated by the formal sums $\sum a \ox b$  where $a \in A $ and $b \in B$ and the operator $\ox$ is bilinear, i.e. for any $a_1, a_2, b_1, b_2$ in $A$ and $B$ respectively, it holds that
\begin{align*}
    (a_1 + a_2) \ox b_1 &= a_1 \ox b_1 + a_2 \ox b_1,\\
    a_1 \ox (b_1 + b_2) &= a_1 \ox b_1 + a \ox b_2.
\end{align*}
The horizontal stacking of spaces, instead, represents their direct sum.  

It is easy to verify that condition \eqref{eq:chain_complex} is verified for $i = 1,\dots,3$ and therefore the chain complex $\mathcal C$ is well defined. As explained above, we define a CSS code on $\mathcal C$ by equating
\begin{equation*}
    H_Z = \delta_0^T,\quad
    H_X = \delta_1.
\end{equation*} 
We refer to the matrix $M  =\delta_2$ as the metacheck matrix for the $X$-stabilisers. Condition \eqref{eq:chain_complex} entails $M H_X = 0$ and as a consequence we can think of the matrix $M$ as a parity check matrix on the $X$ syndromes: any valid $X$-syndrome satisfies the constraints defined by $M$.

Let $[n_{\ell}, k_{\ell}, \delta_{\ell}]/[n_\ell^T, k_{\ell}^T, d_{\ell}^T]$ be the parameters of the classical linear code with parity check matrix $\delta_\ell/\delta_{\ell}^T$. As showed in \cite{pryadko2018hp}, $\mathcal C$ is thus associated with an $[[n, k, d_x, d_z]]$ code such that, if $k \neq 0$,
\begin{align*}
n &=n_a^T n_b n_c + n_a n_b^T n_c + n_a n_b n_c^T,\\
k&=k_a^T k_b k_c + k_a k_b^T k_c + k_a k_b k_c^T,\\
d_x &= \min\{d_b d_c, d_a d_c, d_a d_b\},\\
d_z &= \min\{d_a^T, d_b^T, d_c^T\}.
\end{align*}
By convention, the the distance of a code with dimension $0$ is $\infty$. 
We define the single-shot distance $d_{ss}$ \cite{campbell2019theory} of the chain complex $\mathcal C$ as the minimum weight of a vector in $\mathcal C_2$ that satisfies all the constraints given by $\delta_2$ (i.e.\ it belongs to the kernel of $\delta_2$) but is not a valid $X$-syndrome (i.e.\ it does not belong to the image of $\delta_1$).  In other words, $d_{ss}$ is the minimum weight of a vector in the second  homology group $\mathcal H _2 = \ker \delta_2 / \im \delta_1 $ of the chain complex $\mathcal C$. Following~\cite{pryadko2018hp} it is easy to verify that $d_{ss}= \min\{d_a, d_b, d_c\}$ if $\mathcal{H}_2\neq 0$ and $\infty$ otherwise.

It is important to note that, if the matrices $\delta_{\ell}$ are LDPC, then their 3D product code is quantum-LDPC. In fact, if $\delta_{\ell}$ has column (row) of weight bounded by $c_\ell$ ($r_\ell$), then $\delta_i$ has column and row weight bounded by $c_i$ and $r_i$ respectively where:
\begin{enumerate}[label=\roman*.]
\item $c_0 \le c_a + c_b + c_c$ and $r_0 \le \max\{r_a, r_b, r_c\};$
\item $\begin{aligned}c_1 \le \max\{c_a+c_b, c_a+c_c, c_b+ c_c\}\end{aligned}$\\ and \\
$r_1 \le \max\{r_a+ r_b, r_a+ r_c, r_b+r_c\}$;
\item $c_2 \le \max\{c_a, c_b, c_c\}$ and $r_2 \le r_a + r_b + r_c $.
\end{enumerate}

\subsection{On geometric locality}
\label{sec:3d_space}
In addition to preserving the LDPC properties of the seed matrices, the 3D product yields local codes when qubits are placed on edges of a 3D cubic lattice. We defer the reader to App.~\ref{app:3dspace} for a thorough discussion on the embedding of 3D product codes on a cubic lattice and we here present a loose summary.

Qubits of a 3D product code associated to the chain complex $\mathcal{C}$ are in bijection with basis elements of the space $\mathcal C_1$; since $\mathcal C_1$ is the direct sum of the three vector spaces $C_A^1\ox C_B^0\ox C_C^0$, $C_A^0\ox C_B^1\ox C_C^0$ and $C_A^0\ox C_B^0\ox C_C^1$ we introduce three different type of qubits: \textit{transverse}, \textit{vertical} and \textit{horizontal}. Qubit types naturally correspond to the three different orientation of edges on a cubic lattice, namely edges parallel to each of the three crystal planes. Referring to this particular display of qubits, the stabilisers of the code defined by $\mathcal C$ have support as follows:
\begin{enumerate}
    \item $X$-stabilisers have support on a 2D cross of qubits of two types out of three, contained in one of the three crystal planes; the crossing is defined by a square face of a cube (see Fig.~\ref{fig:cross2d});
    \item $Z$-stabilisers have support on a 3D cross of qubits, with crossing defined by a vertex of a cube (see Fig.~\ref{fig:cross3d}).
\end{enumerate}
The cubic lattice considered can present some irregularities: in general it is a cubic lattice with some missing edges. Nonetheless, square faces and vertices are uniquely defined and they correspond to a stabiliser every time they contain at least one edge. 
More specifically, a square face identifies two perpendicular lines of edges/qubits on a plane which are the edges parallel to the boundary of the square face itself. The corresponding $X$-stabiliser has support contained on those lines of edges/qubits. Similarly, a vertex identifies three perpendicular lines of qubits and the corresponding $Z$-stabiliser has support there contained. When combined with some locality properties of the seed matrices,
this characteristic `cross shape' of the stabilisers support entails that 3D product codes are local on a cubic lattice (Prop.~\ref{prop:locality} in App.~\ref{app:3dspace}). Here, by locality, we mean that for some positive integer $\rho$, hold: 
\begin{enumerate}
   \item any $X$-stabiliser generator has weight at most $2\rho$ with support contained in a 2D box of size $\rho \times \rho$,
    \item any $Z$-stabiliser generator has weight at most $3\rho$ with support contained in a 3D box of size $\rho \times \rho \times \rho$,
    \end{enumerate}

Interestingly, it follows easily as a corollary of our locality proof that the 3D toric and surface codes are in fact 3D product codes. We now detail an explicit construction of the 3D toric and surface codes as 3D product codes and we remind the reader to App.~\ref{app:3dspace} for further details.

The 3D toric code is the 3D product code obtained by choosing $\delta_A = \delta_B = \delta_C = \delta$ as seed matrices, where $\delta$ is the parity check matrix of the repetition code. For instance, the 3D toric code with linear lattice size $L = 3$ is given by
\begin{equation*}
    \delta = \begin{pmatrix}
    1 & 1 & 0\\
    0  & 1 & 1\\
    1 & 0 & 1 
    \end{pmatrix}.
\end{equation*}
In general, the 3D toric code of lattice size $L$, has parameters 
\begin{equation*}
    [[3L^3,\, 3,\, L^2,\, L]],
\end{equation*}
and single-shot distance $d_{ss}= L$.

The 3D surface code is obtained from this construction by choosing, for linear lattice size $L = 3$,
\begin{equation*}
\delta_A = \delta_B = \begin{pmatrix}
    1 & 1 & 0\\
    0  & 1 & 1
    \end{pmatrix}
\end{equation*}
and 
\begin{equation*}
    \delta_C = \begin{pmatrix}
    1 & 1 & 0\\
    0  & 1 & 1
    \end{pmatrix}^T
\end{equation*}
Therefore, for lattice size $L$, it has parameters 
\begin{equation*}
    [[2L(L-1)^2+L^3,\, 1,\, L^2,\, L]],
\end{equation*}
and single-shot distance $d_{ss} = \infty$.\\
Further details can be found in App.~\ref{app:3dspace}.

\subsection{Non-topological codes}

The 3D product code construction can be used to obtain non-topological codes with non-local interactions. Table \ref{tab:qldpc3d} shows the parameters for a class of non-topological codes based on classical LDPC codes. The specific advantage of non-topological codes is that it is possible to construct code families where the number logical qubits increases with the code length. This is in contrast to 3D toric and surface codes, where the number of logical qubits is fixed for all values of the code distance.
\begin{table}[ht!]
    \centering
    \begin{ruledtabular}
	\begin{tabular}{llll}
        $\delta_A$ & $\delta_B$ & $\delta_C$ & $\mathcal{HGP}(\delta_A,\delta_B,\delta_C)$\\
        $[16,4,6]$ & $[6,1,6]$ & $[6,0,\infty]$ & $[[1336,4,6]]$\\
        $[20,5,8]$ & $[8,1,8]$ & $[8,0,\infty]$ & $[[3100,5,8]]$\\
        $[24,6,10]$ & $[10,1,10]$ & $[10,0,\infty]$ & $[[5964,6,10]]$
    \end{tabular}
    \end{ruledtabular}
	\caption{A family of 3D product codes. The base codes $\{\delta_A,\delta_B,\delta_C\}$ are set as follows: $\delta_A$ is a parity check matrix of an $[n,k,d]$ $(3,4)$-LDPC code; $\delta_B$ is a $[L,1,L]$ repetition code; $\delta_C$ is the transpose of a $[L,1,L]$ repetition code. The code distance is set to $\infty$ for codes of dimension $0$.}
	\label{tab:qldpc3d}
\end{table}

\section{Numerics \label{sec:numerics}}

To assess the single-shot performance of the 3D product codes, we simulate the decoding of phase-flip ($Z$) errors. As 3D product codes are CSS codes, the relevant stabilizers are the $X$-stabilizers. Let $\bar e_Z \in \mathbb F_2^n$ describe the support of a phase-flip error, i.e.\ $(\bar e_Z)_i = 1$ if qubit $i$ has a phase-flip error. The syndrome, $\bar s_X$, of this error is then 
\begin{equation}
    \bar s_X=H_X\bar e_Z\rm,
\end{equation}
where $H_X \in \mathbb F_2^{m\times n}$ is the parity check matrix of the $X$-stabilizers of the code (see Eq.~\eqref{eq:css}).

Owing to the chain-complex structure of 3D product codes (outlined in Sec.~\ref{sec:3d_product}) the syndromes $\bar s_X$ are themselves the 
codewords of a classical linear code with parity check matrix $M$ such that $M\bar s_X=0 $ for all $\bar s_X \in {\im}(H_X)$. 
We refer to such a code on the syndromes as metacode. The metacheck matrix can be used to detect and correct syndrome noise.

In a two-stage single-shot decoder, stage 1 decoding corrects the syndrome noise using $M$ before stage 2 decoding corrects the data qubits. In general, decoding is an NP-complete problem that cannot be solved exactly in polynomial-time. However, good heuristic techniques exist that allow approximate solutions to be efficiently computed. In this work, we use two such decoding methods: minimum weight perfect matching (MWPM) and belief propagation plus ordered statistics decoding (BP+OSD). Both MWPM and BP+OSD are algorithms that run over graphical models that encapsulate the structure of the code. We now briefly describe each decoder.

\subsection{Minimum weight perfect matching (MWPM)}

The minimum weight perfect matching (MWPM) decoder is useful for codes that produce pairs of syndromes at the ends of error chains. The method works by mapping the decoding problem to a graphical model in which nodes represent the code syndromes and weighted edges represent error chains of different lengths. For a given pair of syndrome excitations, the MWPM algorithm deduces the shortest error chain that could have caused it \cite{Edmonds1965}. 

MWPM finds use for a variety of topological codes, most famously for the 2D surface and toric codes ~\cite{dennis02,Fowler12,breuckmann2016,Brown2020,kubica2019efficient}. For 3D codes, MWPM is a suitable candidate for syndrome repair step referred to as stage 1 decoding. Specifically, the syndrome of a phase-flip error can be viewed as a collection of closed loops of edges in a simple cubic lattice\footnote{Namely, the dual lattice of the one described in App.~\ref{app:3dspace}.} (with boundary conditions depending on the code).
Measurement errors cause loops of syndrome to be broken, and the job of stage 1 decoding is to repair these broken syndromes. To obtain the corresponding matching problem, we create a complete graph whose vertices correspond to the break-points of the broken syndrome loops, with edge weights that are equal to the path lengths between the break-points. We use the Blossom V~\cite{Kolmogorov2009} implementation of Edmonds's algorithm to solve this matching problem. The edges in the matching correspond to the syndrome recovery operators.

\subsection{Belief propagation + ordered statistics decoding (BP+OSD)}

Belief propagation (BP) is an algorithm for performing inference on sparse graphs that finds widespread use in high-performance classical coding. Classical LDPC codes, for example, achieve performance close to the Shannon-limit when decoded with BP \cite{mackay1997near}. In the context of quantum coding, BP is useful for codes that do not produce pairs of syndromes and therefore cannot be decoded with MWPM.

 The BP algorithm maps the decoding problem to a bipartite \textit{factor graph} where the two node species represent data qubits and syndromes respectively. Graph edges are drawn between the data and syndrome nodes according to the code's parity check matrix. The factor graph is designed to provide a factorisation of the probability distribution that describes the relationship between syndromes and errors. The BP algorithm proceeds by iteratively passing `beliefs' between data and syndrome nodes, at each step updating the probability that a data node is errored. The algorithm terminates once the probability distribution implies a error pattern that satisfies the inputted syndrome.    

For quantum codes, the standard BP algorithm alone does not achieve good decoding performance due to the presence of degenerate errors. These cause `split-beliefs' and prevent the algorithm from terminating. Fortunately, the problem of split-beliefs can be resolved by incorporating a post-processing technique known as ordered statistics decoding (OSD). The OSD step uses the probability distribution outputted by BP to select a low-weight recovery operator that satisfies the syndrome equation.

The BP+OSD algorithm was first applied to quantum expander codes by Panteleev and Kalachev \cite{panteleev2019degenerate}. Following this, Roffe et al. \cite{roffe2020decoding} demonstrated that the BP+OSD decoder applies more widely across a broad range of quantum-LDPC codes, including the 2D surface and toric codes. For this work, we use the software implementation of BP+OSD from \cite{roffe2020decoding}, which can be downloaded from \cite{bp_osd_github}.

\subsection{The two-stage single-shot decoding algorithm \label{subsec:2stage}}

Our simulations of the two-stage single shot decoder employ two strategies. (1) MWPM \& BP+OSD: stage 1 decoding is performed using MWPM and stage 2 decoding uses BP+OSD. (2) BP+OSD$\times$2: both stages are BP+OSD.

Algorithm \ref{alg:ssec} describes our methodology for the simulations of the two-stage single-shot decoder. 

\begin{algorithm}[H]
\caption{\textsc{Single-shot error correction}}
\label{alg:ssec}
\begin{algorithmic}[1]
        \Require Decoder 1, decoder 2, number of rounds $N$, error rate $p$, 
        $X$ parity check matrix $H_X$, metacheck matrix $M$, modified metacheck matrix $M'$
        \Ensure Success or failure
        \State $\bar e_Z \gets 0$ \Comment{Qubit error}
        \State $\bar s_X \gets 0$ \Comment{Syndrome}
        \State $\bar m \gets 0$ \Comment{Metasyndrome}
            \For{$j\gets 1$ \textbf{to} $N$}
                \State Generate phase-flip error $\bar e'_Z$ according to error rate $p$
                \State $\bar e_Z \gets \bar e_Z + \bar e'_Z$
                \State $\bar s_X \gets H_X \bar e_Z$
                \State Generate syndrome error $\bar s'_X$ according to error rate $p$
                \State $\bar s_X \gets \bar s_X + \bar s'_X$
                \State $\bar m \gets M \bar s_X$
                \State Use decoder 1 to obtain syndrome recovery $\bar r_M$ such that $M\bar r_M=\bar m$
                \State $\bar s_X \gets \bar s_X + \bar r_M$
                
                \If{$\bar s_X \notin {\rm{Im}}(H_X)$} \Comment{Failure-mode subroutine} \label{alg:line:subroutine}
                    \State $\bar s_X \gets \bar s_X + \bar r_M$
                    \State Use decoder 1 to obtain valid $\bar r_M$ s.t.\ $M'\bar r_M=\bar m$
                    \State $\bar s_X \gets \bar s_X + \bar r_M$
                \EndIf
                
                \State Use decoder 2 to obtain qubit recovery $\bar r_Z$ s.t.\ $H_X\bar r_Z=\bar s_X$ 
                \State $\bar e_Z \gets \bar e_Z + \bar r_Z$
            \EndFor
        \State Generate phase-flip error $\bar e'_Z$ according to error rate $p$
        \State $\bar e_Z \gets \bar e_Z + \bar e'_Z$
        \State $\bar s_X \gets H_X \bar e_Z$
        \State Use decoder 2 to obtain qubit recovery $\bar r_Z$ s.t.\ $H_X\bar r_Z=\bar s_X$ 
        \State $\bar e_Z \gets \bar e_Z + \bar r_Z$ 
        \If{$\bar e_Z$ is a stabiliser}
            \State \Return Success
        \EndIf
        \State \Return Failure
\end{algorithmic}
\end{algorithm}

The 3D toric code has a failure mode that is not present in the 3D surface code. In such codes, syndromes $\bar s_X$ exist that satisfy all of the metachecks, $M \bar s_X=0$, but are \textit{invalid} syndromes, meaning that $\bar s_X$ does not belong to the image of $H_X$. In other words, $\bar s_X$ is invalid if there is no error vector $\bar e_Z \in \mathcal{C}_1$ with syndrome $\bar s_X$ but it is a codeword of the metacode.

Referring to the chain complex structure of $\mathcal{C}$:
\begin{align*}
    \mathcal{C}_0 \xrightarrow[H_Z^T]{\delta_0} \mathcal{C}_1 \xrightarrow[H_X]{\delta_1} \mathcal{C}_2 \xrightarrow[M]{\delta_2} \mathcal{C}_3
\end{align*}
we see that these non-valid syndromes are non-trivial elements of the $2$nd homology group:
\begin{align*}
    \mathcal{H}_2 = \ker \delta_2 / \im \delta_1 = \ker M / \im H_X.
\end{align*}
If $k_m$ is the dimension of $\mathcal{H}_2$, the set of invalid syndromes is a vector subspace of $\mathcal{C}_2$ of dimension $k_m$ whose vectors can be written as $\bar u + H_X \bar e_Z$ where $\bar u$ is a representative of the equivalence class $[\bar u] \in \mathcal H_2$ and $\bar e_Z \in \mathcal C_1$. Thus, if $F_M$ is a matrix whose columns are $k_m$ vectors in $\mathcal C_2$ that generate $\mathcal H_2$ (meaning that they belong to $k_m$ different equivalence classes in $\mathcal H_2$), we can write any invalid syndrome $\bar s_X$ as:
\begin{align}
    \label{eq:invalid_syndrome}
\bar s_X = F_M \bar v + H_X \bar e_Z,
\end{align}
where $\bar v \in \mathbb F_2^{k_m}$ is non-zero if and only if $\bar s_X$ is invalid and $\bar e_Z$ is any error vector in $\mathcal C_1$. 

By duality on $\mathcal C$, the $2$nd cohomolgy group:
\begin{align*}
    \mathcal{H}_2^{*}= \ker \delta_1^T / \im \delta_2^T = \ker H_X^T / \im M^T,
\end{align*}
has order $k_m$ too. If $L_M$ is a matrix whose $k_m$ rows generates $\mathcal H_2^*$, then the product $\Pi = L_M F_M$ has full rank $k_m$ because both $L_M$ and $F_M$ have full rank. Moreover, since the rows of $L_M$ in particular belongs to $\ker H_X^T$, it holds $L_M H_X = 0$. Combining these two observations with Eq.~\eqref{eq:invalid_syndrome} yields:
\begin{align*}
    L_M \bar s_X &= L_M F_M \bar v+ L_MH_X \bar e_Z\\
    &=\Pi \bar v
\end{align*}
where $\Pi \bar v = 0$ if and only if $\bar v =0$ because $\Pi$ is full rank. In conclusion, we have found that:
\begin{align*}
    L_M \bar s_X \neq 0
\end{align*}
if and only if $\bar s_X$ is an invalid syndrome. As a consequence, we can assess whether a syndrome is invalid or not by  calculating this product. The meaning of matrices $L_M$ and $F_M$ can be understood by looking at elements in $\mathcal H_2$ and $\mathcal H_2^*$ as logical operators of a CSS code defined on $\mathcal C$ with qubits in $\mathcal C_2$ (see Sec.~\ref{sec:3d_product}). In this settings, the full rank condition $\rank(L_M F_M) =k_m$ translates in the anticommuting relation between logical $X$ and logical $Z$ operators of the code.

In the 3D toric code, these invalid syndromes are loops of edges around one of the handles of the torus, and are equivalent to the logical operators of the metacode. It is therefore possible to check whether stage 1 decoding results in such a failure by checking whether the repaired syndrome anti-commutes with a matrix $L_M$ whose rows generate the relevant group of the logical operators of the metacode. When a metacode failure is encountered, a failure-mode subroutine (line \ref{alg:line:subroutine} of Algorithm \ref{alg:ssec}) is called that forces the repaired syndrome into the correct form. This sub-routine involves using BP+OSD to decode a modified version of the metacheck matrix $M'$ defined as follows
\begin{equation}
M'=\begin{pmatrix}M\\L_M\end{pmatrix}\rm.
\end{equation}

The additional constraints in the modified metacheck matrix ensure that the repaired syndrome is never an invalid syndrome. We call this procedure as a subroutine (rather than all the time) as the $L_M$ component causes $M'$ to have higher maximum row and column weights than $M$, resulting in a reduction in BP decoding performance.  Indeed, the rows of $L_M$ must have weight lower bounded by the transpose distances of the seed codes\footnote{More precisely, rows of $L_M$ are vectors in $\mathcal C_2$ that correspond to elements of the second cohomology group $\mathcal{H}_2^{*}$; hence their weight is lower bounded by
$d_2^* = \min\{d_a^Td_b^T, d_a^T d_c^T, d_b^Td_c^T\}$, see \cite{pryadko2018hp}.}. Since the transpose distances of the seed codes also determine the $Z$-distance of the quantum code (Sec.~\ref{sec:3d_product}), we want these quantities to be growing with the length $n$ of the code and therefore the matrix $L_M$ is not, in general, LDPC.

We find that whilst the failure-mode subroutine does not change the error threshold of the decoder, it does considerably reduce the logical error rate. This is illustrated by Fig.~\ref{fig:meta-check-fail-recovery}, which shows the single-shot logical error rate with and without the failure-mode subroutine. For large syndrome error rates, Fig.~\ref{fig:meta-check-fail-recovery} shows the failure-mode subroutine improves decoding performance by over an order of magnitude.

\begin{figure}
    \includegraphics[width=\columnwidth, keepaspectratio]{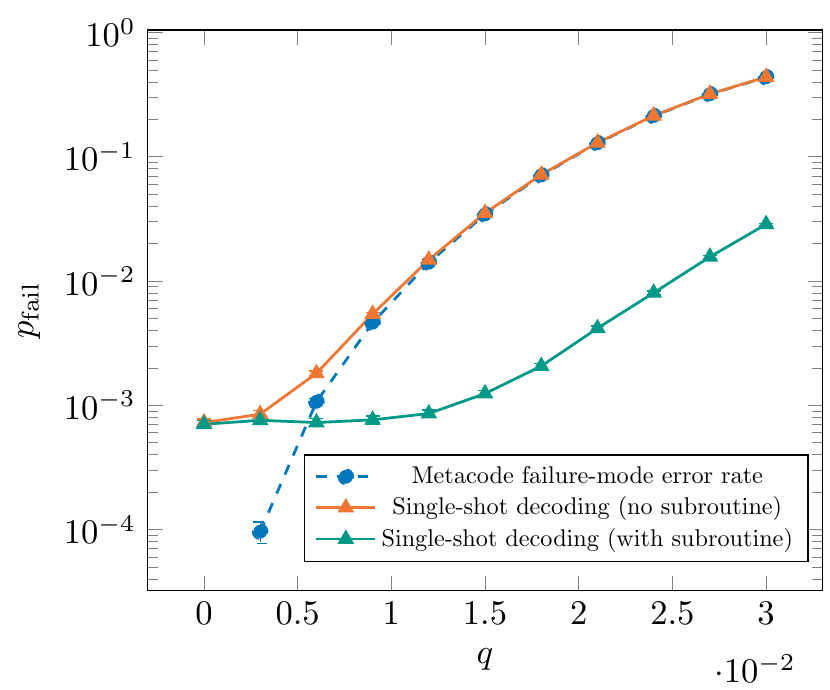}
    \caption{Single-shot decoding
    of the 3D toric code with $L=5$, with and without the metacode failure-mode subroutine. The failure rate $p_{\rm fail}$ is plotted against increasing values of the syndrome error rate $q$, whilst the phase-flip error rate is set to $p=0.1$.  Without the subroutine, the single-shot decoder rapidly converges to the failure-mode error rate (dotted line). For large values of $q$ the subroutine improves the logical failure rate by over an order of magnitude. In this simulation, BP+OSD was used for both stage 1 and 2 decoding.}
    \label{fig:meta-check-fail-recovery}
\end{figure}

\subsection{3D toric and surface codes}

\begin{figure}
    \includegraphics[width=\columnwidth, keepaspectratio]{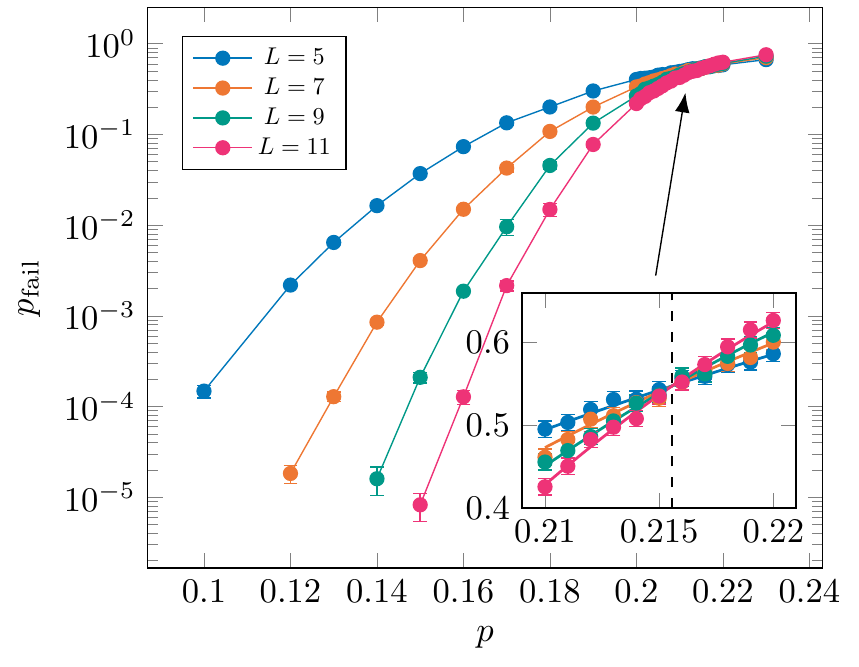}
    \caption{Code capacity threshold of the 3D toric code. We plot the logical error rate $p_{\mathrm{fail}}$ as a function of the phase-flip error rate $p$ for codes with linear lattice size $L$. The inset shows a zoom of the threshold region, where the lines show the threshold fit described in App.~\ref{app:fitting}. All data points have at least 25 failure events. The error bars show the 95\% confidence intervals $p_{\rm fail} = \hat p_{\rm fail} \pm 1.96\sqrt{p_{\rm fail}(1 - p_{\rm fail})/\eta}$, where $\eta \geq 10^4$ is the number of Monte Carlo trials.}
    \label{fig:code-cap}
\end{figure}

We estimate the sustainable threshold of the 3D toric and surface codes using our two decoding strategies. For code-capacity noise (i.e.~perfect syndrome measurements), the syndrome-repair step is not required, so both decoding strategies are the same. For each code family, we observe a code capacity threshold of $p_{\mathrm{th}} \approx 21.6\%$, as illustrated in Fig.~\ref{fig:code-cap}. To obtain our threshold estimates, we use the standard critical exponent method~\cite{Harrington2004} (see App.~\ref{app:fitting} for details). In the single-shot setting, we find similar performance for both our decoding strategies, as summarized in \cref{tab:3D-PC-results}. Our results compare favourably with the performance of other decoders, which we list in \cref{tab:pth_compare}. We obtain the highest reported code-capacity threshold and the highest reported single-shot threshold. 

We remark that the sustainable threshold that we observe for the 3D toric code is very close to the threshold of MWPM for string-like errors in the 3D toric code~\cite{wang03}. This implies that the performance of decoder 1 (the syndrome-repair step) is limiting the performance of the entire decoding procedure, as was suggested in~\cite{duivenvoorden2018renormalization}. Although the sustainable thresholds we observe for 3D surface codes are slightly higher than for 3D toric codes, the codes we consider are relatively small, which means that boundary effects may be having an impact on our sustainable threshold estimates.

We also investigated the suppression of the logical error rate below threshold in the 3D toric code, using MWPM \& BP+OSD. We use the following ansatz for the logical error rate for values of $p < p_{\mathrm{th}}$,
\begin{equation}
    p_\mathrm{fail}(L) \propto (p/p_{\mathrm{th}})^{\alpha L^\beta},   
\end{equation}
where $\alpha$ and $\beta$ are parameters to be determined. The code distance of the 3D toric code for $Z$ errors is $L^2$, so if the decoder is correcting errors up to this size, we would expect $\beta \approx 2$. Using the fitting procedure described in App.~\ref{app:fitting}, we estimate $\beta = 1.91(3)$ for $N=0$ (code capacity) and $\beta = 1.15(3)$ for $N=8$ (eight rounds of single-shot error correction). Therefore, for the (relatively small) codes that we consider, we find evidence that BP+OSD is correcting errors of weight up to the code distance. Viewed as an error correction problem, the distance of the syndrome-repair step of decoding (i.e~the single-shot distance $d_{ss}$) is $L$, which is consistent of our observed value of $\beta$ in the single-shot case. This provides further evidence that the bottle-neck of our single-shot decoding procedure is the syndrome-repair step. 

\subsection{Non-topological codes}

So far, we have focused on the decoding of 3D topological codes. We now show that the BP+OSD decoder can be used for single-shot decoding of more general 3D product codes. Table \ref{tab:qldpc3d} shows a family 3D product codes constructed  by taking the 3D product of a $(3,4)$-LDPC code\footnote{A code whose parity check matrix has rows/columns of weight 3/4.} with two instances of the classical repetition code. The result is a code family where the number of logical qubits is not fixed. This code family was decoded using a two-stage single-shot decoder, BP+OSD $\times 2$,
yielding the threshold plot in Fig.~\ref{fig:qldpc3d_threshold}. The simulation results suggest a sustainable threshold in the region of $2.7\%$.

\begin{figure}
    \centering
    \includegraphics[width=\columnwidth]{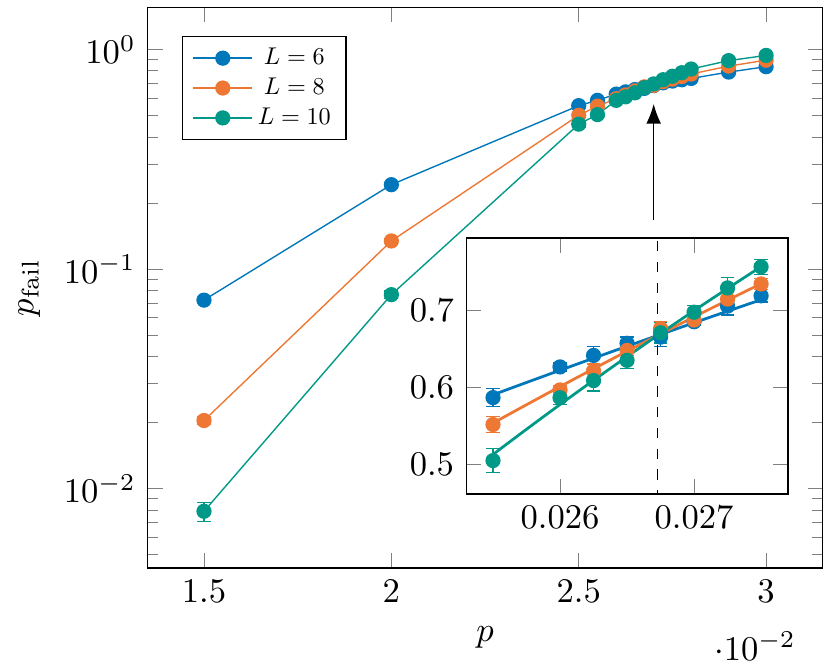}
    \caption{Threshold plot for a family of non-topological 3D product codes after $16$ rounds of single-shot error correction using the the BP+OSD decoder. The simulation results suggest a threshold at $2.7\%$. The error bars show the 95\% confidence intervals $p_{\rm fail} = \hat p_{\rm fail} \pm 1.96\sqrt{p_{\rm fail}(1 - p_{\rm fail})/\eta}$, where $\eta$ is the number of Monte Carlo trials.}
    \label{fig:qldpc3d_threshold}
\end{figure}

\section{Conclusion \label{sec:conc}}

In this article, we investigated single-shot decoding of 3D product codes. We gave a formal definition of confinement in quantum codes and proved that all 3D product codes have confinement for $Z$ errors. We also proved that confinement is sufficient for single-shot error correction against adversarial noise. This is a strengthening of the result of Campbell~\cite{campbell2019theory}, who showed that a property called soundness is sufficient for single-shot error correction, in that soundness implies confinement but the converse is not true.
Remarkably, there are important classes of codes, such as quantum expander codes, which have confinement but not soundness. Further to that, we prove that codes with linear confinement, and so expander codes, do have a single-shot threshold for local stochastic noise.
The obvious open problem arising from our work is how to extend our findings for linear confinement to the super linear case. Is confinement, in general, a sufficient condition for quantum-LDPC codes to exhibit a single-shot threshold? If not, what other requirements should a code satisfy to ensure the existence of a single-shot threshold?

We simulated single-shot error correction for a variety of 3D product codes, concentrating on 3D toric and surface codes. Using MWPM \& BP+OSD, we achieved the best known code capacity error threshold and sustainable single-shot error threshold for this code family (for phase-flip noise). Our results strongly suggest that the bottleneck of two-stage decoders is the first stage where the noisy syndrome is repaired. For the 3D toric code, the optimal threshold of the syndrome repair step is $3.3\%$~\cite{Ohno2004}, whereas the optimal threshold of the entire decoding problem is $11.0\%$~\cite{takeda2004}. This implies that two-stage decoders can never achieve optimal performance in these codes, so perhaps other single-shot decoding methods ought to be investigated in future. 

We also simulated single-shot error correction for a family of non-topological 3D product codes, using BP+OSD for both decoding steps. We achieved performance very close to that of the 3D toric and surface codes, which indicates that BP+OSD is a high-performance single-shot decoder. Furthermore, the versatility of BP+OSD means that we expect it to work as a single-shot decoder for general LDPC 3D product codes. We leave confirmation of this to future work, and we conjecture that BP+OSD will achieve good performance for other classes of quantum-LDPC codes such as topological fracton codes~\cite{Vijay2016,Dua2019}.

\textit{Acknowledgements.-} This work was supported by the Engineering and Physical Sciences Research Council [grant numbers EP/P510270/1 (J.R.S.) and EP/M024261/1 (E.T.C. and J.R.)].
M.V. thanks Aleksander Kubica and Nikolas Breuckmann for illuminating discussions. We thank Rui Chao for comments on an early version of the manuscript.
Research at Perimeter Institute is supported in part by the Government of Canada through the Department of Innovation, Science and Economic Development Canada and by the Province of Ontario through the Ministry of Colleges and Universities.
This research was enabled in part by support provided by Compute Ontario (\url{www.computeontario.ca}) and Compute Canada (\url{www.computecanada.ca}). This work was completed while ETC was at the University of Sheffield.

\bibliographystyle{unsrt}
\bibliography{SSEC_library}


\appendix

\section{Linear confinement and single-shot threshold}
\label{app:threshold}

We present the Stochastic Shadow decoder, a variant of the (Adversarial) Shadow decoder described in Def.~\ref{def:ball_decoder}, and prove that it succeeds in correcting errors that have connected components that are sufficiently sparse and of bounded size, both on the syndrome and the qubits (Lem.~\ref{lem:shadow_sto}). Thm.~\ref{thm:threshold} will then follow from Lem.~\ref{lem:shadow_sto} on the performance of the Stochastic Shadow decoder: a family of codes with good linear confinement has a single-shot threshold under the local stochastic noise model.

This Appendix is organised as follows. After fixing some graph-theory notation in Sec.~\ref{sec:preliminaries}, we introduce a novel weight function for node sets in a graph, the \textit{closeness} function, Sec.~\ref{sec:closeness}. We prove that the closeness weight function preserves confinement and that  the Stochastic Shadow decoder can be used on confined codes to keep the closeness of error under control (Sec.~\ref{sec:confinement_sto}). Crucially, the closeness weight function characterises the structure of local stochastic errors better that the Hamming weight does, as some classic results in percolation theory show. We conclude, in Sec.~\ref{sec:percolation}, by showing that a family of codes with good linear confinement has a sustainable single-shot threshold (Thm.~\ref{thm:threshold}). Our proof is built on the results in \cite{leverrier18, grospellier2019constant}, where the authors prove that expander codes (which have linear confinement) have a single-shot threshold when decoded via the small-set flip decoder.
\subsection{Notation and preliminaries}
\label{sec:preliminaries}
Given a stabiliser code on $n$ qubits with stabiliser group $\mathcal S \subseteq \mathcal P_n$, we associate to it two graphs: $(\mathcal G_q, \sim_q)$, the qubit graph, and $(\mathcal G_s, \sim_s)$, the syndrome graph. The set of nodes are $\mathcal G_q$, the $n$ qubits, and $\mathcal G_s$, a generating set of the stabilizer group $\mathcal S$\footnote{When the code is a CSS code we consider the group generated by the $X$-stabilizers and the $Z$-stabilizers separately. $\mathcal S$ will thus refer either to $\mathcal S_X$ or $\mathcal S_Z$.}. The adjacency relations $\sim_q$ and $\sim_s$ are defined as   
\begin{align*}
q_1 \sim_q q_2 &\Leftrightarrow \exists s \in \mathcal G_s \text{ s.t. } \{q_1, q_2\} \subseteq \supp(s),\\
s_1 \sim_s s_2 &\Leftrightarrow \exists q \in \mathcal G_q \text{ s.t. } q \in \supp(s_1) \cap \supp(s_2);
\end{align*}
where the support $\supp(s) $ of a Pauli operator $s$ in $\mathcal P_{n}$ is the set of all the qubits on which its action is non-trivial. We will use lower-case symbols for Pauli operators in $\mathcal{P}_{n}$ and the corresponding upper-case symbol to indicate its support e.g. $E\coloneqq \supp(e)$. We use the term error to refer interchangeably to a Pauli operator or its support, in particular given two Pauli operators $e_1$ and $e_2$ we use the symbol $+$ to indicate the support of the product operator $e_1 \cdot e_2$ \footnote{E.g. if $e_1$ and $e_2$ are both $X$-operators, $E_1 + E_2$ is the symmetric difference of the sets $E_1$ and $E_2$.}, so that
\begin{align*}
E_1 + E_2 = \supp(e_1\cdot e_2).
\end{align*}
In this picture, the syndrome $\sigma(\cdot)$ maps the set of Pauli operators on $n$ qubits $\mathcal P_n$ into the power set of $\mathcal G_s$, a set of generators for the stabiliser group:
\begin{align*}
\sigma: \mathcal P_n &\longrightarrow \mathscr P(\mathcal G_s) \\
e &\longrightarrow \{s \in \mathcal G_s : se = -es\}.
\end{align*} 
We define the neighbour map $\Gamma$ as
\begin{align*}
\Gamma: \mathcal P_n &\longrightarrow \mathscr P(\mathcal G_s) \\
e &\longrightarrow \{s \in \mathcal G_s : \supp(s)\cap E \neq \emptyset \}.
\end{align*} 
With slight abuse of terminology, we call syndrome any element of $\mathscr P (\mathcal G_s)$, even when such a set does not belong to the image of $\sigma$.  When referring to an error as a set $E$, it is always to be intended as corresponding to a fixed Pauli operator $e \in \mathcal P_n$ such that $E \coloneqq \supp(e)$. We will write interchangeably $\sigma(e)/\sigma(E)$ and $\Gamma(e)/\Gamma(E)$ to indicate the image, via the syndrome map and the neighbour map respectively, of the Pauli error $e$.

Given two syndromes sets in $\mathcal G_s$ we use the symbol $+$ to indicate their symmetric difference. It is easy to verify that the map $\sigma(\cdot)$ preserves the $+$ operation (i.e. it is linear):
\begin{align*}
\sigma(e_1 \cdot e_2) = \sigma(E_1 + E_2) = \sigma(E_1) + \sigma(E_2).
\end{align*}
Moreover, the image via $\sigma$ of disjoint non-connected sets is disjoint. In fact, if $E_1, E_2$ are two disjoint non-connected sets in $\mathcal G_q$ and we suppose that their syndrome sets are not disjoint we find a contradiction.  Let $\hat s$ be a stabiliser in $ \sigma(E_1) \cap \sigma(E_2)$. By definition of $\sigma$, this entails that $e_1$ and $e_2$ both anti-commute with $\hat s$ which is equivalent to saying that their supports have odd overlap with $\supp(\hat s)$. In particular, there exists $q_i \in E_i$ such that $q_i \in \supp(\hat s)$ and, by the definition of the adjacency relation $\sim_q$, $q_1 \in E_1$  and $q_2 \in E_2$ would be connected via $\hat s$, against the assumption. Note that, in general the image via the syndrome map $\sigma(\cdot)$ of a connected set needs not to be connected. However, the neighbour function $\Gamma(\cdot)$ maps connected sets into connected sets. We will make use of these properties in Sec.~\ref{sec:confinement_sto}.
\subsection{The closeness weight function}
\label{sec:closeness}
When errors are local stochastic it can be handy to use definitions of weight other than the cardinality/Hamming weight. For instance, the authors in \cite{leverrier18} define the quantities of Def.~\ref{def:maxconn} and study a related notion of percolation to understand the tolerance to errors of a given connected graph.   
\begin{defin}[$\alpha$-subsets, $\mathrm{MaxConn}_\alpha(E)$ \cite{leverrier18}]\label{def:maxconn}
An $\alpha$-subset of a set $E \subseteq \mathcal G_q$ is a set $K$ such that $|K \cap E|\ge \alpha |K|$. The maximum size of a connected $\alpha$-subset of $E$ is denoted by $\mathrm{MaxConn_{\alpha}(E)}$.
\end{defin}
We here introduce a conceptual cousin to $\mathrm{MaxConn}_{\alpha}(E)$, the $\beta$-closeness of an error set $E$, and prove that it is a well defined weight function (see Lem.~\ref{lemma:bcloseness}). We do not explicitly detail the relations between $\alpha$-subsets and closeness here. However, we will implicitly use them, as our percolation results and ultimately the proof of Thm.~\ref{thm:threshold} heavily rely on those relations and the proofs in \cite{leverrier18, grospellier2019constant}. 
\begin{defin}[$\beta$-closeness]\label{def:closeness}
Let $\mathcal G$ be a connected graph i.e.~a graph in which there exist a path between any two of its nodes. Given a subset $E$ of nodes and a positive integer $\beta$, we define its $\beta$-closeness as the quantity:
\begin{align*}
\|E\|_{\beta} \coloneqq \max \{\lvert K \cap E \rvert : K \text{ is connected}, \, \lvert K\rvert=\beta\}.
\end{align*}
We call any connected subset of $\beta$ nodes a $\beta$-patch and any $\beta$-patch $K$ such that $|K\cap E|=\|E\|_{\beta}$ maximal patch for $E$.
\end{defin}
Since we are interested in the $\beta$-closeness of error sets on a qubit graph $\mathcal G_q$, it is natural to introduce the notion of reduced $\beta$-closeness.
\begin{defin}\label{def:reduced_closeness}
Given a qubit error set $E \subseteq \mathcal G_q$, its reduced $\beta$-closeness $\|E\|^{\mathrm{red}}_{\beta}$ is defined as 
\begin{align*}
\|E\|^{\mathrm{red}}_{\beta} \coloneqq \min\{\|E + T\|_{\beta} :\, &\sigma(E+T)=\sigma(E),\\ &T=\supp(t) \text{ for some } t \in \mathcal{P}_n\}.
\end{align*}
\end{defin}
Crucially, we will see in Lem.~\ref{lem:patches} that the closeness function preserves confinement. As a consequence, we can build a variant of the Shadow decoder (Def.~\ref{def:shadow}) that succeeds in correcting errors of small reduced closeness.

We now prove some basic properties of the $\beta$-closeness weight function $\|\cdot\|_{\beta}$ on a connected graph $\mathcal{G}$. 
\begin{lemma}
\label{lemma:bcloseness}
Let $\mathcal G$ be a connected graph and denote by $|\mathcal G|$ the number of its nodes.
For any positive integer $\beta< |\mathcal G|$, the following hold:
\begin{enumerate}[label=(\roman*)]
\item $\| \cdot \|_{\beta} \le |\cdot|$;
\item \label{prop:size_comp} $\|\cdot \|_{\beta} \le \beta$; the equality holds if and only if the considered set of nodes has a connected component of size at least $\beta$; conversely, if $\|\cdot\|_{\beta} < \beta$ then the connected components of the set all have size less than $\beta$;
\item it is positive:$\|E\|_{\beta} \ge 0$ and equality holds if and only if $E=\emptyset$;
\item it satisfies the triangle inequality: for any $E_1, E_2$, $\|E_1\cup E_2\|_{\beta} \le\|E_1\|_{\beta}  + \|E_2\|_{\beta}$.
\item it is monotonic: if $E_1 \subseteq E_2$ then $\|E_1\|_{\beta} \le \|E_2\|_{\beta}$;
\end{enumerate}
\end{lemma}
\begin{proof}
In the following, let $K\subseteq \mathcal G$ be a maximal $\beta$-patch for $E$, i.e.~$\|E\|_{\beta} = |K \cap E|$.
\begin{enumerate}[label=(\roman*)]
\item $\|E\|_{\beta} = |K \cap E| \le |E|$.
\item $|K \cap E| \le| K| = \beta$. Equality holds if and only if $K \cap E = K \subseteq E$ which entails that $E$ has a connected component of size at least $\beta$, since $K$ is connected. 
\item If $E$ is non empty then there exists at least one node $g \in E$. Since $\mathcal G$ is connected, for any integer $1\le \beta \le \lvert \mathcal G\rvert$ there exists a $\beta$-patch that contains $g$ so that $\|E\|_{\beta} \ge 1$.
\item Let $J$ be any $\beta$-patch in $\mathcal G$. The following hold:
\begin{align*}
\lvert J \cap (E_1 \cup E_2) \rvert &= \lvert (J \cap E_1) \rvert + \lvert (J \cap F_2)\rvert +\\
&\quad -\lvert J \cap (E_1 \cap E_2)\rvert \\
&\le \lvert J \cap E_1 \rvert + \lvert J \cap E_2 \rvert\\
&\le \| E_1 \|_{\beta} + \|E_2\|_{\beta}.
\end{align*}
Since this holds for any $\beta$-patch, we obtain 
\begin{align*}
    \|F_1 \cup F_2 \|_{\beta} \le \|F_1\|_{\beta} + \|F_2\|_{\beta}.
\end{align*}
\item Let $K_1$, $K_2$ be maximal $\beta$-patches for $E_1$ and $E_2$ respectively. Then 
\begin{align*}
|K_1 \cap E_1 |&\le |K_1 \cap E_2| &\text{because }E_1\subseteq E_2,\\
&\le |K_2 \cap E_2|&\text{ by maximality of } K_2,
\end{align*}
which yields $\|E_1\|_{\beta} \le \|E_2\|_{\beta}$.
\end{enumerate}
\end{proof}
Lem.~\ref{lem:canonical_patch} below states that there exists a canonical form for maximal $\beta$-patches of an error set $E$. Roughly speaking, a canonical $\beta$-patch $K$ will be made up of some entire connected components of $E$, plus at most one connected proper subset of a connected component of $E$, and some other nodes not in $E$ (see Fig.~\ref{fig:patches}). The existence of a canonical $\beta$-patch is key in proving that the closeness function preserves confinement in the sense explained by Lem.~\ref{lem:patches}.

\begin{figure*}
\centering
\subfloat[]{\includegraphics[scale=0.38]{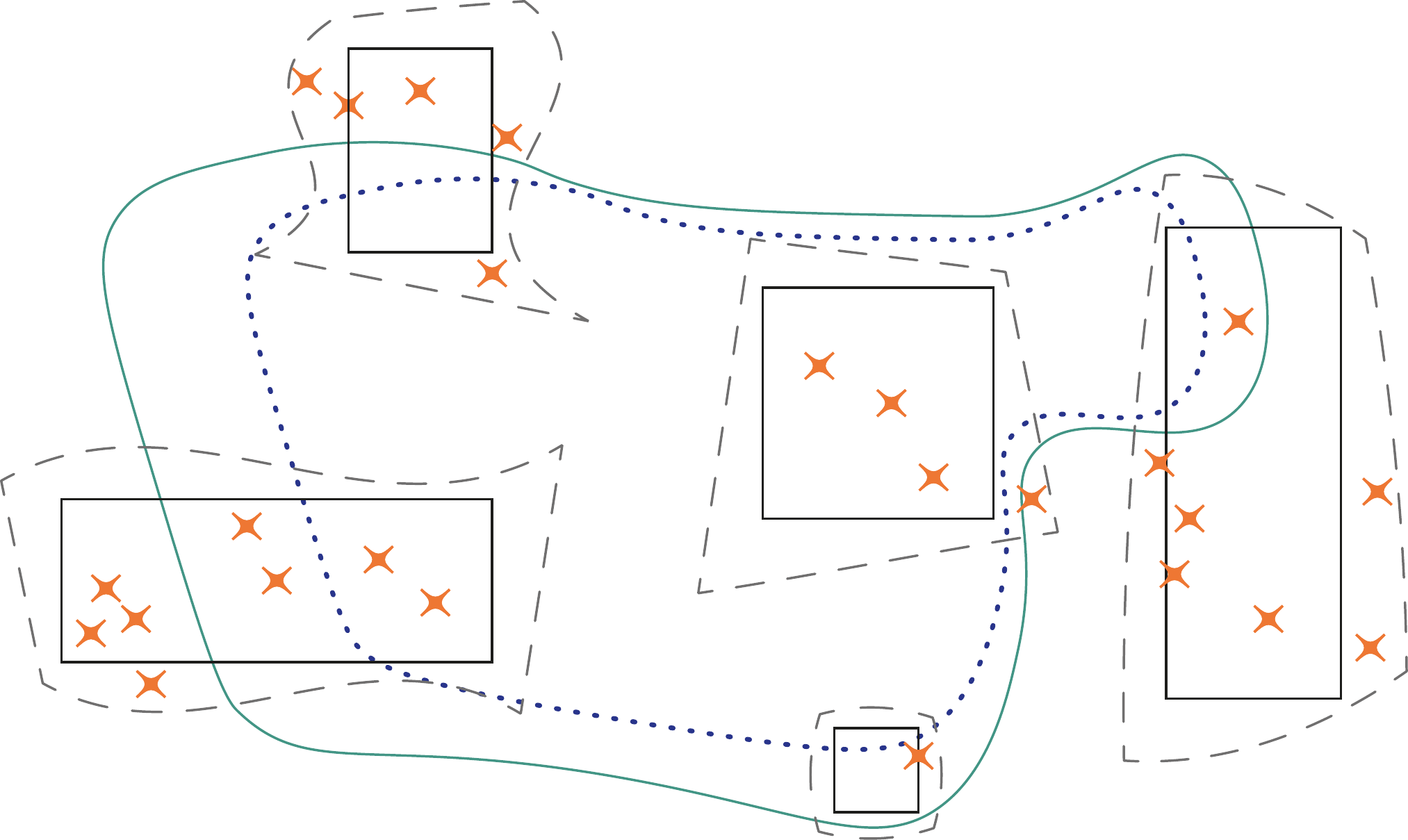}}\hfill
\subfloat[]{\includegraphics[scale=0.38]{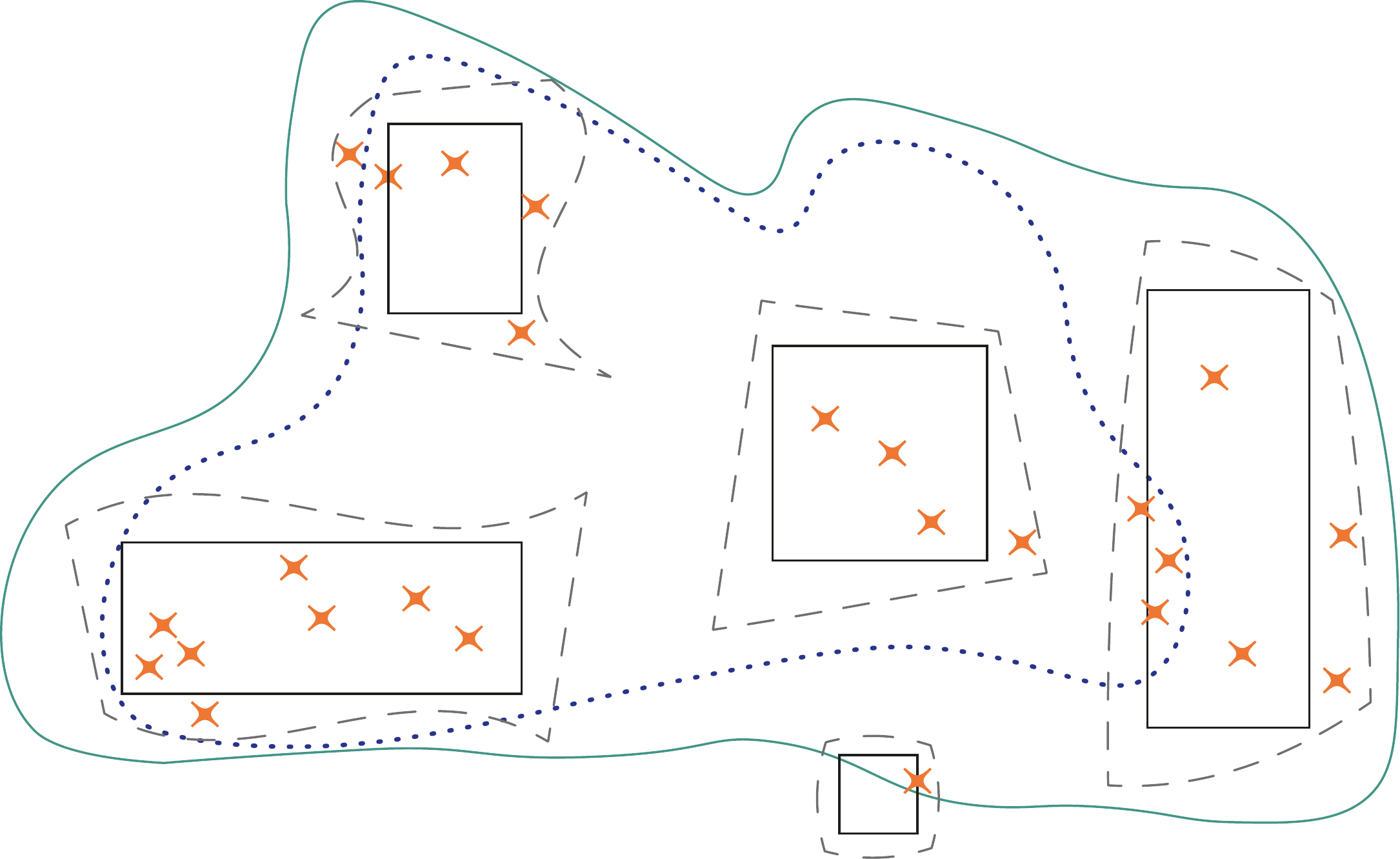}}
\caption{Graphical representation of patches on a graph. To help the visualisation we imagine the qubit graph and the syndrome graph to be superimposed. Black rectangles: connected components of the error $E_1, \dots, E_5$. Dashed grey lines: neighbour sets $\Gamma(E_i)$ of the underlying rectangle/error component. Orange crosses: syndrome nodes in $\sigma(E_i)$. Dotted blue curve: $t$-patches on the qubit graph. Green curve: $\omega t$-patches on the syndrome graph. In (a) the patches are generic while in (b) the dotted/error patch is a canonical patch for the error. The importance of the canonical form for a patch is highlighted in the differences between the patches in (a) and (b). We observe how the crosses/syndrome nodes $\sigma(E_i)$ are scattered inside the dashed curve/neighbour set $\Gamma(E_i)$. For this reason, in order to group enough syndrome nodes inside a patch of bounded size, we need some care in the choice of the error nodes. When we include entire connected components of the error in a patch in $\mathcal G_q$, we are able to build a patch in $\mathcal G_s$ which includes entire neighbour sets and, as a consequence, all the corresponding syndrome nodes. In fact, even if we assume that the dotted blue/error patches in (a) and (b) have same size, when we enlarge them by a factor of $\omega$ to build the green/syndrome patch, we obtain dramatically different results. In (a) since the dotted/error patch contains several incomplete components, the corresponding green/syndrome patch contains incomplete portions of the dashed/neighbour sets $\Gamma(E_i)$. Hence, we have no guarantee on the number of crosses/syndrome nodes included in the patch. In (b) the dotted blue patch is a canonical patch for the error. We can see how the green/syndrome patch entirely contains the dashed/neighbour sets of all but one component of the error contained in the dotted blue/qubit patch. In this way we have the certainty to include in the green/syndrome patch a sufficient number of crosses/syndrome nodes to ensure confinement. }
\label{fig:patches}
\end{figure*}

\begin{lemma}[Canonical $\beta$-patch]
\label{lem:canonical_patch}
For any error $E$ on a qubit graph $\mathcal G$ there exists a maximal $\beta$-patch $T$ such that, for all but one connected component $E_i$ of $E$, it holds:
\begin{align*}
\text{either}\quad E_i \subseteq T &\quad\text{ or}\quad E_i \cap T = \emptyset.
\end{align*}
In other words, if $E_1, \dots, E_m$ are the connected components of $E$, re-ordering if necessary, there exists an index $\nu$ such that:
\begin{align}
    |T \cap E_i| &= |E_i| &\text{if}\quad i<\nu,\nonumber\\\label{eq:cond_patches}
    |T \cap E_i| &\le |E_i| &\text{if }\quad i = \nu, \\
    |T \cap E_i|& = 0 &\text{if }\quad i > \nu. \nonumber
\end{align}
We call any such $T$ a canonical $\beta$-patch for the set $E$.
\end{lemma}
\begin{proof}
Let $J$ be any maximal $\beta$-patch for $E$ i.e.\ $J$ is connected, has size $\beta$ and $|J \cap E| = \|E\|_{\beta}$. Staring from $J$ we build a set $T$ with the desired form.
Write $J \cap E$ as disjoint union of connected sets:
\begin{align*}
J \cap E =J_1 \sqcup \dots \sqcup J_{\nu}.
\end{align*}
We call these $J_i$'s patch-error components. Let $E_1 \dots E_{\mu}$ be the connected components of the error $E$. We recall that a connected component $E_{i}$ of $E$ is a connected set which is connected to no additional nodes in $E\setminus E_{i}$. We say that $E_i$ is \textit{incomplete} with respect to $J$ if it has non trivial overlap with $J$ but it is not entirely contained in $J$, i.e.
\begin{align*}
    J\cap E_i \neq \emptyset \quad \& \quad E_i \not\subseteq J \Rightarrow |J\cap E_i| < |E_i|.
\end{align*}
Note that it can be the case for two disjoint (but internally connected) error-patch components $J_{i_1}$ and $J_{i_2}$ to overlap with the same incomplete error component $E_{i'}$.

We consider a meta-graph $\mathfrak G$ whose meta-nodes are connected sets in $\mathcal G$ and meta-edges are paths in $\mathcal G$. Because the error-patch components are both internally and reciprocally connected in $J$, there exists a meta-spanning-tree $\mathcal T \subseteq \mathfrak G$ whose $\nu$ nodes $\mathcal J_i$ are the error-patch components $J_i$ and whose meta-edges $\varepsilon_{ij}$ are formed by minimum length paths in $\mathcal G$ between the $J_i$'s with nodes in $J\setminus E$. In the following we will indicate with $\mathcal T, \mathcal J_i$ and $\varepsilon_{ij}$ the meta-tree, its meta-nodes and its meta-edges and with $T$, $J_i$ and $e_{i,j}$ the corresponding sets of nodes in $\mathcal G$. Note that, by this identification, $T$ has at most $\beta$ nodes. We now show how to modify the meta-tree $\mathcal T$ so that the corresponding set of nodes $T$ in $\mathcal G$ is canonical for $E$. We do this in two steps: the balancing and the enlargement step.
\paragraph*{\textsc{Balancing}}
We show by induction on the number $\nu$ of the meta-nodes $\mathcal J_i$'s that it is possible to modify $\mathcal T$ so that the corresponding set of nodes $T\subseteq\mathcal{G}$ satisfies conditions \eqref{eq:cond_patches} on its overlap with the connected components of $E$.
\begin{enumerate}
\item[$\nu=1$]: the thesis is trivially verified. 
\item[$\nu >1$]: if $J$ is not canonical for $E$ then $E$ must have at least two incomplete components with respect to the patch $J$. Let $\mathcal J_{\ell}$ be a meta-leaf of $\mathcal T$ and $J_{\ell}$ its corresponding subset of nodes in $\mathcal G$. We iteratively remove from $T$ the nodes of $J_{\ell}$, both preserving connectivity of $T$ and the size of $T\cap E$. 

For any node $q_{\lambda}$ in $J_{\ell}$, we choose a node $q_{\chi}$ such that:
\begin{enumerate}[label=\roman*.]
     \item $q_{\chi}$ belongs to some incomplete component of the error disjoint from $J_{\ell}$: $q_{\chi} \in E_{\chi}$ and $E_{\chi} \cap J_{\ell} = \emptyset$;
    \item $q_{\chi}$ is a new node i.e.\ it does not belong to $J$: $q_{\chi} \in \mathcal G \setminus J$;
   
    \item $q_{\chi}$ is connected to at least one node in some error-patch component other from $J_{\ell}$: $q_{\chi} \sim_q q_{\chi'}$, $q_{\chi'} \in J_{\chi}$ for some $\chi \neq \ell$. 
\end{enumerate}
We remove $q_{\lambda}$ from J and add $q_{\chi}$ to J,  and thereby updating T accordingly.
This process terminates when either (a) we are not able to find such a new node $q_{\chi}$ or (b) there are no more nodes $q_{\lambda}$ in $ J_{\ell}$.

Case (a) entails that $E$ has at most one incomplete component with respect to $T$. In fact, if $E$ had an incomplete component $E_{\chi}$ disjoint from $J_{\ell}$ such a node $q_{\chi}$ always exists. As a consequence, if we are not able to find a new error node to enlarge one of the error-patch components $J_{\chi} \neq J_{\ell}$ the only incomplete component of $E$ must be the one relative to $J_{\ell}$. The updated node set $T$ has the desired property, provided that we had removed nodes $q_{\lambda}$ from $J_{\ell}$ preserving connectivity (for instance, considering a spanning tree for nodes in $J_{\ell}$ and iteratively removing leaves). 
If case (b) is verified, we remove from $\mathcal T$ all the meta-edges that were incident to $\mathcal J_{\ell}$. The updated meta-tree $\mathcal T$ derived from the updated set $T$ has $\nu -1$ meta-nodes. By the induction hypothesis, it can be modified to obtain the desired form.

In other words, we pick a meta-leaf of $\mathcal T$ and we either remove part of its nodes (case (a)) or all of them (case (b)). We preserve the quantity $|T\cap E|$ by adding new error nodes to some different error-patch component that overlaps with an incomplete component of the error $E$. By choosing a leaf, we are able to preserve the connectivity of $\mathcal T$ and thus the connectivity of the corresponding node sets $T$.

We iterate this procedure over meta-leaves of $\mathcal T$ until the overlap of the corresponding set $T$ in $\mathcal G$ and the error set $E$ has the desired form.
\end{enumerate}

\paragraph*{\textsc{Enlargement}}
By contradiction, we prove that it is possible to add nodes to the set $T$ corresponding to the balanced meta-tree $\mathcal T$ so that it is connected, it has size exactly $\beta$ and $|T\cap E| = \|E\|_{\beta}$.
First note that during the balancing procedure, the number $|T\cap E|$ remains constant and it holds:
\begin{align*}
|T \cap E| = \sum_{i=1}^\nu |J_i| = |J \cap E| = \|E\|_{\beta}.
\end{align*}
Moreover, the initial tree is connected and the balancing procedure preserves connectivity. However, we only have an upper bound on the size of $T$. In fact, if $\mathcal T$ is the initial meta-tree and $T$ is its corresponding sub-graph in $\mathcal G$, it holds $T \subseteq J$ and therefore $|T|\le \beta$. During the balancing step the size of $T$ could decrease when we remove nodes of $e_{ij}$, belonging to a meta-edge $\varepsilon_{ij}$. Thus, in general, after the balancing step for the weight of $T$ holds:
\begin{align*}
|T| \le \beta.
\end{align*}
If $|T|=\beta$, then $T$ is a $\beta$-patch with maximum overlap with $E$ and, by balancing, it is canonical.
If $|T|<\beta$, then there  must exist at least $\beta - \|E\|_{\beta}$ nodes in $\mathcal G \setminus (E \cup T)$ that are connected to $T$. In fact, a connected proper subset can always be enlarged in a connected graph. If the only way to enlarge $T$ to a $\beta$-patch were by adding nodes in $E$, then we would have found a $\beta$-patch whose overlap with $E$ has size greater than its $\beta$-closeness, which contradicts the definition of $\|E\|_{\beta}$. In conclusion, any of such enlargements of the tree $T$ is a canonical $\beta$-patch for $E$.
\end{proof}

\subsection{Confinement and Stochastic Shadow Decoder}
\label{sec:confinement_sto}
Here, we first prove that the closeness function preserves confinement, as Lem.~\ref{lem:patches} states. Then, we present the Stochastic Shadow decoder (Def.~\ref{def:shadow}) and prove, in Lem.~\ref{lem:shadow_sto}, that it succeeds in correcting errors of small enough closeness. These findings, together with the percolation results of Sec.~\ref{sec:percolation}, will yield the proof of the existence of a single-shot threshold for codes with linear confinement.
\begin{lemma}[Closeness preserves confinement]\label{lem:patches}
Consider a code with qubit degree at most $ \tilde \omega$ and $(t, f)$-confinement, where $f$ is convex. Then, for any error $E$ with $\|E\|_t^{\mathrm{red}} \le \frac{t}{2}$, it holds:
\begin{align*}
f(\|\sigma(E)\|_{\omega t} )\ge \|E\|_t^{\mathrm{red}} ,
\end{align*}
where $\omega = \tilde \omega + 1$.
\end{lemma}
\begin{proof}
To ease the notation, let $F$ be an error set such that $\sigma(E)= \sigma(F)$, $\|F\|_t = \|E\|_t^{\mathrm{red}}$. If $F_1, \dots, F_{\mu}$ are the connected components of $F$, by Lem.~\ref{lem:canonical_patch} there exists a canonical patch $K$ for $F$ such that:
\begin{align*}
    |K \cap F_i| &= |F_i| &\text{if}\quad i<\nu,\\
    |K \cap F_i| &\le |F_i| &\text{if }\quad i = \nu, \\
    |K \cap F_i|& = 0 &\text{if }\quad i > \nu. 
\end{align*}
for some $\nu \le \mu + 1$.

First, we prove that there exists an $\omega t$-patch $J$ in the syndrome graph $\mathcal G_s$ such that it contains the syndrome of the connected components $F_1, \dots, F_\nu$ of the error which intersect the canonical patch $K$:
\begin{align*}
\bigsqcup_{i=1}^{\nu} \left( \sigma(F_i) \right) \subseteq J.
\end{align*}
Then, we prove that such a patch $J$ has overlap with $\sigma(F)$ of Hamming weight large enough to ensure confinement with respect to the closeness function:
\begin{align*}
    f(\|\sigma(F)\|_{\omega t}) \ge \|F\|_t.
\end{align*}
We will then find the desired bound on $E$ using the initial assumptions $\sigma(F) = \sigma(E)$ and $\|F\|_t = \|E\|_t^{\mathrm{red}}$.
\paragraph*{\textsc{Existence of $J$}}
We build a $\omega t$-patch $J$ on $\mathcal G_s$ as follows. We define $J$ as the disjoint union of the (at most) $\tilde \omega |F_i|$ connected nodes $\Gamma(F_i)$:
\begin{align*}
    J = \bigsqcup_{i=1}^{\nu} \Gamma(F_i).
\end{align*}
Let $\pi$ be the set of edges of a minimum length path in $K$ that connects all its $\nu$ disjoint error components $F_i$. These edges correspond naturally to a set $\pi_s \subseteq \mathcal G_s$ if we associate to the edge $(q_1, q_2)$, the corresponding stabilizer in $\mathcal G_s$, remembering that:
\begin{align*}
q_1 \sim_q q_2 \Leftrightarrow \{q_1, q_2\} \subseteq \supp(s), \quad s \in \mathcal G_s.
\end{align*}
Under this identification, importantly, adjacent edges are mapped into neighbouring syndrome nodes. We add the set $\pi_s$ to $J$. As a result, $J$ is now connected.
For the size of $J$, it holds:
\begin{align*}
|J|& \le \tilde \omega \sum_{i=1}^{\nu} |F_i| + |\pi_s|.
\end{align*}
By hypothesis, $t/2 \ge \|F\|_t=|K \cap F|$ and because $K$ is canonical for $F$, i.e.~$|K \cap F| \le \sum_{i=1}^{\nu} |F_i|$, we have:
\begin{align*}
    \sum_{i=1}^{\nu - 1} |F_i|  &\le \frac{t}{2}.
\end{align*}
Combining property \ref{prop:size_comp} of the closeness weight function and the assumption $\|F\|_t \le \frac{t}{2}$, yields, for any $i$, and $\nu$ in particular,
\begin{align}
\label{eq:bound_comp}
|F_i| \le \frac{t}{2}.
\end{align}
Since $\pi$ has edges in $K$, $\pi_s$ has size at most $|K|$ i.e.~:
\begin{align*}
    |\pi_s| \le t.
\end{align*}
Adding up, we obtain:
\begin{align*}
|J| \le \omega t
\end{align*}
where $\omega = \tilde \omega + 1$. By enlarging $J$ if necessary to include exactly $\omega t$ nodes, and remembering that by construction it is connected, we have found that $J$ is an $\omega t$-patch in $\mathcal G_s$, as desired.

\paragraph*{\textsc{Overlap of $J$ with the error syndrome}}
Eq.~\eqref{eq:bound_comp}  entails in particular that any connected error component $F_1, \dots, F_{\nu}$ that has non-trivial overlap with the patch $K$, has size smaller than $t$ and therefore it has confinement:
\begin{align}
\label{eq:comp_conf}
f(|\sigma(F_i)|) \ge |F_i|.
\end{align}
Because $\sigma$ maps disjoint sets of $\mathcal G_q$ in disjoint sets of $\mathcal G_s$, 
\begin{align}
\sigma \left( \bigsqcup_{i=1}^\nu F_i \right)&= \bigsqcup_{i=1}^{\nu} \sigma(F_i)  \nonumber \\
\Rightarrow \quad \Big\lvert\sigma \left( \bigsqcup_{i=1}^\nu F_i \right)\Big\rvert&= \sum_{i=1}^{\nu} |\sigma(F_i)|. \label{eq:right_summ}
\end{align}
Thus, applying $f$ to each term of the summation of Eq.~\eqref{eq:right_summ} we have:
\begin{align}
\label{eq:key_sum}
   \sum_{i=1}^{\nu} f(|\sigma(F_i)|) &\ge \sum_{i=1}^{\nu} |F_i|.
\end{align}
For the left hand side of Eq.~\eqref{eq:key_sum}, using convexity of $f$ we obtain:
\begin{align*}
    f\left( \sum_{i = 1}^{\nu} |\sigma(F_i)| \right) \ge \sum_{i=1}^{\nu} f(|\sigma(F_i)|),
\end{align*}
for the right hand side of Eq.~\eqref{eq:key_sum} instead, since $K$ is canonical for $F$, it holds that:
\begin{align*}
    \sum_{i=1}^{\nu} |F_i| &\ge |K \cap F|,
\end{align*}

Combining these two bounds for \eqref{eq:key_sum} yields:
\begin{align}
\label{eq:conf_on_sum}
   f\left( \sum_{i = 1}^{\nu} |\sigma(F_i)| \right) \ge \|F\|_{t}.
\end{align}
To obtain the thesis from Eq.~\eqref{eq:conf_on_sum}, we just need to substitute the Hamming weight on the left hand side with the closeness weight $\|\cdot \|_{\omega t}$.  
By construction, for $J$ it holds that:
\begin{align}
\label{eq:wt_1}
 |J \cap \sigma(F)|\ge \sum_{i=1}^{\nu} \lvert \sigma(F_i) \rvert.
\end{align}
Moreover, since $J$ is a $\omega t$-patch:
\begin{align}
\label{eq:wt_2}
\|\sigma(F)\|_{\omega t}&\ge |J \cap \sigma(F)|
\end{align}
Using the monotonicity of $f$ and combining Eq.~\eqref{eq:wt_1}, \eqref{eq:wt_2} and \eqref{eq:conf_on_sum} yields:
\begin{align*}
f(\|\sigma(F)\|_{\omega t})    \ge \|F\|_t.
\end{align*}
\paragraph*{\textsc{Conclusion}}
Because $F$ is an error set equivalent to $E$ i.e.~$\sigma(F)=\sigma(E)$, such that $\|F\|_t = \|E\|_t^{\mathrm{red}}$, we conclude:
\begin{align*}
f(\|\sigma(E)\|_{\omega t}) \ge \|E\|_{t}^{\mathrm{red}}.
\end{align*}
for $\omega = \tilde \omega+ 1$. 

\end{proof}

Lem.~\ref{lem:patches} in particular entails that the closeness weight is in fact a sensible quantity to look at when dealing with errors on confined codes.

We now introduce the Stochastic Shadow decoder. The difference between this variant and the one previously presented (Def.~\ref{def:ball_decoder}) is on the weight functions used. While the standard/Adversarial Shadow decoder tries to minimise the Hamming weight of the residual error, the Stochastic Shadow decoder attempts to keep under control its closeness.
\begin{defin}[Stochastic Shadow decoder]
\label{def:shadow}
The Stochastic Shadow decoder has variable parameters $0< \alpha \le 1$, and $0< \beta, \gamma \in \mathbb Z$. Given an observed syndrome $S = \sigma(E) + S_e$ where $S_e \subseteq \mathcal{G}_s$ is the syndrome error, the Stochastic Shadow decoder of parameters $(\alpha, \beta, \gamma)$ performs the following 2 steps:
\begin{enumerate}
\item Syndrome repair: find $S_r$ of minimum $\gamma$-closeness $\| S_r\|_{\gamma}$ such that $S + S_r$ belongs to the $(\alpha, \beta)\tshadow$ of the code, where
\begin{equation*}
(\alpha, \beta)\tshadow = \{\sigma(E) \text{ s.t. }\| E \|_{\beta} \le \alpha \beta\}.
\end{equation*}
\item Qubit decode: find $E_r$ of minimum $\beta$-closeness $\| E_r \|_\beta$ such that $\sigma(E_r)= S + S_r$.
\end{enumerate}
We call $R = E + E_r$ the residual error.
\end{defin}
We have the following promise on the Stochastic Shadow decoder, which mirrors the results of Lem.~\ref{lem:ball_decoder} for the Adversarial Shadow decoder.
\begin{lemma}\label{lem:shadow_sto}
Consider a stabiliser code that has $(t, f)$-confinement and qubit degree $\le \omega -1$. Provided that the original error pattern $E$ has $\| E\|_t^{\mathrm{red}} \le t/2$, on input of the observed syndrome $S = \sigma(E) + S_e$, the residual error $R$ left by the Stochastic Shadow decoder of parameter $(\frac{1}{2}, t, \omega t)$ satisfies:
\begin{align}
  \| R\|_t^{\mathrm{red}} \leq f(2\| S_e\|_{\omega t})  .  
\end{align}
\end{lemma}
\begin{proof}
Thanks to Lem.~\ref{lem:patches}, we know that the closeness function preserves confinement. The proof is then a straightforward adaption of the proof of Lem.~\ref{lem:ball_decoder}, where the Hamming weight has to be substituted with $\|\cdot\|_t$ on error sets and $\|\cdot\|_{\omega t}$ on syndrome sets respectively. We here briefly report the proof for completeness.

Assume $\|E\|_t^{\mathrm{red}}\le  t/2$, and let $E_r$ be the output of the qubit decode step. By construction, it has minimum $t$-closeness among the errors with syndrome $S+S_r$, which belongs to the  $(\frac{1}{2}, t)\tshadow$ of the code. In particular, $\|E_r\|_t\le \frac{t}{2}$. 
We recall that the $+$ operation between two error sets in $\mathcal G_q$ denotes the support of the product of the two corresponding Pauli operators and, as such, it holds that (see Sec.~\ref{sec:preliminaries}):
\begin{align*}
E + E_r \subseteq E \cup E_r.
\end{align*}
By the property of the closeness weight function, this entails:
\begin{align*}
\|E + E_r\|_t \le \|E \cup E_r\|_t \le \|E \| + \| E_r\|_t. 
\end{align*}
The linearity of the syndrome function $\sigma(\cdot)$ yields:
\begin{align*}
\sigma(E + E_r) = \sigma(E) + \sigma(E_r) = S_e + S_r.
\end{align*}
Since $S_e$ is a possible solution of the syndrome repair step $\|S_r\|_{\omega t} \le \|S_e\|_{\omega t}$ and so,
\begin{align*}
\|S_e + S_r\|_{\omega t} &\le \|S_e\|_{\omega t} + \|S_r\|_{\omega t} \\
&\le 2 \|S_e\|_{\omega t}.
\end{align*}
Combining this and the monotonicity of $f$ gives:
\begin{align*}
\|E+E_r\|^{\mathrm{red}}_t \le f(2\|S_e\|_{\omega t}).
\end{align*}
\end{proof}
Lem.~\ref{lem:patches} tells us that the Stochastic Shadow decoder succeeds whenever the $t$-closeness of the error is small enough. Importantly then, if we are able to bound the probability of the complement of this event, we could infer an upper bound on the failure probability of our decoder.  This is the subject of the next Section.
\subsection{Percolation results and proof of Theorem \ref{thm:threshold}}
\label{sec:percolation}
We consider error sets $E$ on the qubit graph $\mathcal G_q$ and error sets $S_e$ on the syndrome graph $\mathcal G_s$ and we assume that the probability of observing a particular error is at most exponential in its size. Formally, we use this error model
\begin{defin}[Local stochastic error]
An error set $E$ on a graph $\mathcal G$ is local stochastic of parameter $p$ if, for all set of nodes $G \subseteq \mathcal G$, holds:
\begin{align*}
 \mathbb P(G \subseteq E)\le p^{|G|}.
 \end{align*}
\end{defin}
We then use some results in percolation theory, Lem.~\ref{lem:number_of_patches} and Lem.~\ref{lem:lemma_prob} below, to understand the probability that errors of closeness linear in the patch size (i.e.~$\|E\|_{\beta}=\alpha \beta $ for some $0< \alpha \le 1$) occur when the noise is local stochastic. 
\begin{lemma}[Corollary 28 of \cite{leverrier18}]\label{lem:number_of_patches}
Let $\mathcal G$ be a graph with vertex degree upper bounded by $z$. Then the number $N_{\beta}$ of connected components of size $\beta$ ($\beta$-patches) satisfies 
\begin{align*}
N_{\beta}\le |\mathcal G| \Phi^\beta
\end{align*}
where $\Phi = (z-1)\left(1+\frac{1}{z-2}\right)^{z-2}$.
\end{lemma}
\begin{lemma}\label{lem:lemma_prob}
Let $\mathcal{G}$ be a graph with vertex degree upper bounded by $z$. Let $t$ be a positive integer and $0< \alpha \le 1$. Then there exists $p_{\mathrm{th}} > 0$ such that, for local stochastic errors $E$ of parameter $p < p_{\mathrm{th}}$, we have
\begin{align}
\label{eq:lemma_prob}
\mathbb P (\|E\|_t \ge \alpha t) \le \frac{|\mathcal G|}{1-2^{h(\alpha)/\alpha}p}\left( \frac{p}{p_{\mathrm{th}}}\right)^{\alpha t},   
\end{align}
where $h(\alpha)=\alpha \log_2(\frac{1}{\alpha}) + (1-\alpha) \log_2 \frac{1}{1-\alpha}$ is the binary entropy function.
\end{lemma}
\begin{proof}
The proof is a straightforward adaption of the proof of Theorem 17 in \cite{leverrier18}.
By expanding the left hand side of Eq.~\eqref{eq:lemma_prob}, we find:
\begin{align*}
\mathbb P(\|E\|_{t} \ge \alpha t ) &= \mathbb P (\exists K \, t\text{-patch : } |K \cap E| \ge \alpha t)\\
&\le \sum_{K \text{ is a $t$-patch}}\mathbb P(|K \cap E| \ge \alpha t).
\end{align*}
Observe that, for a $t$-patch $K$, 
\begin{align}
\mathbb P(|K \cap E| \ge \alpha t) &\le \sum_{m \ge \alpha t}\sum_{\substack{K' \subseteq K\\|K'|=m}}\mathbb P\left(K \cap E = K'\right) \nonumber\\
    &\le \sum_{m \ge \alpha t} \sum_{\substack{K' \subseteq K\\|K'|=m}}\mathbb P\left(K' \subseteq E\right)\nonumber\\
    &\le \sum_{m \ge \alpha t}\sum_{\substack{K' \subseteq K\\|K'|=m}}p^m\nonumber\\
    &\le \sum_{m \ge \alpha t}\binom{t}{m}p^m. \label{eq:last_patches}
\end{align}
By Stirling's approximation\footnote{$\binom{n}{k} \simeq 2^{n h(k/n)}$, where $h(x) = x \log_2\frac{1}{x}+ (1-x)\log _2\frac{1}{(1-x)}$ is the binary entropy function.}:
\begin{align}\label{eq:prob_patches}
    \mathbb P\left( |K \cap E| \ge \alpha t\right) \le \frac{(2^{h(\alpha)/\alpha}p)^{\alpha t}}{1-2^{h(\alpha)/\alpha}p}.
\end{align}
Combining Eq.~\eqref{eq:last_patches}, \eqref{eq:prob_patches} and Lem.~\ref{lem:number_of_patches} yields:
\begin{align*}
    \mathbb P\left(\|E\|_t \ge \alpha t \right) &\le N_{t} \frac{(2^{h(\alpha)/\alpha}p)^{\alpha t}}{1-2^{h(\alpha)/\alpha}p}\\
    & \le \frac{|\mathcal G|}{1-2^{h(\alpha)/\alpha}p}\cdot \left(\Phi 2^{h(\alpha)}p^\alpha\right)^{t}
\end{align*}
By imposing the right hand side to decrease with $t$, we find 
\begin{align*}
p \le \left(\frac{1}{\Phi 2^{h(\alpha)}}\right)^{\frac{1}{\alpha}} \coloneqq p_{\mathrm{th}}
\end{align*}
And in conclusion:
\begin{align*}
    \mathbb P (\|E\|_t \ge \alpha t ) \le \frac{|\mathcal G|}{1-2^{h(\alpha)/\alpha}p} \left(\frac{p}{p_{\mathrm{th}}}\right)^{\alpha t}.
\end{align*}
\end{proof}
Finally, we are able to prove  that there exists a threshold under which the probability of local stochastic errors to be non-correctable via the Stochastic Shadow decoder becomes exponentially small in the system size, provided that the graphs $\mathcal G_s$ and $\mathcal G_q$ have bounded degree and linear confinement.
\begin{proof}[Proof of Thm.~\ref{thm:threshold}]
By Lem.~\ref{lem:shadow_sto}, the residual error left by the Stochastic Shadow decoder on a $(t, f)$-confined code is kept under control provided that 
\begin{align}
\label{eq:cond_sto}
\|E\|_{t} \le \frac{t}{4} \quad \text{and}\quad f(2\|S_e\|_{\omega t})\le \frac{t}{4}. 
\end{align}
If the function $f$ is linear, i.e. $f(x)=\kappa x$ for some $\kappa > 0 \in \mathbb Z$, then conditions \eqref{eq:cond_sto} can be written as
\begin{align}
\|E\|_{t} \le \frac{t}{4} \quad \text{and}\quad \|S_e\|_{\omega t} \le \frac{t}{8\kappa}. 
\end{align}
If the qubit error $E$ is local stochastic of parameter $p$ and the syndrome error $S_e$ is local stochastic of parameter $q$, thanks to Lem.~\ref{lem:lemma_prob}, we obtain:
\begin{align*}
\mathbb P (\|E\|_{t} \ge t/4) &\le \frac{|\mathcal G_q|}{1-2^{4h(\frac{1}{4})}p}\left(\frac{p}{p_{\mathrm{th}}}\right)^{\frac{t}{4}}\\
&\quad \coloneqq C_q |\mathcal G_q| \left(\frac{p}{p_{\mathrm{th}}}\right)^{\frac{t}{4}}
 \end{align*}
 and
\begin{align*}
\mathbb P \left(\|S_e\|_{\omega t} \ge \frac{t}{8 \kappa} \right) &\le \frac{|\mathcal G_s|}{1-2^{8 \omega \kappa h(\frac{1}{8 \omega \kappa})}q}\left(\frac{q}{q_{\mathrm{th}}}\right)^{\frac{t}{8 \omega \kappa}}\\
&\quad \coloneqq C_s |\mathcal G_s|\left(\frac{q}{q_{\mathrm{th}}}\right)^{\frac{t}{8 \omega \kappa}}
\end{align*}
where:
\begin{align*}
&p_{\mathrm{th}}\coloneqq \left( \frac{1}{ \Phi_{q} 2^{h(\frac{1}{4})}}\right)^4 \quad \text{ and }\quad q_{\mathrm{th}}\coloneqq \left(\frac{1}{\Phi_{s} 2^{h(\frac{1}{8\omega \kappa})}}\right)^{8 \omega \kappa} 
\end{align*}
As a result, by Lem.~\ref{lem:shadow_sto}, the residual error is correctable except with probability at most 
\begin{align*}
\max \left\{ C_q |\mathcal G_q| \left(\frac{p}{p_{\mathrm{th}}}\right)^{\frac{t}{4}}, \quad C_s |\mathcal G_s|\left(\frac{q}{q_{\mathrm{th}}}\right)^{\frac{t}{8 \omega \kappa}}\right\}.
\end{align*}
In other words, for local stochastic noise of intensity $p \le p_{\mathrm{th}}$ on the qubits and $q \le q_{\mathrm{th}}$ on the syndrome, the Stochastic Shadow decoder has a sustainable single-shot threshold.
\end{proof}

We conclude by noting that the assumption of linear confinement is key in the proof of Thm.~\ref{thm:threshold}. However, we speculate that the limitations of Thm.~\ref{thm:threshold} are an artefact of our proof and \emph{super linear} confinement is a sufficient condition for a family of codes to exhibit a single-shot threshold. In fact, the existence of a threshold $p_{\mathrm{th}}$ and $q_{\mathrm{th}}$ relies on the bounds given in Lem.~\ref{lem:lemma_prob}. There, it is fundamental that the relation between the chosen size of the patch and the size of the overlap with the error is linear (see Eq.\eqref{eq:last_patches} and Eq.\eqref{eq:prob_patches}). In other words, Lem.~\ref{lem:lemma_prob} states that, if we take $\beta$-patches on the error graph and $\gamma$-patches on the syndrome graph, we are able to estimate the probability that errors have closeness less than $\alpha \beta$ and $\tilde \alpha \gamma$ respectively. By \eqref{eq:cond_sto}, in order to bound the closeness of the residual error left by the Stochastic Shadow decoder, we need:
\begin{align*}
\|S_e\|_{\gamma} \le \frac{1}{2}f^{-1}(\alpha\beta).
\end{align*}
As a consequence, combining this with the requirements of Lem.~\ref{lem:lemma_prob}, entails
\begin{align*}
\gamma = \kappa \left(\frac{1}{2}f^{-1}(\alpha\beta)\right),
\end{align*}
for some positive constant $\kappa$.
In conclusion, building up on the results of Lem.~\ref{lem:lemma_prob}, we either need to prove that confinement is preserved if we take on the syndrome graph patches of size linear in $f^{-1}(\alpha\beta)$ or, using our Lem.~\ref{lem:patches} without modification, that the function is itself linear.
\section{Qubit placement on a 3D lattice}
\label{app:3dspace}
Here we detail how to embed a 3D product code on a cubic lattice, where qubits sit on edges, $Z$-stabilisers on vertices, $X$-stabilisers on faces and metachecks on cells.

Let $C^0$ and $C^1$ be two vector spaces over $\mathbb{F}$ with basis $\mathcal{B}^0 = \{e^0_1, \dots, e^0_n\}$ and $\mathcal{B}^1 = \{e^1_1, \dots, e^1_m\}$ respectively. Given a linear map from $C^0$ into $C^1$, it can be represented as a $m \times n$ matrix $\delta$ over $\mathbb F$ such that its action on the elements of the basis $\mathcal B^0$ is given by:
\begin{align}
\label{eq:basis_delta}
    \delta: C^0 &\longrightarrow C^1 \nonumber\\
    e^0_i &\longmapsto \delta e^0_i= \sum_{\alpha = 1}^m \delta_{\alpha, i} e^1_\alpha.
\end{align}
Expression \eqref{eq:basis_delta} allows us to write the support of vectors in $ \delta(\mathcal{B}^0) = \left\{\delta e^0_i\right\}_{i}$ in a compact form. In fact, the support of $\delta e^0_i$ is the subset of $\mathcal{B}^1$:
\begin{align*}
    \supp(\delta e_i^0)= \left\{e_\alpha^1 : \delta_{\alpha, i}\neq 0\right\}_\alpha.
\end{align*}
Since basis vectors are uniquely identified by their index, we can compactly write \eqref{eq:basis_delta} as a relation $*$ on the set of indices of the basis $\mathcal{B}^0$ and $\mathcal{B}^1$:
\begin{align}
\label{eq:indexing_1}
    \{1, \dots, n\} &\longrightarrow \{1, \dots, m\} \nonumber\\
    \kappa &\longrightarrow {\kappa}^*, 
\end{align}
where
\begin{align*}
    {\kappa}^* &= \{\eta : \delta_{\eta, \kappa} \neq  0\}_{\eta}.
\end{align*}    
Similarly, the transpose $\delta^T$ of the matrix $\delta$ induces a map from $C^1$ to $C^0$ which is defined on $\mathcal B^1$ as
\begin{align*}
    \delta^T: C^1 &\longrightarrow C^0 \nonumber\\
    e_\alpha^1 &\longmapsto \delta^T e_\alpha^1 = \sum_{i = 1}^n \delta_{\alpha, i} {e_i^0},
\end{align*}
yields the relation on indices
\begin{align}
\label{eq:indexing_2}
    \{1, \dots, m\} &\longrightarrow \{1, \dots, n\} \nonumber \\
    \eta &\longrightarrow {\eta}^*,
\tag{\ref*{eq:indexing_1} T}    
\end{align}
where
\begin{align*}
    \eta^* &= \{\kappa : \delta_{\eta,\kappa}\neq 0\}_\kappa.
\end{align*}

Referring to the chain complex $\mathcal{C}$ described in Sec.~\ref{sec:3d_product}, we choose bases $\mathcal {B}^{\tau}_{\ell}=\left\{e^{\ell_\tau}_{\iota}\right\}_{\iota}$ of $C^{\tau}_{\ell}$ for $\tau = 0, 1$ and $\ell = A,B,C$. We accordingly fix matrix representations of the maps $\delta_A, \delta_B$ and $\delta_C$; with slight abuse of notation, we indicate with the same symbol the $m_{\ell} \times n_{\ell}$ matrix representation of a map and the map itself.
We indicate with $i, j, k$ indices of $\mathcal{B}^0_A, \mathcal{B}^0_B$ and $\mathcal{B}^0_C$ respectively and with $\alpha, \beta, \gamma$ indices of $\mathcal{B}^1_A, \mathcal{B}^{1}_A, \mathcal{B}^{1}_C$.
Since we deal with 3-fold tensor product spaces (e.g. $C^0_A \ox C^0_B \ox C^0_C$) we consider triplets $(i, j, k)$ of valid indices; we indicate with $(i^*, j, k )$ the set of indices $\{(\eta, j, k) : \eta \in i^* \}$, and similarly for any possible triplet combination of starred ($\iota^*$) and non starred ($\iota$) indices.

As illustrated in Sec.~\ref{sec:3d_product}, when defining a CSS code on the chain complex $\mathcal{C}$, the following relations hold:
\begin{enumerate}
    \item basis elements of $\mathcal{C}_0$ are in one-to-one correspondence with a generating set of $Z$-stabilisers;
    \item basis elements of the vector space $\mathcal{C}_1$ are in one-to-one correspondence with the qubits;
    \item basis elements of the vector space $\mathcal C_2$ are in one-to-one correspondence with a generating set of $X$-stabilisers;
    \item basis elements of $\mathcal C_3$ are in one-to-one correspondence with a generating set of metachecks.
\end{enumerate}
Combining these with \eqref{eq:indexing_1} and \eqref{eq:indexing_2}, we obtain the relations reported in Table \ref{table:indexing}. More precisely, we choose as bases for the spaces $\mathcal{C}_3, \mathcal{C}_2, \mathcal{C}_1$ and $\mathcal{C}_0$ the product bases obtained by combining $\mathcal{B}^0_{\ell=A, B, C}$ and $\mathcal{B}^1_{\ell=A, B, C}$ and we index qubits, stabilisers and metachecks on $\mathcal{C}$ accordingly. Equivalently, basis vectors are labelled with consecutive integers so as to preserve the ordering induced by the bases.
\begin{table}[htb!]
\begin{equation*}
\begin{tabular}{c|c|c}
\textsc{Object } & \textsc{Indexing} &\textsc{Basis vector}\\
\hline
qubits & 
$\begin{aligned} (\alpha, j, k)\\
\\
(i, \beta, k)\\
 \\
(i, j, \gamma)
\end{aligned}$ &
$\begin{aligned}
\begin{pmatrix}e^{A_1}_{\alpha} \ox e^{B_0}_{j} \ox e^{C_0}_{k}, &0, &0\end{pmatrix}\\
\\
\begin{pmatrix}0,& e^{A_0}_{i} \ox e^{B_1}_{\beta} \ox e_{k}^{C_0},& 0\end{pmatrix}\\
\\
\begin{pmatrix}0,& 0,& e_{i}^{A_0}\ox e_{j}^{B_0} \ox e_{\gamma}^{C_1}\end{pmatrix}
\end{aligned}$
\\
\hline
$X$-stabilisers & 
$\begin{aligned}
(\alpha, \beta, k)\\
\\
(\alpha, j, \gamma)\\
\\
(i, \beta, \gamma)
\end{aligned}$
& 
$\begin{aligned}
\delta_2^T \begin{pmatrix}
e_{\alpha}^{A_1} \ox e_{\beta}^{B_1} \ox e_{k}^{C_0}, & 0, & 0\end{pmatrix}
\\
\\
\delta_2^T \begin{pmatrix}
0, & e_{\alpha}^{A_1} \ox e_{j}^{B_0} \ox e_{\gamma}^{C_1}, &0
\end{pmatrix}
\\
\\
\delta_2^T \begin{pmatrix}
0, & 0, & e_{i}^{A_0} \ox e_{\beta}^{B_1} \ox e_{\gamma}^{C_1}\end{pmatrix}
\end{aligned}$\\
\hline
$Z$-stabilisers & $(i, j, k)$
&
$\delta_1(e_{i}^{A_0} \ox e_{j}^{B_0} \ox e_{k}^{C_0})$
\\
\hline
metacheck & $(\alpha, \beta, \gamma)$ &
$\delta_3^T(e_{\alpha}^{A_1} \ox e_{\beta}^{B_1}\ox e_{\gamma}^{C_1})$
\\
\hline
\end{tabular}
\end{equation*}
\caption{Notation and correspondences between objects of the chain complex $\mathcal C$.}
\label{table:indexing}
\end{table} 

We use the relations of Table \ref{table:indexing} to visualise the chain complex $\mathcal C$ on a 3D cubic lattice. In order to do so, we first fix a coordinate system
\begin{center}
\begin{tikzpicture}
\node[anchor=east] at (0,1,0) {$O$};
\draw[thick,->] (0,1,0) -- (1,1,0)
 node[anchor=north west]{$x$};
\draw[thick,->] (0,1,0) -- (0,0,0) node[anchor=north east]{$y$};
\draw[thick,->] (0,1,0) -- (0,1,-1) node[anchor=south west]{$z$};
\end{tikzpicture}
\end{center}
where $O$ is the origin. Since basis vectors are labelled with integers (the $i$th basis vector corresponds to the integer $i$, and vice versa) we can build a 3D grid of points where any point corresponds to a basis vector of $\mathcal{C}_0, \mathcal{C}_1, \mathcal{C}_2$ or $\mathcal{C}_3$. More precisely we fix a set of valid coordinates for points in the grid:
\begin{enumerate}
    \item integer coordinates $(z, y, x) = (i, j,k)$ for $i = 1, \dots, n_a$, $j = 1, \dots, n_b$, and $k = 1, \dots, n_c$;
    \item half-integers coordinates $(z, y, x) = (\alpha + 0.5, \,\beta +0.5,\, \gamma + 0.5)$ for $\alpha = 1, \dots, m_a$, $\beta = 1, \dots, m_b$ and $\gamma = 1, \dots, m_c$;
    \item the origin has coordinates $O = (1, 1, 1)$.
\end{enumerate} 
In this way, any point with valid coordinates uniquely identifies a basis vector (and therefore an object in the chain complex, see Table~\ref{table:indexing}). For example:
\begin{enumerate}
    \item the point $(1, 4, 2)$ corresponds to the basis vector $(e^{A_0}_{1}\ox e^{B_0}_{4} \ox e^{C_0}_{2}) \in  \mathcal{C}_0$ ($Z$-stabilisers);
    \item the point $(1.5, 4, 2)$ corresponds to the basis vector $(e^{A_1}_{1}\ox e^{B_0}_{4}\ox e^{C_0}_{2},\, 0,\, 0) \in \mathcal{C}_1$ (qubits);
    \item the point $(1.5, 4, 2.5)$ corresponds to the basis vector $(0,\, e^{A_1}_{1}\ox e^{B_0}_{4}\ox e^{C_1}_{2},\, 0) \in \mathcal{C}_2$ ($X$-stabilisers);
        \item the point $(1.5, 4.5, 2.5)$ corresponds to the basis vector $(e^{A_1}_{1}\ox e^{B_1}_{4}\ox e^{C_1}_{2}) \in \mathcal{C}_3$ (metachecks);

\end{enumerate}

We draw an edge for any qubit of the code defined on $\mathcal C$. Qubits are in one-to-one correspondence with basis element of $\mathcal C_1$ and therefore are of three different types: $(v,\, 0, \, 0)$, $(0,\, v,\, 0)$ and $(0,\, 0,\, v)$. Accordingly, we draw edges of three different types as detailed in Table \ref{table:qubit} (see also Fig.~\ref{fig:lattice_1}). In other words, any point with two integer entries and one half integer entry is the middle point of an edge of unit length, which corresponds to a qubit. 
In this way we obtain a cubic lattice with (possibly) some missing edges. 

\begin{table}[htb!]
\begin{tabular}{C|C}
\textsc{Qubit}    & \textsc{Edge}\\
\hline
\text{transverse qubits}&\text{edges parallel to the
$z$ axis}\\
(\alpha, j, k)&
\text{middle point: } (\alpha+0.5, j, k )\\
\hline
\text{vertical qubits} &\text{edges parallel to the $y$ axis}\\
(i, \beta, k)& 
\text{middle point: } 
(i, \beta + 0.5, k)\\
\hline
\text{horizontal qubits}&\text{edges parallel to the
$x$ axis}\\
(i, j, \gamma)&
\text{middle point: } (i, j, \gamma + 0.5) 

\end{tabular}
\caption{Correspondence between qubits in $\mathcal C_1$ and edges of the lattice.}
\label{table:qubit}
\end{table}

Points with two half-integer and one integer entries do not intersect any edge and sit in the center of a (possibly incomplete) square face. These points correspond to $X$-stabilisers which we therefore identify with faces. Given a triplet corresponding to one of such a point, the associated $X$-stabiliser has support contained in the set of edges which are parallel to the edges of the square, forming a cross in a plane. $X$-stabilisers, like qubits, are of three different types, being in one-to-one correspondence with basis elements of $\mathcal{C}_2$. Namely, each $X$-stabiliser in $\mathcal C_2$ has support in two out of three type of qubits: transverse-vertical, transverse-horizontal or vertical-horizontal (see Table~\ref{table:correspondence_lattice} and Fig.~\ref{fig:cross2d}).

Points with integer coordinates are associated to $Z$-stabilisers; these are points where endpoints of edges intersect. The $Z$-stabiliser corresponding to $(i, j, k)$ has support on a 3D cross of edges/qubits centered in $(z, y, x) = (i, j, k)$ (see Table~\ref{table:correspondence_lattice} and Fig.~\ref{fig:cross3d}).

Points with half-integer coordinates sit in the center of a cube. To any such cube is associated a metacheck in $\mathcal{C}_3$. Metachecks have support on a 3D cross of faces/$X$-stabilisers parallel to the faces of the cube they are associated to (see Table~\ref{table:correspondence_lattice}).

\begin{figure}
\subfloat[]{\includegraphics[scale=0.35]{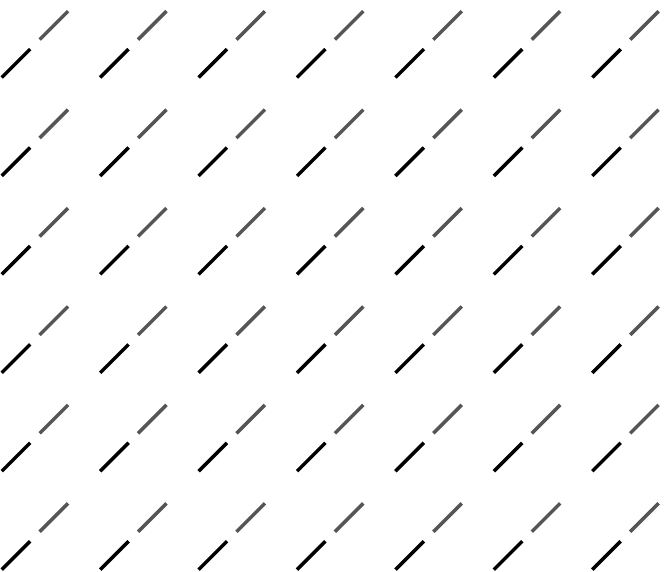}}
\hfill
\subfloat[]{\includegraphics[scale=0.35]{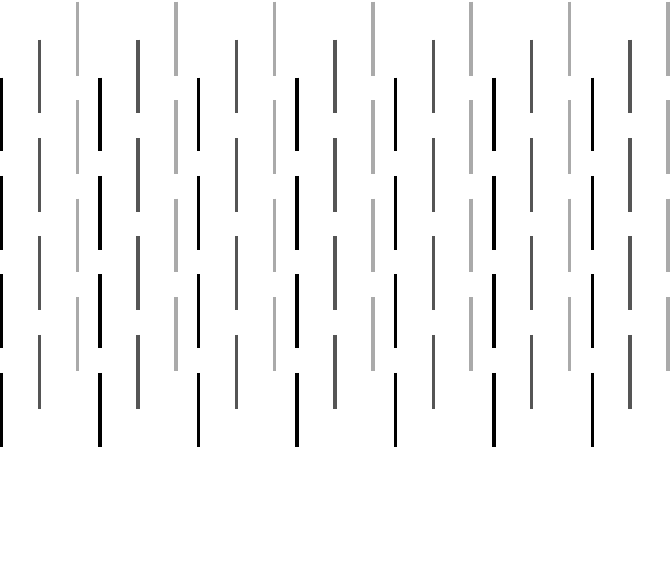}}
\hfill
\subfloat[]{\includegraphics[scale=0.35]{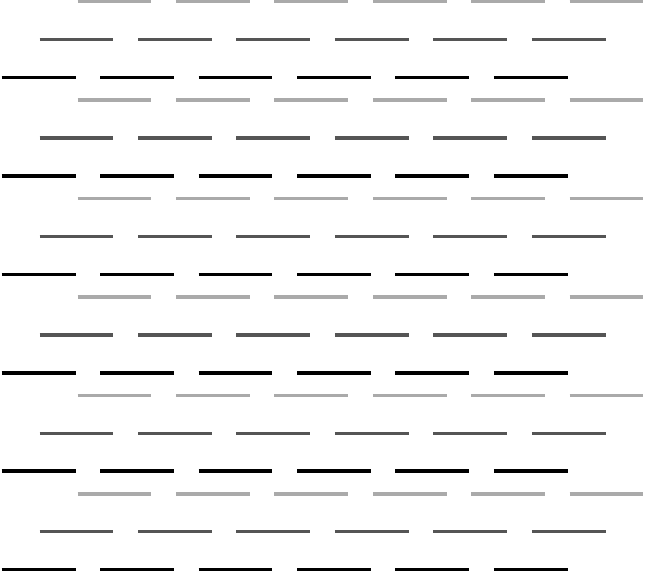}}
\hfill
\subfloat[]{\includegraphics[scale=0.35]{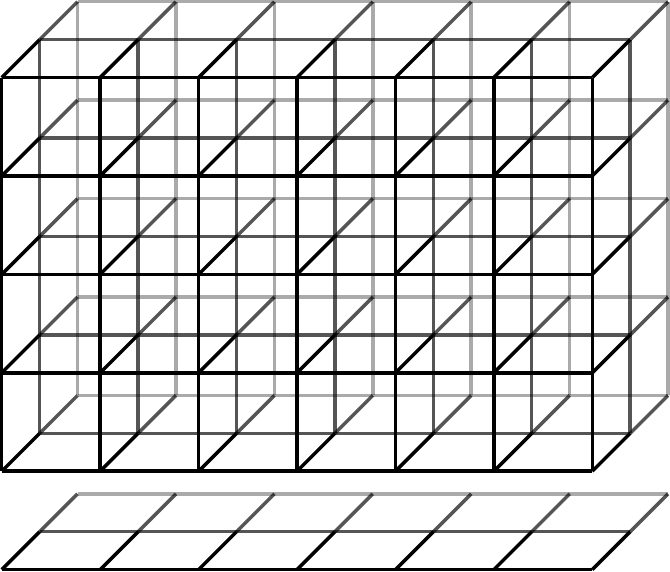}}
\hfill
\caption{Graphical representation of the cubic lattice associated to a 3D product code where the seed matrices $\delta_A$, $\delta_B$, $\delta_C$ have size $2\times 3$, $4 \times 6$, and $6 \times 7$ respectively. In (a), (b) and (c) only transversal, vertical and horizontal edges are depicted. In (d) we can see the complete lattice obtained by matching the origin $O=(1,1,1)$ of the three lattices of edges.  
}
\label{fig:lattice_1}
\end{figure}
\begin{figure}
\subfloat[]{\includegraphics[scale=0.35]{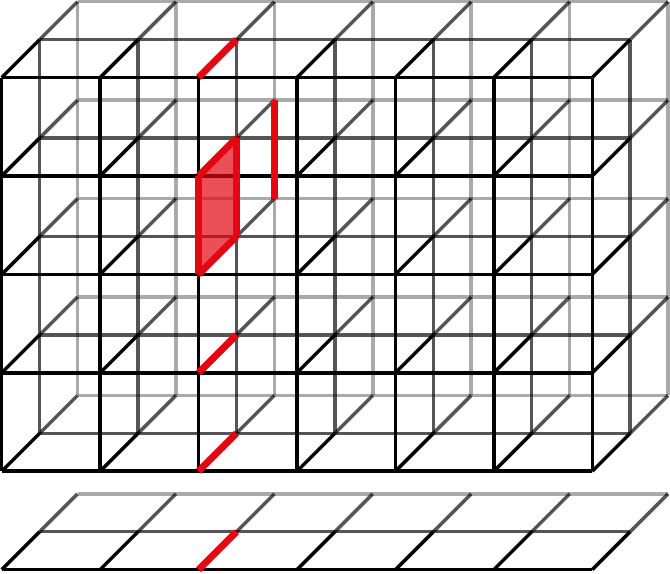}
}
\hfill
\subfloat[]{\includegraphics[scale=0.35]{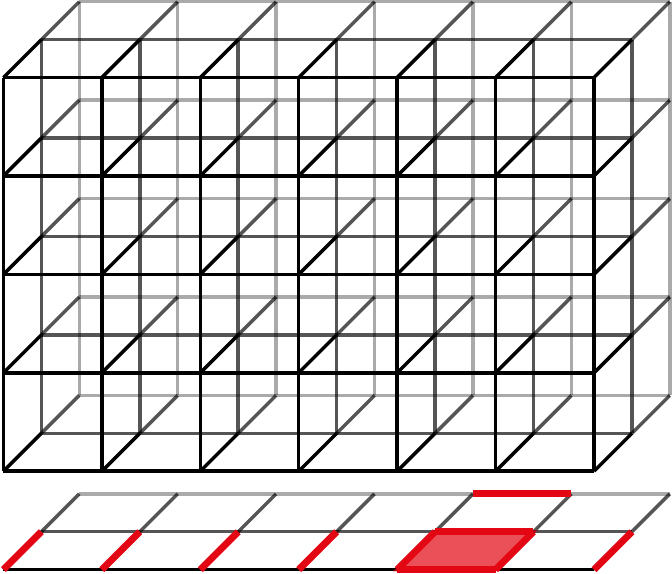}}
\hfill
\subfloat[]{\includegraphics[scale=0.35]{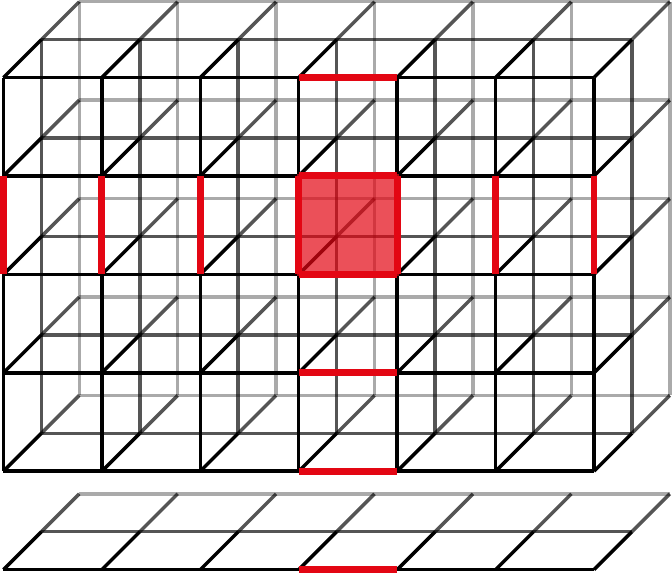}}
\caption{
$X$-stabilisers on the lattice described in Fig.~\ref{fig:lattice_1}.
(a) $X$-stabiliser corresponding to the transversal-vertical square indexed by $(\alpha, \beta, k)=(1, 2, 3)$; its support is contained in the cross of transversal and vertical qubits (red edges) in the $yz$-plane $\{x=3\}$. The crossing has coordinates $(z, y, x) = (1.5, 2.5, 3)$ and sits in the center of the red square. (b) $X$-stabiliser corresponding to the transversal-horizontal square indexed by $(\alpha, j, \gamma)=(1, 6, 5)$; its support is contained in the cross of transversal and vertical qubits (red edges) in the $xz$-plane $\{y=6\}$. The crossing has coordinates $(z, y, x) = (1.5, 6, 5.5)$ and sits in the center of the red square.
(c) $X$-stabiliser corresponding to the vertical-horizontal square indexed by $(i, \beta, \gamma)=(1, 4, 2)$; its support is contained in the cross of transversal and vertical qubits (red edges) in the $xy$-plane $\{z=1\}$. The crossing has coordinates $(z, y, x) = (1, 2.5, 4.5)$ and sits in the center of the red square.
}
\label{fig:cross2d}
\end{figure}
\begin{figure}
\includegraphics[scale=0.35]{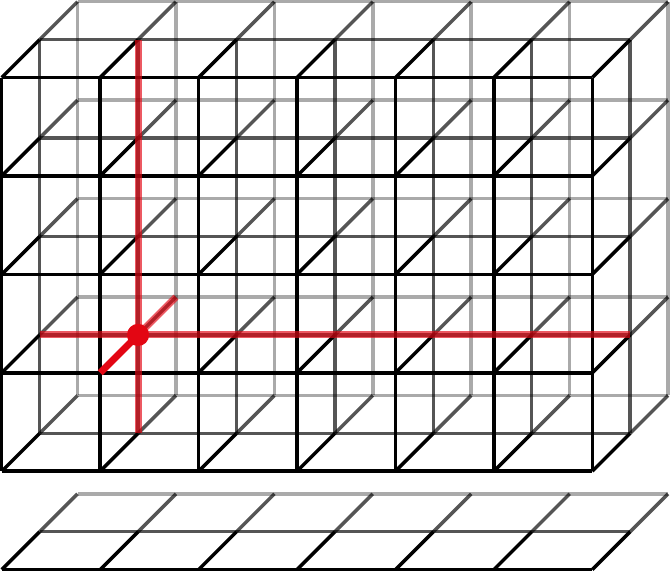}
\includegraphics[scale=0.35]{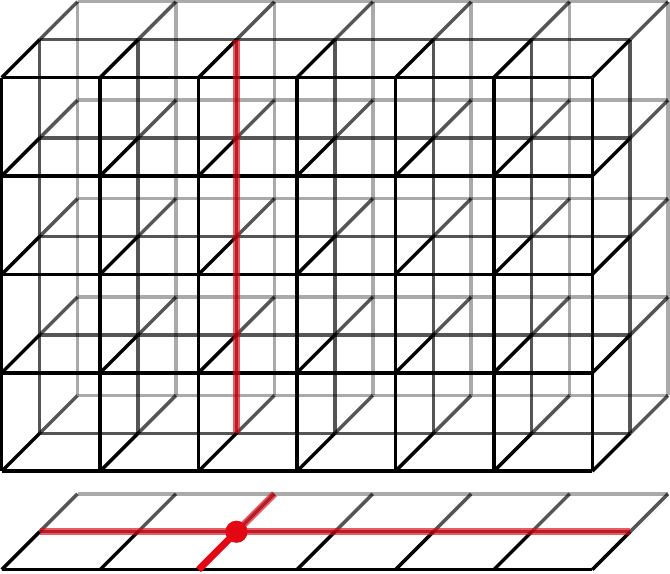}
\caption{
$Z$-stabilisers on the lattice described in Fig.~\ref{fig:lattice_1}. (a) $Z$-stabiliser corresponding to the vertex indexed by $(i, j, k)=(2, 4, 2)$; its support is contained in the cross of qubits highlighted as red edges in the picture. The crossing has coordinates $(z, y, x) = (2, 4, 2)$ (red circle). (b) $Z$-stabiliser corresponding to the vertex indexed by $(i, j, k)=(3, 6, 2)$; its support is contained in the cross of qubits highlighted as red edges in the picture. The crossing has coordinates $(z, y, x) = (3, 6, 2)$ (red circle).}
\label{fig:cross3d}
\end{figure}

\begin{table*}[htb!]
\begin{tabular}{c|C|R}
\textsc{Operator} &\textsc{Type} & \textsc{Support}\\
\hline
\multirow{6}{*}{$X$-stabilisers}&
\text{transverse-vertical square}&\text{transverse qubits: } (\alpha, \beta^*, k)\\
&(\alpha, \beta, k)&
\text{vertical qubits: } (\alpha^*, \beta, k)\\
\cline{2-3}

&
\text{transverse-horizontal square}&\text{transverse qubits: } (\alpha, j, \gamma^*)\\
&(\alpha, j, \gamma)
&\text{horizontal qubits: } (\alpha^*, j, \gamma)\\

\cline{2-3}

& \text{vertical-horizontal square}&\text{vertical qubits: } (i, \beta, \gamma^*)\\
&(i, \beta, \gamma)& \text{horizontal qubits: } (i, \beta^*, \gamma)\\
\hline
\multirow{3}{*}{$Z$-stabilisers}
&\multirow{3}{*}{$(i, j, k)$} &
\text{transverse qubits: } (i^*, j, k)\\
&& \text{vertical qubits: } (i, j^*, k)\\
&&\text{horizontal qubits: } (i, j, k^*)\\
\hline
\multirow{3}{*}{metachecks}&  \multirow{3}{*}{$(\alpha, \beta, \gamma)$}&
       \text{transverse-vertical faces: } (\alpha, \beta, \gamma^*)\\
&&\text{transverse-horizontal faces: } (\alpha, \beta^*, \gamma)\\
&&\text{vertical-horizontal faces: } (\alpha^*, \beta, \gamma)
\end{tabular}
\caption{Correspondence between operators of the chain complex $\mathcal C$, their type as geometric objects on the lattice, and their support. Note that the support of $X$ and $Z$ stabilizers is a set of qubits/edges while the support of metachecks is a set of $X$-stabilisers/faces.
\label{table:correspondence_lattice}}
\end{table*}

\subsection{On geometric locality}

\label{app:3dspace_locality}
One interesting feature of the embedding of 3D product codes on a cubic lattice is that it preserves some locality properties of the seed matrices $\delta_A, \delta_B$ and $\delta_C$. Thus, if we were able to place qubits on a 3D cubic lattice we could use the 3D homological product to build LDPC codes with nearest neighbour interactions.

Let $\delta$ be an $m\times n$ matrix with row/column indices $\alpha \in \{1, \dots, m\}$ and $i \in \{1, \dots, n\}$ respectively and let $\nu = \max\{m, n\}$.
We say that $\delta$ is \textit{geometrically $\rho$-local on a torus} if for any row and any column index:
\begin{align}
\label{eq:geom_local_torus}
    \alpha^* \subseteq U_{\rho,\, \nu}(\alpha) \quad\text{ and }\quad i^* \subseteq U_{\rho,\,\nu}(i),
\end{align}
where $U_{\rho, \,\nu}(\zeta)$ is any set of $\rho$ consecutive integers modulo $\nu$ which contains $\zeta$. In particular, we require the $\alpha$th rows to have support on columns with index that is \textit{close} to the integer $\alpha$, and similar for columns. The reason for this  choice will be clear when we prove Prop.~\ref{prop:locality}. Briefly, conditions \eqref{eq:geom_local_torus} says that $\delta$ is geometrically $\rho$-local on a torus if: (1) any of its rows has support on a bounded box of $\rho$-columns, and the box for row $\alpha+1$ is a right shift of the box for row $\alpha$; (2) any of its columns has support on a bounded box of $\rho$ rows, and the box for column $i +1 $ is a downward shift of the box for column $i$. In particular, if we associate row/column indices with integer points on a circle of $\nu$ points:
\begin{center}
\begin{tikzpicture}[scale=0.5, transform shape]
\foreach \ang\lab\anch in {90/1/north, 45/2/{north east}, 0/3/east, 270/ /south, 180/{
\nu-1}/west, 135/\nu/{north west}}{
  \draw[fill=black] ($(0,0)+(\ang:3)$) circle (.08);
  \node[anchor=\anch] at ($(0,0)+(\ang:2.8)$) {$\lab$};
}
\foreach \ang\lab in {90/1,45/2,180/{n-1},135/n}{
  \draw[-,shorten <=7pt, shorten >=7pt] ($(0,0)+(\ang:3)$) arc (\ang:\ang-45:3);
}
\draw[-,shorten <=7pt] ($(0,0)+(0:3)$) arc (360:325:3);
\draw[-,shorten >=7pt] ($(0,0)+(305:3)$) arc (305:270:3);
\draw[-,shorten <=7pt] ($(0,0)+(270:3)$) arc (270:235:3);
\draw[-,shorten >=7pt] ($(0,0)+(215:3)$) arc (215:180:3);
\foreach \ang in {310,315,320,220,225,230}{
  \draw[fill=black] ($(0,0)+(\ang:3)$) circle (.02);}
\end{tikzpicture}
\end{center}
locality means that any set $\alpha^*/i^*$ is contained in a closed interval on the circle such that (i) it has length at most $\rho$ and (ii) it contains the point $\alpha/i$. 
For instance, the degenerate parity check matrix of the repetition code:
\begin{equation*}
\begin{pmatrix}
1 & 1 & 0 & 0 & 0\\
0 & 1 & 1 & 0 & 0\\
0 & 0 & 1 & 1 & 0\\
0 & 0 & 0 & 1 & 1\\
1 & 0 & 0 & 0 & 1\\
\end{pmatrix}
\end{equation*}
is $\rho$-local for $\rho = 2$. 

A closely related notion of locality on a torus is geometric locality in Euclidean space. We say that an $m \times n$ matrix is \textit{geometrically $\rho$-local in Euclidean space} if for any row and column index:
\begin{align}
\label{eq:geom_local_eu}
    \alpha^* \subseteq U_{\rho}(\alpha) \quad\text{ and }\quad i^* \subseteq U_{\rho}(i),
\end{align}
where $U_{\rho}(\zeta)$ is any set of $\rho$ consecutive integer in $[1, \dots, \nu]$, $\nu=\max\{m, n\}$, which contains $\zeta$.
In this case we can graphically picture locality by associating row/column indices with integer points on a line of $\nu$ points:
\begin{center}
\begin{tikzpicture}[scale=1, transform shape]
\draw[-] (1, 0) -- (3, 0); 
\draw[dotted,thick] (3, 0) -- (6, 0); 
\foreach \x [count = \xi] in {1,2,3, , {\nu-1}, \nu}{
\node at (\xi, +0.5) {$\x$};
\draw [fill] (\xi, 0) circle [radius=0.05];
};
\end{tikzpicture}
\end{center}
A matrix is local if any set $\alpha^*/i^*$ is contained in a closed interval on the line such that (i) it has length at most $\rho$ and (ii) it contains the point $\alpha/i$. For example, the full-rank parity check matrix of the repetition code:
\begin{equation*}
\begin{pmatrix}
1 & 1 & 0 & 0 & 0\\
0 & 1 & 1 & 0 & 0\\
0 & 0 & 1 & 1 & 0\\
\end{pmatrix}
\end{equation*}
is $2$-local.

Geometric locality also applies to codes other than the repetition code. For instance, the matrix 
\begin{equation*}
    H = \begin{pmatrix}
1 & 1 & 0 & 0 & 0 & 0 & 0 & 0 & 0\\
0 & 1 & 0 & 0 &0 & 0 & 1 & 0 & 0\\
0 & 0 & 1 & 1 &0 & 0 & 0 & 0 & 0\\
0 & 0 & 0 & 1 &0 & 0 & 0 & 1 & 0\\
0 & 0 & 0 & 0 &1 & 1 & 0 & 0 & 0\\
0 & 0 & 0 & 0 &0 & 1 & 0 & 0 & 1\\
1 & 0 & 1 & 0 &1 & 0 & 0 & 0 & 0\\
    \end{pmatrix},
\end{equation*}
obtained via the edge augmentation procedure presented in \cite{roffe2020decoding} is $7$-local on a torus. We remark that geometric locality is a property of matrices. For example, the matrix with same row as $H$ but different ordering $\{1, 2, 3, 4, 7, 5, 6\}$, is geometrically $5$-local on torus. 

In general, geometric locality is a relaxation of the locality property of the repetition code which only allows for interactions between pairs of nearest bits. Importantly, as Prop.~\ref{prop:locality} states, it is preserved by the 3D product construction. For this reason, geometrically local classical codes, combined with the 3D product construction, could be good candidates in the quest to quantum local codes beyond the toric and the surface codes. 

The remainder of this Appendix is organised as follows. We first state Prop.~\ref{prop:locality} and prove that geometric locality is preserved by the 3D product construction. We conclude by observing how this proof provides an explicit identification of the 3D toric and surface codes as 3D product codes.  

To ease the notation, in the following we will shortly refer to codes as geometrically local, dropping the specification on a torus/in Euclidean space. When considering qubits on a cubic lattice, the lattice would be on a torus or in Euclidean space depending on the definition of locality that applies to the seed matrices.
\begin{proposition}
\label{prop:locality}
Consider the 3D product code obtained from three seed matrices geometrically $\rho$-local.
If its qubits are displayed on the edges of a cubic lattice as detailed in Sec.~\ref{app:3dspace}, then it is geometrically $\rho$-local in the following sense:
\begin{enumerate}
\item any $X$-stabiliser generator has weight at most $2\rho$ with support contained in a 2D box of size $\rho \times \rho$,
\item any $Z$-stabiliser generator has weight at most $3\rho$ with support contained in a 3D box of size $\rho \times \rho \times \rho$.

    \end{enumerate}
\end{proposition}
\begin{proof}
We prove the condition on the $Z$-stabilisers, the proof for the $X$-stabiliser being similar.

Let $S_z$ be a $Z$-stabiliser generator. As reported in Table~\ref{table:indexing} and \ref{table:correspondence_lattice}, it is the image of a basis vector $(e^{A_0}_{i}\ox e^{B_0}_{j} \ox e^{C_0}_{k}) \in \mathcal{C}_0$ via the map $\delta_1$ and it corresponds to the point on the lattice of integers coordinates $(i, j, k)$. By exploiting the choice of the basis for the spaces $\mathcal{C}_0$ and $\mathcal{C}_1$ and some linear algebra :
\begin{align*}
\delta_1 (e^{A_0}_i\ox e^{B_0}_{j}\ox e^{C_0}_{k})&=\sum_{ \alpha \in i^*} (e^{A_1}_{\alpha} \ox e^{B_0}_{j} \ox e^{C_0}_{k}, \, 0, \,0 ) \\
& + \sum_{ \beta \in j^*} (0, \, e^{A_0}_{i} \ox e^{B_1}_{\beta} \ox e^{C_0}_k,\,0)\\
& + \sum_{ \gamma \in k^*} (0, \,0,\, e^{A_0}_{i} \ox e^{B_0}_{j} \ox e^{C_1}_{\gamma}).
\end{align*}
Again using Table~\ref{table:indexing}, the set of indices which corresponds to this sum of basis vectors of $\mathcal{C}_1$ can be written as:
\begin{align*}
\mathrm{indices(S_z)} =     \quad\{(\alpha, j, k) : \alpha \in i^*\}&\\
\cup\, \{(i, \beta, k) : \beta \in j^*\} &\\
\cup\, \{(i, j, \gamma) : \gamma \in k^*\}&.
\end{align*}
Following the nomenclature of qubits as traversal, vertical and horizontal, we see that the three components of the support of $S_z$ given above respect this division and therefore we can write:
\begin{align*}
    \mathrm{indices}(S_z)= \mathrm{indices}(S_z)_{t}\cup \mathrm{indices}(S_z)_{v} \cup \mathrm{indices}(S_z)_{h}.
\end{align*}
Using \eqref{eq:geom_local_eu} (or \eqref{eq:geom_local_torus}), we see that the sets $\mathrm{indices}(S_z)_t$, $\mathrm{indices}(S_z)_v$ and $\mathrm{indices}(S_z)_h$ correspond respectively to the three sets of consecutive coordinates on the lattice:
\begin{align*}
\Pi_t &= \{(\bar i, j, k)   : \bar i \in U_{\rho}(i) \},  \\
\Pi_v &=\{(i, \bar j, k)   : \bar j \in U_{\rho}(j)\},\\ \Pi_h &= \{(i, j, \bar k)   : \bar k \in U_{\rho}(k)\}.
\end{align*}
Since we required $\zeta \in U_{\rho}(\zeta)$ (or $\zeta \in U_{\rho, \nu}(\zeta)$), the three sets of coordinates intersect on the point $(z, y, x) = (i, j, k)$. Moreover, all three intervals $\Pi_t, \Pi_v, \Pi_h$ have length at most $\rho$. Combining these, we find that the support of $S_z$ indexed by $(i, j, k)$ is contained in in a $\rho \times \rho \times \rho$ neighborhood of the point $(z, y, x) = (i, j, k)$ and has cardinality at most $3\rho$. In other words, we have showed that the support of $S_z$ is contained on a 3D cross of qubits with arms of length at most $\rho$.
\end{proof}

As previously said, the 3D toric and planar codes are particular instances of the 3D product construction. Furthermore, it is well known that they are local on a torus and in the Euclidean space respectively. To see how this is the case, we remind the reader that the 3D toric code is obtained by choosing as seed matrices the degenerate parity check matrix of the repetition code in the standard basis. For matrix size $L \times L$, it holds that
\begin{align*}
    \{1, \dots, L\} &\longleftrightarrow \{1, \dots, L\}\\
    i &\longrightarrow \{i, {i+ 1 \bmod{L}}\}\\
    \{\alpha, {\alpha+1 \bmod{L}}\} &\longleftarrow \alpha
\end{align*}
Therefore, stabilisers have support on pairs of consecutive edges, and it is straightforward to see that they have the usual shape:
\begin{enumerate}
    \item $Z$-stabilisers have support on edges incident to a vertex;
    \item $X$-stabilisers have support on edges on the boundary of a square face;
    \item metachecks have support on the faces of a cube.
\end{enumerate}
A similar argument holds for the 3D surface code, which is local in Euclidean space.


\section{All 3D product codes have X-confinement}
\label{app:3dconfinement}
In this Section we prove Thm.~\ref{thm:3d} which states that all 3D product codes have $X$-confinement. Our proof follows the proof of soundness for 4D codes given in \cite{campbell2019theory} with some minor adaptions and it is here reported for completeness.

First, we show that an opportunely chosen length-2 chain complex has confined maps. Secondly, we explain how to use this chain complex as a building block of the length-3 chain complex $\mathcal C$ described in Sec.~\ref{sec:3d_product}. Lastly, we prove that the confinement property is preserved and thus 3D codes defined on $\mathcal{C}$ as explained in Sec.~\ref{sec:3d_product} have $X$-confinement.

Let $\delta_A : C_A^0 \rightarrow C_A^1$ and $\delta_B: C_B^0 \rightarrow C_B^1$ be two length-1 chain complexes. We consider the length-2  product complex $\tilde{\mathcal{C}}$ defined as:~\\
\begin{tikzpicture}[transform canvas={scale=0.95}]
\begin{tikzcd}
& C_A^1 \ox C_B^1 \ar[from=2-1]   \ar[from=2-3, swap]&\\
C_A^1 \ox C_B^0 && C_A^0 \ox C_B^1 \\
& C_A^0 \ox C_B^0 \ar[to=2-1] \ar[to=2-3, swap]&
\end{tikzcd}\hspace{0.5 cm}
\begin{tikzcd}
\mathcal{\tilde C}_2\\
\mathcal{\tilde C}_1 \ar[to=1-1, swap, "\tilde \delta_{1}"]\\
\mathcal{\tilde C}_0 \ar[to=2-1, swap, "\tilde \delta_{0}"]\\
\end{tikzcd}
\end{tikzpicture}
where
\begin{align*}
\tilde \delta_0 &= \begin{pmatrix}\delta_A \ox \1& \1 \ox \delta_B\end{pmatrix},\\
\tilde \delta_1 &= \begin{pmatrix}\1 \ox \delta_B\\\delta_A \ox \1\end{pmatrix};
\end{align*}
We first show that the map $\tilde \delta_0$ has confinement. 
\begin{lemma}\label{lemma:confinement_1}
$\tilde \delta_0$ has $(t, f)$-confinement where $t = \min\{d_A, d_B\}$ and $f(x) = x^2/4$.
\end{lemma}
In order to prove Lem.~\ref{lemma:confinement_1} we first introduce some useful notation. When considering vectors $v$ in a two-fold tensor product space $\mathbb F^{n_1} \ox \mathbb F^{n_2}$ it can be handy to consider their reshaping, which is $n_1 \times n_2 $ matrix on $\mathbb F$. Namely, fixed bases $\mathcal{B}^1 = \{a_1, \dots, a_{n_1}\}$ and $\mathcal{B}^2=\{b_1, \dots, b_{n_2}\}$ of $\mathbb F^{n_1}$ and $\mathbb F^{n_2}$ respectively, their product 
\begin{align*}
  \mathcal{B} = \{a_i \ox b_j\}_{\substack{i = 1, \dots, n_1\\
j = 1, \dots, n_2}}  
\end{align*}
is a basis of $\mathbb{F}^{n_1} \ox \mathbb{F}^{n_2}$. Therefore, we can write 
\begin{align}
\label{eq:reshaping_1}
    v = \sum_{a_i \ox b_j \in \mathcal B} v_{ij} a_i \ox b_j
\end{align}
for some $v_{ij} \in \mathbb F$. We call the matrix $V$ whose entries are the coefficient $v_{ij}$ the reshaping of $v$. Given matrices $M$ and $N$ of size $n_1 \times m_1$ and $n_2 \times m_2$ associated to linear maps from $\mathbb F^{n_1}$ and $\mathbb F^{n_2}$ respectively, the map $M \ox N$ from $\mathbb F^{n_1} \ox \mathbb F^{n_2}$ to $\mathbb{F}^{m_1}\ox \mathbb{F}^{m_2}$ acts on the reshaping of $v$ as:
\begin{align}
\label{eq:reshaping_2}
    (M \ox N) V \longmapsto M V N^T.
\end{align}
In the following we will always indicate with lower-case symbol vectors and with the corresponding upper-case symbols their reshaping.
We can now use this notation to prove Lem.~\ref{lemma:confinement_1}. 
\begin{proof}
Let $v \in C_0^A \ox C_0^B$ and let $s = \tilde \delta_0(v)$. By reshaping, 
\begin{align*}
S= \begin{pmatrix}
\delta_A V\\
\\
V \delta_B^T
\end{pmatrix}.
\end{align*}
If we assume $|v| = |v|^{\mathrm{red}} \le t = \min\{d_A, d_B\}$ then $V$ has no column in $\ker \delta_A$ and no row in $\ker \delta_B^T$ so that
\begin{equation*}
    \mathrm{col}(\delta_A V) = \mathrm{col}(V) \text{ and } \mathrm{row}(V\delta_B^T) = \mathrm{row}(V),
\end{equation*} 
where $\mathrm{col}(V)/\mathrm{row}(V)$ is the number of non-zero columns/rows of the matrix $V$. Therefore, for the weight of $S$, it holds that:
\begin{align*}
|S| = |\delta_A V | + |V \delta_B^T| &\ge \mathrm{col}(\delta_A V ) + \mathrm{row}(V \delta_B^T)\\
&=\mathrm{col}(V) + \mathrm{row}(V).
\end{align*} 
Combining this with $(a+b)^2/4 \ge ab$ for integers $a, b$ yields:
\begin{equation*}\label{eq:confinement_square}
|S|^2/4 \ge \mathrm{col}(V)\cdot \mathrm{row}(V) \ge |V|.
\end{equation*}
\end{proof}
We want to use Lem.~\ref{lemma:confinement_1} to infer that the code defined on the chain complex $\mathcal{C}$ has $X$-confinement. To see how this is the case, we consider an `asymmetrical' version of $\mathcal C$ as the product of the length-2 chain complex $\tilde{\mathcal C}$ and the length-1 chain complex $\delta_C : C_0^C \rightarrow C_1^C$.
The asymmetric product complex $\breve{\mathcal C}$ is then ~\\
\begin{tikzpicture}
\begin{tikzcd}
& \tilde C_2 \ox C^1_C\ar[from=2-1]  \ar[from=2-3, swap] &\\
\tilde C_1 \ox C^1_C & & \tilde C_2\ox C_0^C\\
\tilde C_0 \ox C^1_C \ar[to=2-1] & & \tilde C_1 \ox C_0^C \ar[to=2-1] \ar[to=2-3]\\
 &\tilde C_0 \ox C_C^0  \ar[to=3-1] \ar[to=3-3, swap]&
\end{tikzcd}\hspace{0.5 cm}
\begin{tikzcd}
\mathcal{\breve C}_3\\
\mathcal{\breve C}_2 \ar[to=1-1, swap, "\breve\delta_{2}"]\\
\mathcal{\breve C}_1 \ar[to=2-1, swap, "\breve\delta_{1}"]\\
\mathcal{\breve C}_{0}\ar[to=3-1, swap, "\breve\delta_{0}"]
\end{tikzcd}
\end{tikzpicture}
where
\begin{align*}
\breve \delta_{0}&=\begin{pmatrix}
\1 \ox \delta_C \\
\tilde \delta_0 \ox \1
\end{pmatrix},  \\ &\\
\breve \delta_1 &= \begin{pmatrix}
\tilde \delta_0 \ox \1 & \1 \ox \delta_C\\
0 & \tilde \delta_1 \ox \1
\end{pmatrix}, \\ &\\
\breve\delta_2 &= \begin{pmatrix}
\tilde \delta_1\ox \1 & \1 \ox \delta_C
\end{pmatrix}.
\end{align*}
\begin{claim}\label{claim:m}
Let $(v, w)\in \breve{\mathcal C_1}$ have weight less than $t$ and $s = \breve{\delta}_1(v, w)$ be its syndrome . If $(V, W)$ is the reshaping of the vector $(v, w)$ then the following Syndrome Equation holds: 
\begin{align}\label{eq:syndrome_equation}\tag{SE}
S = \begin{pmatrix}
S_1\\
\\
S_2
\end{pmatrix}
= \begin{pmatrix}
\tilde \delta_0 V+ W \delta_C^T\\
\\
\tilde \delta_1 W
\end{pmatrix},
\end{align}
where $S$ is the reshaping of $s$.

Note that a stabiliser for the chain complex $\breve{\mathcal C}_0 \rightarrow \breve{\mathcal C}_1 \rightarrow \breve{\mathcal C}_2\rightarrow \breve{\mathcal C}_3$ and the syndrome map $\breve{\delta}_1(\cdot)$ has the form $\breve{\delta}_0(m)$ for some $m \in \breve{\mathcal C}_0$. By construction, we can add any stabiliser to $(v, w)$ without violating the Syndrome Equation \eqref{eq:syndrome_equation}.
In particular,
\begin{enumerate}
\item $|(v, w)|<t$ entails that its reshaping satisfies the following properties:
\begin{enumerate}
\item Both $V$ and $W$ have at most $t$ non-zero rows. Thus all their columns have weight at most $t$.
\item Both $V$ and $W$ has at most $t$ non-zero columns. Thus all their rows have weight at most $t$.
\end{enumerate}
\item Fix a row index $i$ and a column index $j$. Let $M$ be a matrix in $\breve{\mathcal C}_0$ with columns 
\begin{align*}
M^h = \begin{cases}
V^j &\text{ for $h$ in }\supp(W \delta_C^T)_i= \supp(W_i \delta_C^T),\\
0 &\text{elsewhere}.
\end{cases}
\end{align*}
Its image $(M \delta_C^T, \,\tilde{\delta}_0 M)$ through $\breve\delta_0$ is a stabiliser. Define $V^{\star}$ and $W^{\star}$ as
\begin{align*}
V^{\star}= V + M \delta_C^T\quad\text{ and }\quad W^{\star}= W + \tilde{\delta}_0M.
\end{align*}
Observe that:
\begin{enumerate}
\item $M$ is a matrix whose non-zero columns are equal to a column of $V$. Therefore $M$ has row support contained in the row support of $V$:
\begin{align}\label{eq:row_m}
\mathrm{row}(V^{\star}) \subseteq \mathrm{row}(V).
\end{align}
\item $M$ is a matrix whose column support is $\supp(W_i\delta_C^T)$ for some row $W_i$ of $W$. Therefore $M$ has column support contained in the column support of $W$:
\begin{align}\label{eq:col_m}
\mathrm{col}(W^{\star}) \subseteq \mathrm{col}(W).
\end{align}
\end{enumerate}
\end{enumerate}
\end{claim}
\begin{lemma}[Inheritance of confinement] $\breve{\delta}_1$ has $(t, f)$-confinement where $t = \min\{d_A, d_B, d_C\}$ and $f(x)=x^3/2$.
\end{lemma}
\begin{proof}
Let $(v, w)\in \breve{\mathcal C_1}$ be such that $|(v, w)| = |(v, w)|^{\mathrm{red}} \le t$ and $s = \breve{\delta}_1(v, w)$ be its syndrome. Reshaping vectors into matrices (see Eq.~\eqref{eq:reshaping_1} and  \eqref{eq:reshaping_2}) yields the following Syndrome Equation:
\begin{align}\tag{SE}
S = \begin{pmatrix}
S_1\\
\\
S_2
\end{pmatrix}
= \begin{pmatrix}
\tilde \delta_0 V+ W \delta_C^T\\
\\
\tilde \delta_1 W
\end{pmatrix}
\end{align}
We will transform the vector $(V, W)$ by adding stabilisers to it in order to change its column and row support. We do this by iterating the following two steps.
\begin{enumerate}[label=Step \arabic*:]
\item Let $i, j$ be row and column indices s.t.
\begin{enumerate}[label=(\alph*)]
\item $(W\delta_C^T)_i \neq 0$ and $(S_1)_i = 0$;
\item $(\tilde \delta_0 V)_{ij} \neq 0$ and $(W\delta_C^T)_{ij} = 1$.
\end{enumerate}
Build a matrix $M$ as in Claim~\ref{claim:m}.\\
Transform $V$ and $W$ accordingly:
\begin{align*}
V \longmapsto V+ M\delta_C^T\\
W\longmapsto W + \tilde \delta_0 M
\end{align*}
Note that in this way we are able to delete row $i$ of $W\delta_C^T$.\\
Iterate this step until we obtain:
\begin{align}\label{eq:row_0}
\mathrm{row}(W \delta_C^T) \subseteq \mathrm{row} (S_1).
\end{align}
\item Let $i, j$ be row and column indices s.t.
\begin{enumerate}[label=(\alph*)]
\item $(W\delta_C^T)^j \neq 0$ and $(S_1)^j = 0$; 
this entails $(\tilde{\delta}_0V)^j = (W\delta_C^T)^j$;
\item $(\tilde \delta_0 V)_{ij} = (W\delta_C^T)_{ij} = 1$.
\end{enumerate}
Build a matrix $M$ as in Claim~\ref{claim:m}.\\
Transform $V$ and $W$ accordingly:
\begin{align*}
V \longmapsto V+ M\delta_C^T\\
W\longmapsto W + \tilde \delta_0 M
\end{align*}
Note that in this way we are able to delete row $i$ of $W\delta_C^T$ and by repeatedly doing so we can delete any column $j$ of $W(\delta_C^T)$ which does not belongs to the column support of $S_1$.\\
Iterate this step until we obtain:
\begin{align}\label{eq:col_0}
\mathrm{col}(W \delta_C^T) \subseteq \mathrm{col}(S_1).
\end{align}
\end{enumerate}
Let $\bf{M}$ be the matrix formed by summing over all the matrices $M$ found during these two steps. Define $V^{\star}$ and $W^{\star}$ as 
\begin{align}
V^{\star}= V + \mathbf{M} \delta_C^T\quad\text{ and }\quad W^{\star} = W + \tilde{\delta}_0\mathbf{M}.
\end{align}
We now proceed to prove an upper bound for the weight of $W^{\star}$ first and then one for the weight of $V^{\star}$. By combining these two bounds we obtain the desired confinement relation between the weight of the syndrome and the weight of the error.

\textsc{Bound on the weight of $W^{\star}$}
\begin{enumerate}
\item By Claim~\ref{claim:m}, no row of $W^{\star}$ has weight more than $t$ and therefore none of them belongs to $\ker \delta_C^T$ so that $\mathrm{row}(W^{\star}\delta_C^T) = \mathrm{row}(W^{\star})$. Combining this with Equation \eqref{eq:row_0}  yields:
\begin{align}\label{eq:row_1}
 \mathrm{row}(W^{\star}) \subseteq \mathrm{row}(S_1).
\end{align}
\item By Claim~\ref{claim:m}, the column support of $W^{\star}$ is contained in the column support of $W$ which is equal to the column support of $S_2$, by assumption on its weight. Summing these up:
\begin{align}\label{eq:col_1}
\mathrm{col}(W^{\star}) \subseteq \mathrm{col}(S_2).
\end{align}
\item Combining \eqref{eq:row_1} and \eqref{eq:col_1} yields:
\begin{align}\label{eq:last_w}
|S_1||S_2| \ge |W^{\star}|.
\end{align}
\end{enumerate}
\par
\textsc{Bound on the weight of $V^{\star}$}.
\begin{enumerate}
\item By rearranging the Syndrome Equation  \eqref{eq:syndrome_equation}, we can write $\tilde \delta_0 V^{\star} = S_1 + W^{\star}\delta_C^T$. Equations \eqref{eq:row_0} and \eqref{eq:col_0} therefore entail 
\begin{align}\label{eq:row_2}
\mathrm{row}(\tilde \delta_0V^\star) \subseteq \mathrm{row(S_1)},
\end{align}
and
\begin{align}\label{eq:col_2}
\mathrm{col}(\tilde \delta_0 V^{\star}) \subseteq \mathrm{col}(S_1).
\end{align}
\item By Claim~\ref{claim:m}, the row support of $V^{\star}$ is contained in the row support of $V$ which has cardinality at most $t$. In particular, all the columns of $V^{\star}$ have weight at most $t$ and therefore we can use the confinement property of the map $\tilde{\delta}_0$ column wise (see Lem.~\ref{lemma:confinement_1}). In other words, for each column $j$ of $V^{\star}$, the following holds:
\begin{align}
\frac{\lvert(\tilde{\delta}_0 V^{\star})^j\rvert}{4}^2\ge |(V^{\star})^{j}|.
\end{align}
Combining this with Equation \eqref{eq:row_2} yields:
\begin{align}\label{eq:row_3}
\frac{|\mathrm{row}(S_1)|}{4}^2 \ge |(V^{\star})^j|.
\end{align}
\item By Claim~\ref{claim:m}, no column of $V^{\star}$ has weight more than $t$ and therefore none of them belongs to $\ker \tilde \delta_0$ so that $\mathrm{col}(V^{\star}) = \mathrm{col}(\tilde{\delta}_0 V^{\star})$. By Equation \eqref{eq:col_2} this entails
\begin{align}
    \mathrm{col}(V^{\star}) \subseteq \mathrm{col}(S_1).
\end{align}
In other words, $V^{\star}$ has at most $|\mathrm{col}(S_1)|$ non-zero columns and combining this with Equation \eqref{eq:row_3} yields:
\begin{align}
\frac{|\mathrm{row}(S_1)|}{4}^2|\mathrm{col}(S_1) | \ge |V^{\star}|,
\end{align}
which entails:
\begin{align}\label{eq:last_v}
\frac{1}{4}|S_1|^3  \ge |V^\star|.
\end{align}
\end{enumerate}
Since $|S| = |S_1|+ |S_2|$ and $|(V, W)|= |V|+|W|$,  we can add the bounds found for $V^{\star}$ and $W^{\star}$. Observing that $(a+b)^3 \ge (a^3 + a^2b + ab)$ for integer $a, b$, we obtain that $(v^{\star}, w^{\star})$ is a vector equivalent to $(v, w)$ (i.e. it satisfies the Syndrome Equation \eqref{eq:syndrome_equation}) for which it holds:
\begin{align}
\frac{1}{4}|s|^3\ge |(v^{\star}, w^{\star})|.
\end{align}
In conclusion, since $|(v^{\star}, w^{\star})| \ge |(v,w)|=|(v, w)|^{\mathrm{red}}$, we have proved that $\breve{\delta}_1$ has confinement with respect to the function $f(x) = x^3/2$.
\end{proof}

\section{Fitting details \label{app:fitting}}
 
To obtain our threshold estimates, we use the standard critical exponent method~\cite{Harrington2004}. Specifically, in the vicinity of the threshold, we fit our data to the following ansatz
\begin{equation}
    a_0 + a_1 x + a_2 x^2,
    \label{eq:pth-ansatz}
\end{equation}
where the rescaled variable $x=(p-p_{\mathrm th})L^{1/\mu}$. Examples of this fit are shown in Fig.~\ref{fig:rescaled_fit}.

\begin{figure*}
    \centering
    \subfloat[]{
        \centering
        \includegraphics[width=0.33\linewidth]{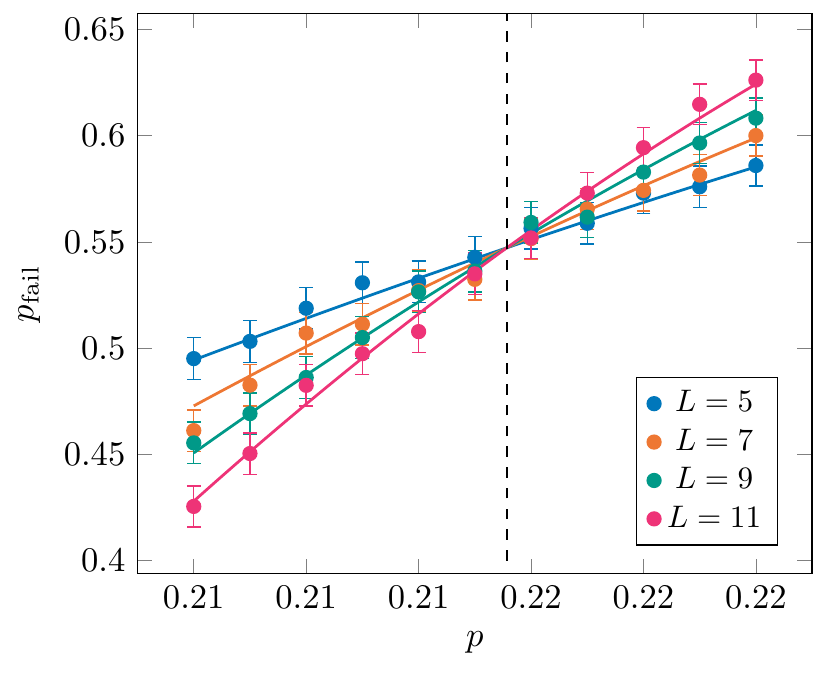}
    }
    \subfloat[]{
        \centering
        \includegraphics[width=0.33\linewidth]{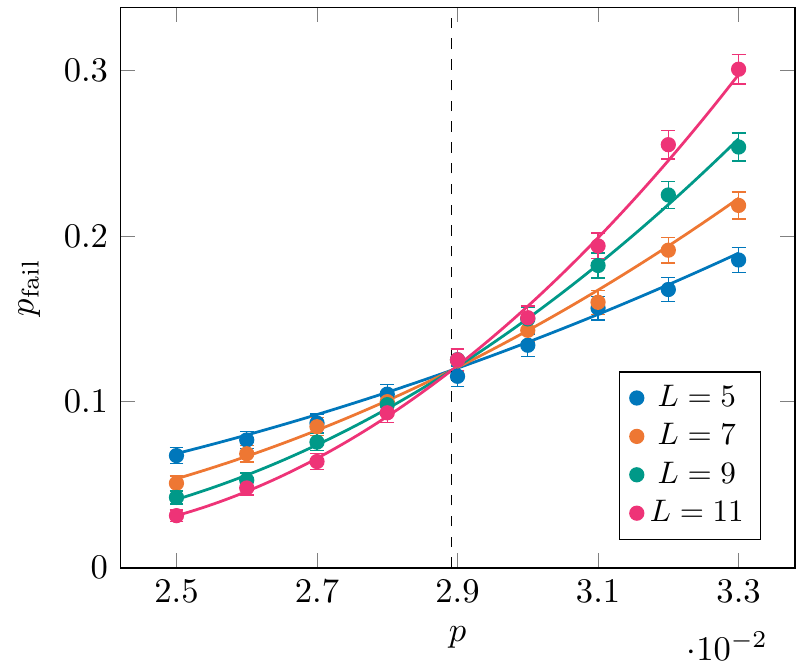}
    }
    \subfloat[]{
        \centering
        \includegraphics[width=0.33\linewidth]{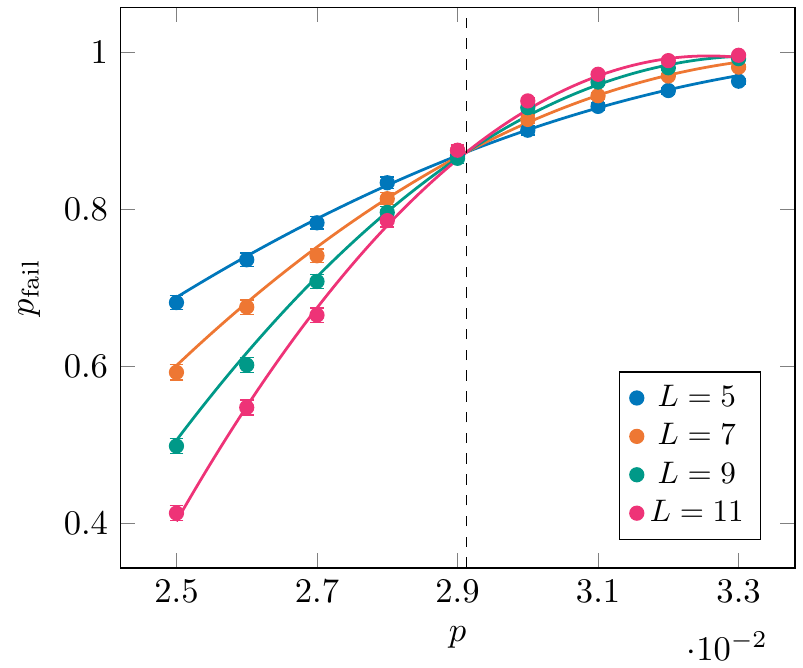}
    }
    \hfill
    \subfloat[]{
        \centering
        \includegraphics[width=0.33\linewidth]{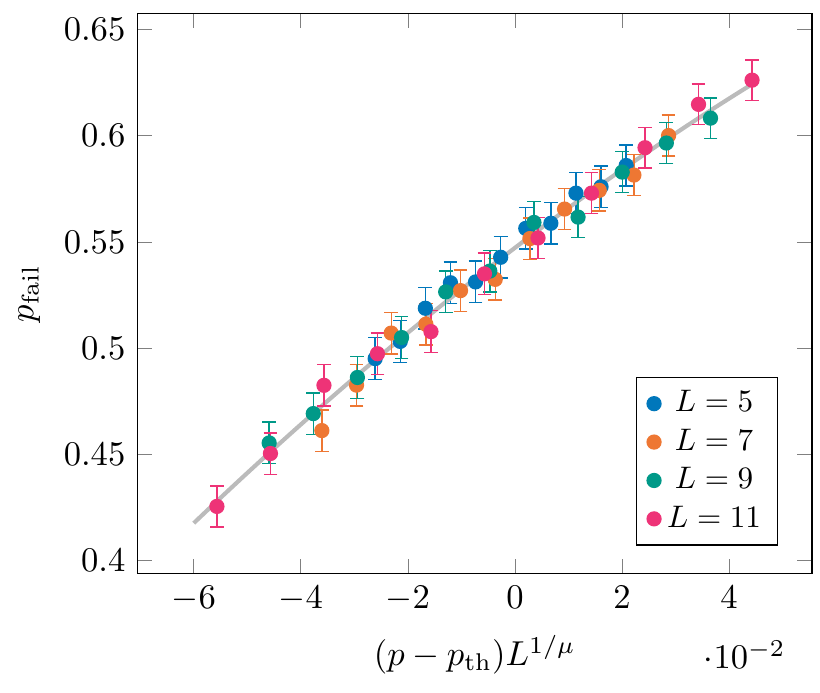}
    }
    \subfloat[]{
        \centering
        \includegraphics[width=0.33\linewidth]{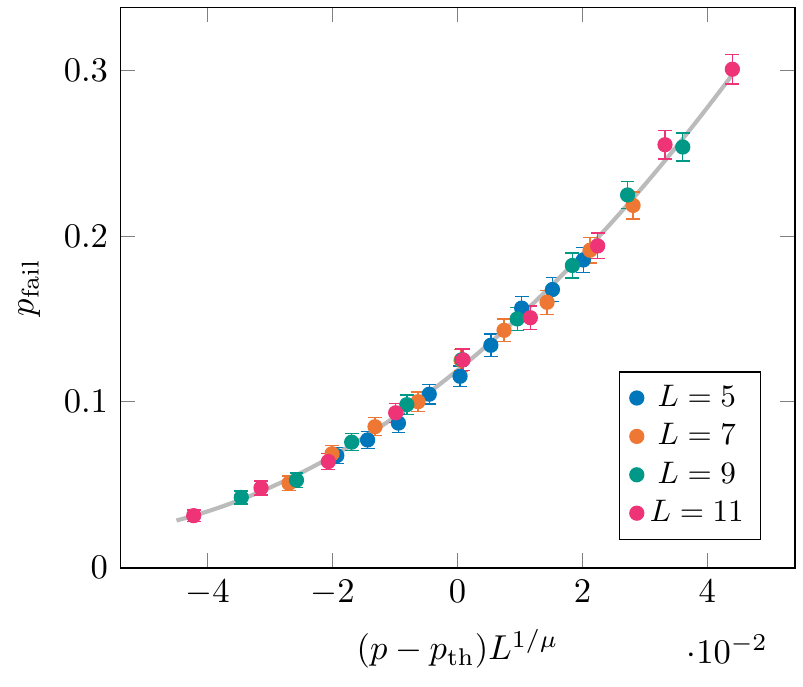}
    }
    \subfloat[]{
        \centering
        \includegraphics[width=0.33\linewidth]{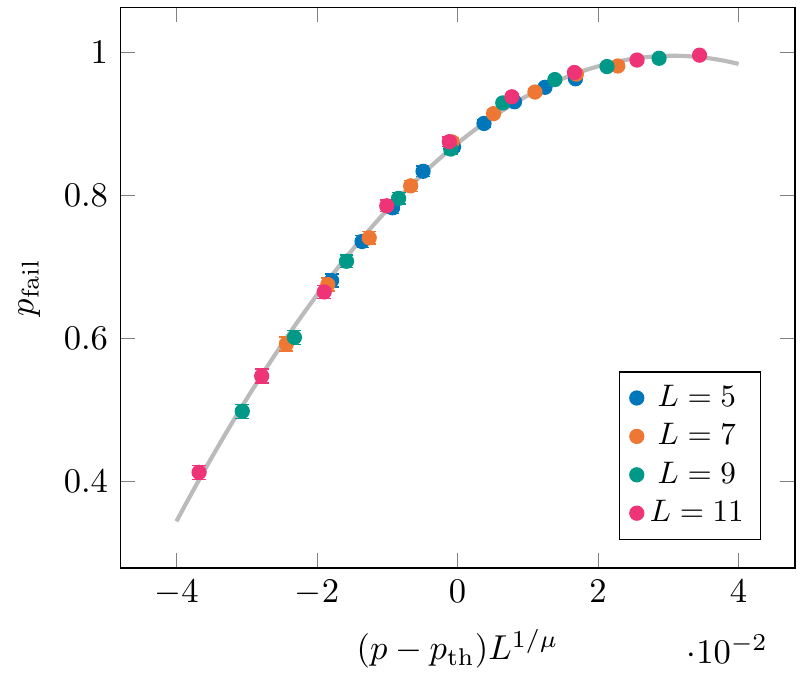}
    }
    \caption{Threshold fits for the 3D toric code using MWPM \& BP+OSD to decode. In (a), we plot the logical error rate $p_{\mathrm{fail}}$ as a function of the phase-flip error rate $p$, for values of $p$ close to the threshold. The coloured lines show the fit given by \cref{eq:pth-ansatz}, with parameters $a_0 = 0.547$, $a_1 = 1.92$, $a_2 = -4.04$, $\mu = 1.04$, and $p_{\mathrm{th}} = 0.216$ (dashed grey line). In (d), we show the same data using the rescaled variable $x=(p-p_{\mathrm{th}})L^{1/\mu}$. Subfigures (b) and (e) show equivalent data for one round of single-shot error correction, with fit parameters $a_0 = 0.119$, $a_1 = 3.04$, $a_2 = 22.9$, $\mu = 1.01$, and $p_{\mathrm{th}} = 0.0289$. Subfigures (c) and (f) show equivalent data for sixteen rounds of single-shot error correction, with fit parameters $a_0 = 0.873$, $a_1 = 7.99$, $a_2 = -130$, $\mu = 1.10$, and $p_{\mathrm{th}} = 0.0291$.   
    The error bars show the 95\% confidence intervals $p_{\rm fail} = \hat p_{\rm fail} \pm 1.96\sqrt{p_{\rm fail}(1 - p_{\rm fail})/\eta}$, where $\eta \geq 10^4$ is the number of Monte Carlo trials.}
\label{fig:rescaled_fit}
\end{figure*}

We use the fitting method described in~\cite{brown15} to understand the behaviour of the 3D toric code logical error rate for error rates $p$ significantly below threshold. Recall from Sec.~\ref{sec:numerics} that we use the following ansatz:
\begin{equation}
    p_{\mathrm{fail}}(L) \propto (p/p_{\mathrm{th}})^{\alpha L^\beta},
    \label{eq:scaling-fit-2}
\end{equation}
we take the logarithm of both sides to obtain
\begin{equation}
    \log p_{\mathrm{fail}} = \log f(L) + \alpha L^\beta \log (p/p_{\mathrm{th}}).
    \label{eq:lin_fit_1}
\end{equation}
For different values of $L$, we plot $\log p_{\mathrm{fail}}$ as a function of $\log (p/p_{\mathrm{th}})$ and fit to a straight line to obtain gradients
\begin{equation}
    g(L) = \frac{\partial \log p_{\mathrm{fail}}}{\partial u} = \alpha L^{\beta},
\end{equation}
where $u=\log (p/p_{\mathrm{th}})$. Finally, take the logarithm of both sides of the above to give
\begin{equation}
    \log g = \log \alpha + \beta \log L.
    \label{eq:lin_fit_2}
\end{equation}
We then plot $\log g$ as a function of $\log L$ and fit to a straight line to get $\alpha$ and $\beta$. Fig.~\ref{fig:scaling-fit} illustrates the above fitting procedure for code-capacity noise (no measurement errors) and for eight rounds of single-shot error correction. 

\begin{figure*}
    \centering
    \subfloat[]{
        \centering
        \includegraphics[width=0.5\linewidth]{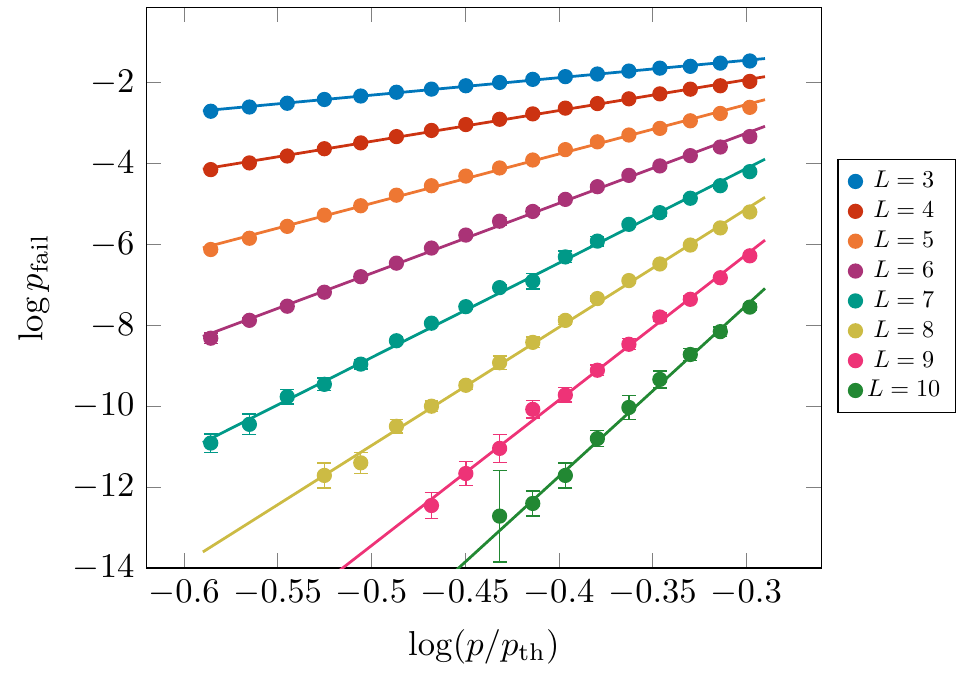}
    }
    \subfloat[]{
        \centering
        \includegraphics[width=0.43\linewidth]{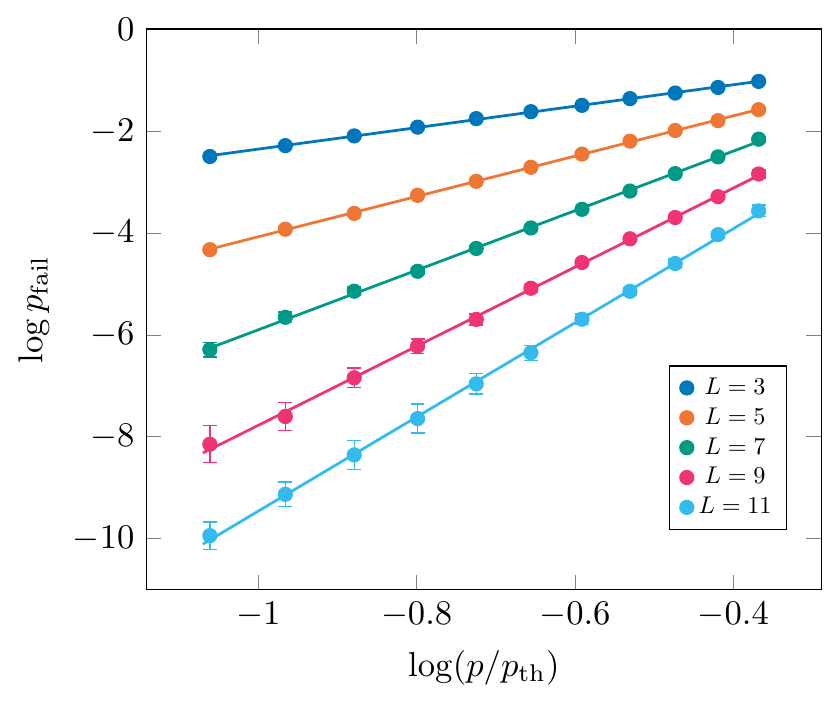}
    }
    \hfill
    \subfloat[]{
        \centering
        \includegraphics[width=0.48\linewidth]{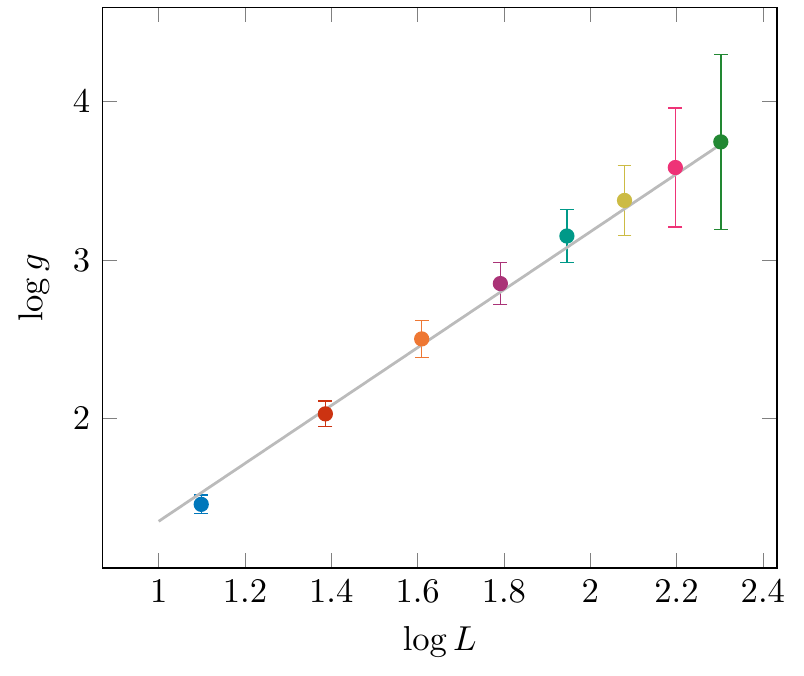}
    }
    \subfloat[]{
        \centering
        \includegraphics[width=0.48\linewidth]{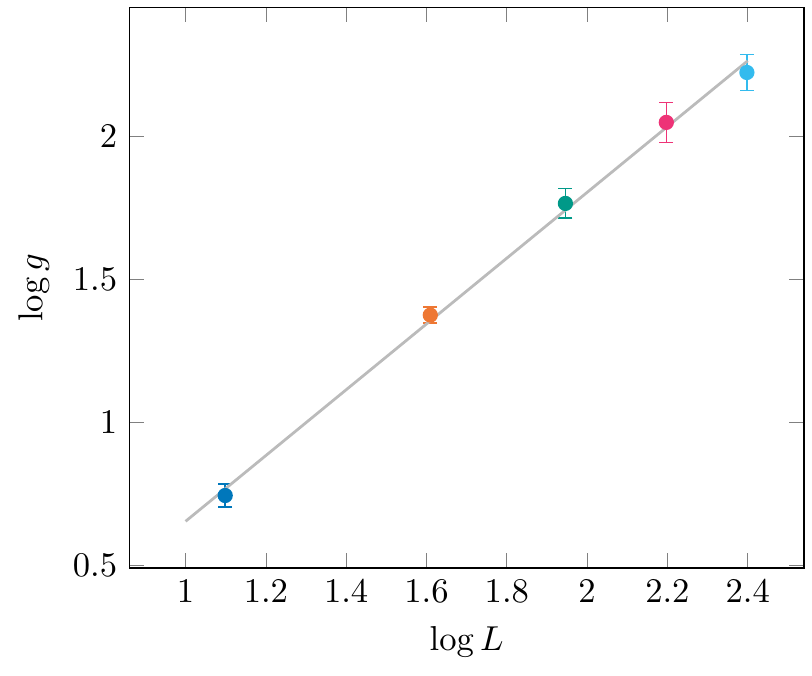}
    }
    \caption{Illustration of the fitting procedure for finding the coefficients describing the suppression of the logical error rate for phase-flip error rates substantially below threshold. (a) and (c) show data for code capacity noise (no measurement errors), and (b) and (d) show data for eight rounds of single-shot error correction. In both cases, we first plot $\log p_{\mathrm{fail}}$ as a function of $\log (p/p_{\mathrm{th}})$ for differing values of $L$, observing trends that agree with the straight line prediction of \cref{eq:lin_fit_1} [(a) and (b)]. We note that for the single-shot case there is an odd-even effect so we only include the data for odd $L$. We extract the gradients $g(L)$ from the corresponding straight line fits in (a) and (b) (grey lines), and plot the logarithms of these values against $\log L$ [(c) and (d)]. The data fit well to the linear ansatz given in \cref{eq:lin_fit_2}, which allows us to estimate the parameters $\alpha$ and $\beta$, which control the suppression of the logical error rate as per \cref{eq:scaling-fit-2}. For code capacity noise, we estimate $\alpha = 0.546(33)$ and $\beta = 1.91(3)$, and for eight rounds of single-shot error correction, we estimate $\alpha = 0.610(37)$ and $\beta = 1.15(3)$. The error bars in (a) and (b) show the 95\% confidence intervals $\log p_{\rm fail} = \log \hat p_{\rm fail} \pm \frac{1.96}{p_{\rm fail}}\sqrt{p_{\rm fail}(1 - p_{\rm fail})/\eta}$, where $\eta \geq 10^4$ is the number of Monte Carlo trials. We only include data points with at least 25 failures. The error bars in (c) and (d) show the 95\% confidence intervals given by the \textsc{LinearModelFit} function of Mathematica.
    }
    \label{fig:scaling-fit}
\end{figure*}
\clearpage
\end{document}